\documentclass[pra,twocolumn,floatfix,nofootinbib,superscriptaddress,longbibliography]{revtex4-2}
\usepackage{placeins}

\usepackage{bm,graphicx,mathrsfs,amsmath,amssymb,mathtools,makecell,bbm, amsthm,dsfont,color,times,txfonts,nicefrac,framed,enumitem,tikz,physics,wrapfig,amsfonts,tcolorbox,lipsum,nicematrix}
\usepackage[utf8]{inputenc}
\usepackage{graphicx}
\usepackage{framed}
\usepackage[colorlinks=true,linkcolor=equationcolor,citecolor=teal]{hyperref}

\definecolor{equationcolor}{RGB}{222,94,100}
\definecolor{alecolor}{RGB}{198,113,190}
\definecolor{changescolor}{rgb}{0, 0, 0.7}

\usepackage{booktabs}
\usepackage{soul}
\usepackage{amssymb}
\newtheorem{theorem}{Theorem}

\usepackage{mathtools}

\usepackage{subcaption}
\usepackage{ragged2e} 
\DeclareCaptionJustification{justified}{\justifying}

\usepackage{microtype}
\microtypecontext{spacing=nonfrench}
\microtypesetup{
protrusion={true,nocompatibility},
activate={true,nocompatibility},
tracking=true,
kerning=true,
spacing={true}
}
  
\usepackage{tikz} 
\usetikzlibrary{calc, arrows, patterns, shapes, decorations, arrows.meta, fit, positioning}
\tikzset{
    auto,node distance =1 cm and 1 cm,semithick,
    var/.style = {minimum width = 0.5cm},
    intvar/.style = {circle, draw, minimum width = 0.5cm, double},
    latent/.style = {minimum width = 0.5cm},
    point/.style = {circle, draw, inner sep=0.06cm, fill, node contents={}},
    triangle/.style = {regular polygon, regular polygon sides=3, draw, inner sep=0.06cm, fill, node contents={}},
    bidir/.style={Latex-Latex,dashed},
    dir/.style={-Latex},
    el/.style = {inner sep=2pt, align=left, sloped}
}
\tikzstyle{vertex}=[circle, fill=black!10, draw=black]
\tikzstyle{edge}=[thick]
\tikzstyle{clique}=[line width=4, draw=black!70]

\graphicspath{{imgs/}}

\newtheorem{lemma}{Lemma}

\begin{document}
\title{Using spatiotemporal Born rule for testing macroscopic realism}

\author{Naim Elias Comar} 
\email{naim.comar94@gmail.com}
\affiliation{\iip}

\author{Lucas C. Céleri}
\affiliation{\qpequi}
\affiliation{Instituto de Física de São Carlos, Universidade de São Paulo, CP 369, 13560-970, São Carlos, SP, Brazil}

\author{Mia Stamatova}
\affiliation{\clarendon}

\author{Vlatko Vedral}
\affiliation{\clarendon}

\author{Aditya Varna Iyer}
\affiliation{\clarendon}

\author{Rafael Chaves}
\affiliation{\iip}
\affiliation{\ect}

\newcommand{\iip}{International Institute of Physics, Federal University of Rio Grande do Norte, 59078-970, Natal, Brazil}
\newcommand{\qpequi}{QPequi Group, Institute of Physics, Federal University of Goiás, Goiânia, Goiás, 74.690-900, Brazil}
\newcommand{\ect}{School of Science and Technology, Federal University of Rio Grande do Norte, Natal, Brazil}
\newcommand{\clarendon}{Clarendon Laboratory, University of Oxford, Parks Road, Oxford OX1 3PU, United Kingdom}

\begin{abstract}
We show that, given an evolving quantum system and the quasiprobability distribution generated by the spatiotemporal generalization of the Born rule in pseudo-density matrices (PDMs), this distribution deviates from the sequential measurement probability distribution, given by the Lüders-von Neumann distribution, if and only if the non-signaling in time is violated; therefore, macroscopic realism can be tested using the spatiotemporal Born rule. Furthermore, we study a definition of temporal entanglement according to the structure of the PDMs that is analogous to the definition of spatial entanglement in density matrices, and we employ our results to explore the relationship between the negativity of the PDM, temporal entanglement, violation of temporal Bell inequalities, and macroscopic realism. From the examples, a clear distinction between violation of macroscopic realism and of temporal Bell inequalities is made, as well as simple examples where a clear causal witness is present while non-signaling in time and macroscopic realism are respected.
\end{abstract}

\maketitle

\section{Introduction} 
Beyond the importance given to quantum spatial correlations~\cite{Horodecki_2009_entanglement,Bell_nonlocality_RevModPhys.86.419,Uola_2020} due to their fundamental and technological importance, quantum temporal correlations~\cite{A_J_Leggett_2002,maroney2014quantumvsmacrorealism,Emary_2013,Vitagliano_2023,semi-device_independent_temporal_corr_PhysRevLett.132.220201,bounding_temporal_corr_PhysRevLett.111.020403,brukner2004quantumentanglementtime} have also started to gain significant attention due to their fundamental influence in understanding quantum causal models~\cite{Allen_PhysRevX.7.031021,chaves2015information,barrett2020quantumcausalmodels,hutter2023quantifying}, indefinite causal orders~\cite{Oreshkov_2012,Oreshkov_2016,Castro_Ruiz_2018}, quantum to classical transition~\cite{MAL20162265,Calarco_1995,Conditions_violating_MR_2008_PhysRevLett.101.090403, Kofler_CG_PhysRevLett.99.180403,bose2015uncoveringnonclassicalityschrodingercoherent,Geong_2014_PhysRevLett.112.010402,bibak2025classicallimitquantummechanics,harmonic_oscilator_CG_PhysRevLett.120.210402}, emergence of geometric structure~\cite{Fullwood_Geometry_PhysRevA.111.052438}, and attempts to quantum gravity~\cite{Hardy_2007,Oreshkov_2012}, as well as relevant in practical applications such as the certification of quantum dimensionality~\cite{Spee_2020,Wolf_PhysRevLett.102.190504,Vieira_2024}, the functioning of quantum clocks -- which is also related to the inquiry of the fundamental status of time~\cite{Erker_PhysRevX.7.031022,Burdroni_PhysRevResearch.3.033051,Woods_PRXQuantum.3.010319,rankovic2015quantumclockssynchronisation,Autonomus_temporal_prob_PhysRevX.11.011046,physicalimplicationsoffundamentalperiod_PhysRevLett.124.241301}, random access codes~\cite{wiesner_10.1145/1008908.1008920,Ambainis_10.1145/581771.581773}, quantum communication~\cite{Brierley_2015,Pironio_2003,brask2017bell,zukowski2010temporalleggettgargbellinequalitiessequential}, the simulation of stochastic processes~\cite{Garner_2017,Elliott_2018,Elliott_2020}, and the certification of causal asymmetry~\cite{Causality_assimetry_PhysRevX.8.031013}. Temporal correlations are defined in terms of sequential measurements made in the same system, that is, time-like separated measurements, whereas spatial correlations involve measurements that are space-like separated. As in the case of spatial correlations, assumptions about classical hidden variable models can be tested in sequential measurements to certify the nonclassical status of the correlations, with the difference that sequential projective measurements of incompatible observables cause disturbance (or signaling) between the measurements. In this case, the best-known realism test was proposed by Leggett and Garg~\cite{Leggett-Garg_original}, which led to the Leggett-Garg inequalities (LGIs), which test the notion of classical and realistic models called macroscopic realism (MR)~\cite{maroney2014quantumvsmacrorealism,Vitagliano_2023,Emary_2013}. Another well-known proposal for nonclassical temporal correlations is violations of temporal Bell inequalities~\cite{brukner2004quantumentanglementtime,Fritz_2010,Paz_PhysRevLett.71.3235,ringbauer2017probing}, which propose temporally nonlocal hidden-variable models.

Parallel to and related to these developments, several proposals of spatiotemporal states emerged, inspired by different motivations, including a search for a quantum conciliatory approach to relativity, which are frameworks that attempt to describe quantum states for different positions in space and time, in contrast with the usual quantum description of states, which are defined for single instants of time. These proposals involve process matrices~\cite{Oreshkov_2012,Araujo_2014,Castro_Ruiz_2018,Procopio_2015,Araujo_2015,Oreshkov_2016}, process tensors~\cite{Pollock_2018,Pollock_2018_2,Jorgensen_2019,White_2020,Milz_2019}, multiple time states~\cite{Ahanorov_1_PhysRev.134.B1410,Aharonov_2009}, superdensity operators~\cite{Cotler_2018}, consistent histories~\cite{Dowker_1995}, and quantum combs~\cite{supermaps_Chiribella_2008,Chiribella_2009,Chiribella_2013}. A framework in which temporal correlations are represented in analogy to density matrices is given by the pseudo-density matrices (PDMs) formalism, first proposed in~\cite{fitzsimons2013quantumcorrelationsimplycausation}, maintaining a symmetric approach to temporal and spatial correlations. The formalism was first proposed with the intention of indicating quantum correlations that certify causality~\cite{Ried_2015,fitzsimons2013quantumcorrelationsimplycausation}, and was used in addition to classifying spatiotemporal causal correlations~\cite{Causal_classification_of_spatiotemporal_cor_Song}, to indicate limits in quantum communication~\cite{Pisarczyk_2019}, to study the emergence of dynamics through temporal quantum teleportation~\cite{Marletto_temporal_teleportation}, to confirm forward or backward evolution of temporal correlations~\cite{Liu_arrow_of_time}, and to quantify temporal correlations through temporal mutual information~\cite{Wu_2025_mutualinformationintime}. Importantly, it is the realization of the quantum states over time (QSOT) formalism~\cite{Fullwood_2022}, which is the only formalism satisfying the spatiotemporal states desiderata proposed by~\cite{Horsman_2017} (see~\cite{quantumstatesovertimeisunique_PhysRevResearch.6.033144} for the proof) and developing the search started by Leifer and Spekkens in a proposal of conditional quantum states~\cite{Leifer_2013}.

Despite increasing work on temporal quantum correlations and in spatiotemporal formalisms, direct relations and certifications of nonclassical temporal correlations, given a spatiotemporal state, still have considerable room for development, especially compared with the wide variety of methods to certify spatial quantum correlations in density matrices~\cite{Horodecki_2009_entanglement,Bell_nonlocality_RevModPhys.86.419,Modi_RevModPhys.84.1655,Bera_2018,Uola_2020}. Important developments in this direction were given in~\cite{Causal_classification_of_spatiotemporal_cor_Song,liu2025spatialincompatibilitywitnessesquantum}, which certifies the quantum spatial incompatibility using PDMs, and in~\cite{Ku_2018} which constructs a hierarchy between the negativity of PDM, temporal steering~\cite{temporalsteering_PhysRevA.89.032112} and MR.

With the intention to contribute to new methods of certification of temporal nonclassicality in spatiotemporal formalisms, in this paper, we present a novel connection between the structure of the PDMs and the violation of MR. As our main result, we show that, in a sequential measurement scenario, the distribution given by the spatiotemporal Born rule deviates from the distribution given by the LvN rule if and only if MR is violated. The spatiotemporal Born rule, first proposed in~\cite{fullwood2025}, is a spatiotemporal analog of the Born rule of two projective measurements in a density matrix in spatially separated subsystems. In~\cite{fullwood2025}, this rule was shown to result in the spatiotemporal Margenau-Hill distribution~\cite{Margenau} for two-time-step. We use this fact to show our main result in two-time-step PDMs, and we generalize the spatiotemporal Born rule for a three-time step scenario, showing that we also obtain an analog distribution for this case, which can be separated into the Lüders-von Neumann distribution for three time steps plus a disturbance term; further, we use this to show our main result for three time steps. This three-time-step scenario allows us to relate these results to the standard LGIs, and we conjecture that such a construction may be possible for $m$-time steps. 

In addition to the method of certifying MR violation, we analyze a temporal entanglement definition on the PDM in the same way that spatial entanglement depends on the Hilbert space separability of the density matrices. This proposal aims to define temporal entanglement uniquely in terms of the spatiotemporal state (the PDM) structure, rather than being defined in terms of indirect witnesses of quantum correlations as in previous approaches~\cite{milekhin2025observablecomputableentanglementtime,Lerose_2021,Thoenniss_2025,luchnikov2024scalabletomographymanybodyquantum,2021Sonner,brukner2004quantumentanglementtime,Fritz_2010}. We discuss simple examples showing a clear hierarchy between temporal entanglement, PDM negativity and violation of temporal Bell inequalities. However, the violation of MR and the violation of temporal Bell inequalities do not respect a clear hierarchy between themselves and PDM negativity. Unlike what is done in \cite{Ku_2018}, we conclude that violation of MR does not happen only when the PDM has negativity; this is because we make a clear distinction between PDM negativity and temporal entanglement, and also between the violation of MR and temporal Bell nonlocality (which Ref.~\cite{Ku_2018} treats as equivalent). Furthermore, we highlight examples where no-signalling in time is respected, although the PDM negativity, which is a witness of causality, is present, and also present a simple example where our main result can be used to test MR in a stricter manner than using the Leggett-Garg inequality. This discussion helps show the contrast between these temporal nonclassical concepts, presenting rich distinctions as well as the classes of spatial nonclassical correlations, and highlights the practical advantage of using the spatiotemporal Born rule to certify MR.

The manuscript is organized as follows: Section~\ref{sec_background} gives a brief presentation of previous research and definitions used in our results. In Section~\ref{sec_main_results}, we present our main results, considering two-time-step and three-time-step separately. In Section~\ref{sec_types_of_temporal_nonclasscic} we discuss examples of applications of our methods and use them to contrast different types of temporal nonclassicality. Finally, in Section~\ref{sec_conclusions} we discuss the possibilities of further developments of our results and draw our conclusions.

\section{Background} \label{sec_background}

\subsection{The pseudo-density matrix (PDM)} 

Consider the density matrix of an $n$-qubit system, $\rho \in \mathcal{D}(\mathcal{H}_{n})$, the space of bounded operators acting on the Hilbert space $\mathcal{H}_{n}$ of $n$-qubits, described as the following linear sum of tensor products of Pauli operators 
\begin{equation}
    \rho = \frac{1}{2^n} \sum_{\mu_1 = 0}^3 \cdots \sum_{\mu_n =0}^3 T^{\mu_1, \cdots, \mu_n} \sigma_{\mu_1} \otimes \cdots \otimes \sigma_{\mu_n},  \nonumber
\end{equation}
with $T^{\mu_1, \cdots, \mu_n} = \langle \sigma_{\mu_1}, \cdots, \sigma_{\mu_n}\rangle := \text{Tr}[\rho (\sigma_{\mu_1} \otimes \cdots \otimes \sigma_{\mu_n})]  $, where $\sigma_{\mu_i}$ represents the Pauli operator in the $\mu_i \in\{0,1,2,3\}$ direction ($\mu_i = 0$ indicates the identity, while $\mu_1,\mu_2,$ and $\mu_3$ indicate the $x,y$ and $z$ directions) acting in the space of the $i^{\text{th}}$ qubit. To simplify the notation, we define $\mu = (\mu_1, \dots, \mu_n)$, allowing us to write the state as
\begin{eqnarray}
    \rho = \frac{1}{2^n} \sum_{\mu = 0}^{4^n -1} T^{\mu} \tilde{\sigma}_{\mu}, 
\end{eqnarray}
where $\{ \tilde{\sigma}_\mu \}_\mu$ is the set of all possible $4^n $ Pauli strings for the $n$-qubit system and $ T^{\mu} = \langle \tilde{\sigma}_\mu \rangle$.

The pseudo-density matrix framework allows the extension of the quantum state to a more general object describing the statistics of $n$ qubits in multiple time steps. To illustrate, the information on a transformation that takes the system from time $t_0$ to time $t_1$ is codified in the matrix $R_{01} \in \mathcal{D}( \mathcal{H}_n \otimes \mathcal{H}_n)$ defined as
\begin{align}
    R_{01} =  \frac{1}{2^{2n}} \sum_{\mu_0 =0}^{4^n -1} \sum_{\mu_1=0}^{4^n-1} T^{\mu_0, \mu_1}   \tilde{\sigma}_{\mu_0} \otimes \tilde{\sigma}_{\mu_1}, \label{PDM_2_steps}
\end{align}
    where $\tilde{\sigma}_{\mu_i}$ means that the Pauli string is considered in time $t_i$ and $T^{\mu_0, \mu_1} = \langle \tilde{\sigma}_{\mu_0}, \tilde{\sigma}_{\mu_1} \rangle$ means the temporal correlation of the Pauli string operator $\tilde{\sigma}_{\mu_0}$ in time $t_0$ with the Pauli string operator $\tilde{\sigma}_{\mu_1}$ in time $t_1$. The temporal correlations between two observables, $A_0$ and $B_1$, at two time steps are the statistical correlations between these observables obtained by repeatedly performing sequential projective measurements ($A_0$ at time $t_0$ and $B_1$ at time $t_1$), where $\mathcal{E}$ represents the CPTP map describing the evolution of the system between times $t_0$ and $t_1$. One can use the Lüders-von Neumann (LvN) probability distribution for sequential measurements
\begin{equation}
    P(i,j) = \text{Tr} \left[ \Pi_j^{B_1} \mathcal{E}(\Pi_i^{A_0} \rho \Pi_i^{A_0}) \right],
\end{equation}
which is the probability of obtaining sequential outcomes $a_i$ (for $A_0$) and $b_j$ (for $B_1$) for the projective measurements obtained with the projection-valued measures (PVMs) $\{ \Pi_i^{A_0} =\ket{a_i} \bra{a_i} \}$ and $\left\{ \Pi_j^{B_1} =\ket{b_j} \bra{b_j} \right\}$, where $\rho$ is the initial state of the qubits before the first measurement, and one can compute the correlation $\langle A_0, B_1 \rangle = \sum_{i,j} a_i b_j P(i,j)$. Therefore, the mtrix $R_{01}$ is a spatiotemporal matrix that represents the correlations of the observable Pauli operators of qubits separated by space and time. 

In this case, the PDM $R_{01}$ can be written in terms of the initial state $\rho$ and the Choi-Jamio{\l}kowski matrix $M_{01}$ (see Appendix \ref{CJ_Appendix} for a disambiguation between the definition of the Choi matrix and the Choi-Jamio{\l}kowski matrix used here) of the completel positive trace preserving (CPTP) map $\mathcal{E}$ describing the evolution of the system between the two consecutive measurements~\cite{Liu_2025}
\begin{equation}
    R_{01} = \frac{1}{2} \left\{ \rho\otimes\mathbb{I}_1 ,M_{01}\right\}, \label{R_2_steps}
\end{equation}
where $\mathbb{I}_i$ is the identity operator that acts on the Hilbert space of the system at time $t_i$ and the braces denote the anti-commutator of the operators: $\{A,B \} = AB+BA$, for operators $A$ and $B$.

To describe the evolution of a system measured at times $(t_0,t_1,\cdots,t_{m-1})$, the extension to a matrix $R_{01\cdots {m-1}} \in \mathcal{D} (\mathcal{H}_n^{\otimes m})$ is given by
\begin{align}
    R_{01\cdots {m-1}} =\frac{1}{2^{mn}} \sum_{\mu_0=0}^{4^n-1} \cdots \sum_{\mu_{m-1}=0}^{4^n-1} T^{\mu_0, \cdots, \mu_{m-1}}  \bigotimes_{k=0}^{m-1} \tilde{\sigma}_{\mu_k} , \label{PDM_m_time_steps}
\end{align}
where $T^{\mu_0, \cdots, \mu_{m-1}} = \langle \tilde{\sigma}_{\mu_0} \otimes \cdots \otimes \tilde{\sigma}_{\mu_{m-1}} \rangle$ is the temporal correlation between the Pauli string operators $\tilde{\sigma}_{\mu_0}, \cdots, \tilde{\sigma}_{\mu_{m-1}}$, each operator corresponding to a respective time step. These temporal correlations are also the statistical correlations obtained after the sequential projective measurements made of the respective observables in each time step. In analogy to the two-time-step case, here one can obtain the correlation with the use of the LvN probability distribution for $n$-time steps. For example, for the case of three observables $A_0, B_1$, and $C_2$ at times $t_0,t_1$, and $t_2$, the LvN probability distribution for sequential outcomes $a_i,b_j$, and $c_k$ is $P(i,j,k)=\text{Tr}\left[\Pi_k^{C_2} \mathcal{E}_2\left(\Pi_j^{B_1}\mathcal{E}_1\left(\Pi_i^{A_0}\rho\Pi_i^{A_0}\right)\Pi_j^{B_1}\right)\right]$, where $\left\{\Pi_k^{C_2}=\ket{c_k}\bra{c_k}\right\}$ is the PVM relative to $C_2$, and $\mathcal{E}_1$ and $\mathcal{E}_2$ are the CPTP maps respective to the evolutions between the first and second measurement, and second and third measurements.  

Importantly,  Ref.~\cite{Liu_2025}  also shows that $R_{01\cdots {m-1}}$ can be obtained by the recurrence relation
\begin{align}
    R_{01\cdots {m-1}} = \frac{1}{2} \{ R_{01\cdots m-2}, M_{m-2,m-1} \}, \label{R_m_steps}
\end{align}
where $M_{m-2,m-1}$ is the Choi-Jamio{\l}kowski matrix associated with the CPTP map $\mathcal{E}_{m-1}$ acting on the state from time $t_{m-2}$ to $t_{m-1}$. 


In Refs.~\cite{fitzsimons2013quantumcorrelationsimplycausation,Liu_2025}, it was shown that PDM satisfies the following properties. First, every PDM is a Hermitian operator with unit trace; however, they are not necessarily positive semi-definite.
Second, the marginals (obtained by taking a partial trace) of the PDM at some respective time steps are the same as the PDM of the remaining time steps. In other words,
\begin{equation}
    \text{Tr}_{k,\cdots,j} [R_{01\cdots {m-1}}] = R_{01\cdots {m-1} \backslash \{k,\cdots,j\}}. \label{partial_trace_PDM}
\end{equation}
In particular, when all but the Hilbert space associated with the time step $j$ are traced out, the remaining matrix is the state $\rho_j$  associated with that time step
\begin{equation}
    \text{Tr}_{0,1,\cdots, m-1 \backslash \{j\}} [R_{01\cdots m-1}] = \rho_j. \label{partial_trace_density_matrix}
\end{equation}
Third, the temporal correlation terms of the Pauli strings can be obtained with the PDM in the same way that spatial correlation terms can be obtained with density matrices:
\begin{equation}
    \langle \tilde{\sigma}_{\mu_0} , \cdots , \tilde{\sigma}_{\mu_{m-1}} \rangle = \text{Tr} [ \tilde{\sigma}_{\mu_0} \otimes \cdots \otimes \tilde{\sigma}_{\mu_{m-1}}R_{0,\cdots, m-1}].
\end{equation}
Since any operator acting on an $n$-qubit Hilbert space $\mathcal{H}_n$ can be described as a linear combination of Pauli strings, it follows that for any tensor product of operators $\mathcal{O}_i$ acting on $\mathcal{H}_n^i$ at time $t_i$ (that is, each operator $\mathcal{O}_i$ acts in a single time step), we have:
\begin{eqnarray}
   & \langle \mathcal{O}_0 \otimes \mathcal{O}_1 \otimes \cdots \otimes \mathcal{O}_{m-1} \rangle \nonumber \\
   & = \text{Tr} \left[\mathcal{O}_0 \otimes \mathcal{O}_1 \otimes \cdots \otimes \mathcal{O}_{m-1} R_{01\cdots m-1}\right]. \label{average_operators_R}
\end{eqnarray}
From now on, when there is no need to specify the number of qubits in the system, we will simply write $\mathcal{H}^i$ to represent the Hilbert space relative to the instant of time $t_i$.

In Ref.~\cite{fitzsimons2013quantumcorrelationsimplycausation}, analogous to entanglement monotones, the authors defined the so-called \textit{quantum causality monotone} $f(R)$, expressed as a function of PDMs~\cite{Wei_2003,Vidal_2000} 
\begin{equation}
    f(R) := \|R\| - 1, \label{negativity_def}
\end{equation}
where $\|R\| : = \text{Tr}\left[ \sqrt{R^\dagger R}\right]$ is the trace norm of the operator $R$ \cite{nielsen2010quantum}. This quantity measures the degree of negativity of $R$, being zero exactly for positive semi-definite $R$. Here, we will simply call $f(R)$ the negativity of $R$. In Refs.~\cite{fullwood2025,Liu_2025,fitzsimons2013quantumcorrelationsimplycausation,Zhang_2020}, it is shown that this negativity measure is a witness of correlations that cannot be established for space-like separated systems, and therefore serves as a proxy for time-like correlations and causality~\cite{fitzsimons2013quantumcorrelationsimplycausation,Causal_classification_of_spatiotemporal_cor_Song}. 

For clarity, we must say that the literature often refers to the term PDM as operators in the form of Eqs. \eqref{PDM_2_steps} or \eqref{PDM_m_time_steps}, where the coefficient terms $T^{\mu_0,\mu_1}$ and $T^{\mu_0,\cdots,\mu_{m-1}}$ are not specified to be temporal or spatial correlation terms. In Ref. \cite{Causal_classification_of_spatiotemporal_cor_Song}, a method was proposed to certify whether such correlations cannot be generated by a valid quantum evolution (CPTP map), so that they cannot respect Eqs. \eqref{R_2_steps} and \eqref{R_m_steps}, called atemporal PDMs, whereas correlations that cannot be generated by spatial correlations are called aspatial PDMs. Clearly, aspatial PDMs are the ones with nonzero causality monotone. However, the method for certifying atemporal correlations is less simple, and in this manuscript, we choose to consider only PDMs that are not atemporal, that is, PDMs that can always be generated by a valid quantum evolution (Eqs. \eqref{R_2_steps} and \eqref{R_m_steps}) since already in the definition we affirm that the coefficients represent temporal correlations.


\subsection{Temporal entanglement} \label{entanglement_in_time_subsec}

Here, we analyze an alternative definition of temporal entanglement, which was previously proposed in Refs. \cite{Jia_2023,Causal_classification_of_spatiotemporal_cor_Song,Ku_2018}. In this proposal, temporally entangled states are defined as PDMs that are not temporally separable in the same analogous manner as spatially separable density matrices are defined. We stress that this alternative definition has the feature of being determined by the structure of the spatiotemporal quantum state (the PDM). This differs from previous definitions of temporal entanglement, where it was described by means of some entanglement witness, such as the violation of a temporal Bell inequality~\cite{brukner2004quantumentanglementtime,Fritz_2010}, a non-vanishing von Neumann entanglement entropy of an alternative spatiotemporal state proposal~\cite{milekhin2025observablecomputableentanglementtime}, or a non-vanishing temporal entanglement entropy computed with the use of Feynman-Vernon influence functional~\cite{Lerose_2021,Thoenniss_2025,luchnikov2024scalabletomographymanybodyquantum,2021Sonner}.

Given a scenario with $m$ time steps and the PDM in Eq.~\eqref{R_m_steps} describing the spatiotemporal state of the system, we define the state as time-separable if and only if its PDM can be written as
\begin{equation}
    R_{01\cdots m-1} = \sum_k \mu_k \ket{\psi^0_k} \bra{\psi^0_k} \otimes \ket{\psi^1_k} \bra{\psi^1_k} \otimes \cdots \otimes \ket{\psi_k ^{m-1}} \bra{\psi_k ^{m-1}}, \label{time_separable_def}
\end{equation}
where $\mu_k \geq 0$ such that $\sum_k \mu_k = 1$ and $\ket{\psi^i_k}$ belong to the Hilbert space $\mathcal{H}^i$ with respect to time $t_i$. Therefore, a spatiotemporal state is time entangled if its PDM cannot be decomposed as in Eq.~(\ref{time_separable_def}). With this definition, we trivially conclude that a PDM $R$ such that $f(R)>0$ is temporally entangled.

\subsection{Macroscopic realism, Leggett-Garg inequalities, and no-signaling in time} \label{MR_LGIs_NSIT_sec}

The concept of macroscopic realism (MR) arises from the attempt to define the circumstances in which probability distributions describe the reality of macroscopic objects. Therefore, it is a set of conditions  useful for distinguishing quantum from classical behavior~\cite{Vitagliano_2023,Emary_2013}:
\begin{itemize}
    \item \textit{Macroscopic realism per se} (MRPS): The value of a macroscopic quantity (observable) $Q(t)$ is well defined at each time $t$.
    \item \textit{Noninvasive measurability} (NIM): It is possible, in principle, to measure the quantity $Q(t)$ with an arbitrarily small perturbation of its subsequent dynamics.
\end{itemize}
For simplicity, we refer to the joint assumption of the two conditions above as MR.

Leggett and Garg first proposed this concept~\cite{Leggett-Garg_original}, along with an inequality that the statistics of an experiment should respect to satisfy MR. Inequalities of this type are today called Leggett-Garg inequalities (LGIs)~\cite{Emary_2013,Vitagliano_2023,A_J_Leggett_2002,maroney2014quantumvsmacrorealism}. An additional assumption commonly made is the following: 
\begin{itemize}
    \item \textit{Induction} (IND): The properties of the probability distributions are determined independently of future conditions.
\end{itemize}
This condition is equivalent to the Arrow of Time (AoT) condition~\cite{A_J_Leggett_2002,Clemente_2015,Kofler_2013,Vitagliano_2023}, which is always satisfied by the statistics obtained from quantum states evolving under CPTP maps. This condition demands that the marginalization of the joint probability distributions of outcomes at different time points, with respect to a future time step, be independent of the future time step. For example, if $P_{01}(Q_0,Q_1)$ is the probability distribution of the outcomes of an observable $Q$ at times $t_0$ and $t_1$, then AoT demands $P_{01}(Q_0)=\sum_{Q_1} P_{01}(Q_0, Q_1) = P_0(Q_0)$. Notice that we write $P_{01}(Q_0)$ and $P_0(Q_0)$ with different indices because they represent, in principle, different probability distributions obtained through different experimental procedures. 

Although LGIs are necessary conditions for MR, they are not sufficient~\cite{Clemente2016}, as another condition is required to define the set of probability distributions that satisfy MR. This condition is the so-called No-Signaling in Time (NSIT)~\cite{Kofler_2013}, and requires that the probability of outcome $P_j(Q_j)$ of an outcome $Q_j$ for the observable $Q$ measured at time $t_j$ does not depend on whether a measurement was performed on the system at a previous time $t_i < t_j$. For example, if $P_{01}(Q_0,Q_1)$ is the probability distribution of the outcomes of an observable $Q$ at times $t_0$ and $t_1$, then NSIT demands $P_{01}( Q_1)=\sum_{Q_0} P_{01}(Q_0, Q_1) = P_1(Q_1)$. In Ref.~\cite{Clemente_2015}, it was shown that for three-time-step $(t_0,t_1,t_2)$, the NSIT together with the AoT are necessary and sufficient to satisfy MR. In other words, we have
\begin{theorem} [Adapted from Ref.~\cite{Clemente_2015}] \label{Theorem_Clemente}
    \begin{eqnarray}
        \text{NSIT}_{012} \land \text{AoT}_{012} \Leftrightarrow \text{MR}_{012}, \nonumber
    \end{eqnarray}
where the subindex $012$ explicitly refers to the three-time-step scenario, $(t_0,t_1,t_2)$. Note that it is also valid for the two-time-step scenario $(t_0, t_1)$.
\end{theorem}

Importantly, central to the proofs of any MR conditions is the equivalence between the MR and the existence of a joint probability distribution in time for an observable that is independent of whether measurements exist or not. Such a joint probability distribution should necessarily recover the individual probabilities when the other outcomes are marginalized. For example, if $P_{01}(Q_0,Q_1)$ is the joint probability distribution of the outcomes of an observable $Q$ at times $t_0$ and $t_1$, then MR is equivalent to verifying $P_0(Q_0)=\sum_{Q_1}P_{01}(Q_0,Q_1)$ and $P_1(Q_1)=\sum_{Q_0}P_{01}(Q_0,Q_1)$. This simple example explicitly demonstrates the relationship between AoT and NSIT to MR in two-time-step.

\subsection{Temporal Bell nolocality} 
\label{subsec_temporal_Bell}

Here we briefly review another proposal for tests of temporal non-classicality, called temporal Bell nonlocality, which is often used, and present aspects that differentiate it from the violation of macro realism; our discussion of such differences will be made in Section~\ref{sec_types_of_temporal_nonclasscic}. Temporal Bell nonlocality is certified by the violation of the so-called temporal Bell inequalities. The simplest scenario of the temporal Bell inequality is given in the following example: Suppose that Alice and Bob are experimentalists who make a sequential measurement on the same system. Alice makes her measurement at time $t_0$, and Bob makes his measurement at time $t_1$. In each round of the experiment, Alice can choose between two dichotomic observables to make a measurement, $A_1$ and $A_2$, while Bob chooses between the dichotomic observables $B_1$ and $B_2$. If one assumes (a) \emph{Realism:} The measurement results are predetermined by latent properties of the system, independent of observation, and (b) \emph{Locality in time:} The results of the second measurement (at time $t_1$) are independent of the measurement performed at time $t_0$; then one can show that the correlation function of the observables must satisfy~\cite{brukner2004quantumentanglementtime,Fritz_2010}:
\begin{equation}
    \mathcal{B} := | \langle A_1 B_1 \rangle + \langle A_1 B_2 \rangle + \langle A_2 B_1 \rangle - \langle A_2 B_2 \rangle | \leq 2. \label{time_CHSH}
\end{equation}
This inequality, first proposed in~\cite{brukner2004quantumentanglementtime}, is an example of a temporal Bell inequality and is sometimes called temporal CHSH in complete analogy to spatial CHSH~\cite{CHSH_1969PhRvL..23..880C} used to test Bell nonlocality. 

If we choose the system in consideration to be a qubit and the quantum observables $A_{1(2)} = \vec{\sigma} \cdot \vec{a}_{1(2)}$ and $B_{1(2)} = \vec{\sigma} \cdot \vec{b}_{1(2)}$, where $|\vec{a}_{1(2)}|=|\vec{b}_{1(2)}| = 1$ and $\vec{\sigma}$ is the Pauli vector, and the evolution between $t_0$ and $t_1$ to be trivial; then the quantum prediction for the inequality in~\eqref{time_CHSH} is~\cite{brukner2004quantumentanglementtime}
\begin{equation}
    \mathcal{B}_{QM} = |\vec{a}_1 \cdot (\vec{b}_1+ \vec{b}_2) + \vec{a}_2 \cdot (\vec{b}_1 - \vec{b}_2) |. \label{time_CHSH_term_quantum}
\end{equation}
The equation above is also valid for any unitary evolution, because a unitary can be embedded in the second operator choice $B_{0
(1)}$ since in the Heisenberg picture the evolution of the state is always trivial. When the choice of operators is such that 
\begin{equation}
    \vec{a}_1= \frac{1}{\sqrt{2}}(\vec{b}_1+\vec{b}_2),~~ \text{and}~~ \vec{a}_2= \frac{1}{\sqrt{2}}(\vec{b}_1-\vec{b}_2), \label{CHSH_max_violation}
\end{equation}
then the maximum violation of~\eqref{time_CHSH} is reached, with $B_{QM} = 2\sqrt{2}$~\cite{brukner2004quantumentanglementtime}. This is sometimes called the time Cirel'son bound~\cite{brukner2004quantumentanglementtime,Fritz_2010}, in analogy to the spatial case. Importantly, the equation above does not depend on the initial state of the system, given a unitary evolution; therefore, the maximum violation can be obtained for any preparation of the initial system. 

Finally, we stress that, although the assumptions (a) \textit{Realism} and (b) \textit{Locality in time} seem very similar to the MR assumptions (see the previous Subsection), they are not equivalent and are tested in different scenarios. In fact, one of the goals of this paper is to strengthen the distinction between these notions, which is done by using our main results (Theorems \ref{theorem_NSIT_D} and \ref{non-signaling_D_theorem_3steps}) in the examples of Section \ref{sec_types_of_temporal_nonclasscic}. The main difference is that in the original MR assumption \cite{Leggett-Garg_original}, which is the one used here, the realism assumption is made about one single observable $Q$ for different instants of time, together with noninvasive measurability for the real value of this observable. In the meantime, the assumptions (a) and (b) do not necessarily involve one single observable nor necessarily imply noninvasive measurability. In fact, the temporal CHSH inequality above involves tests with four observables.

\subsection{Margenau-Hill quasiprobability distribution generated by the PDM}

In Ref.~\cite{fullwood2025}, the authors show that, for two-time-step $(t_0,t_1)$, the corresponding pseudo-density matrix $R_{01}$ generates the quasiprobability distribution
\begin{eqnarray}
    \mathbf{Q}_{01} (i,j) = \text{Tr} [R_{01} (\Pi^0_i \otimes \Pi^1_j)], \label{Q_definition_R}
\end{eqnarray}
where $\{ \Pi^0_i \}_i$ and $\{ \Pi^1_j \}_j$ are the projectors in the eigenspaces of the observables $Q^0$ and $Q^1$ at times $t_0$ and $t_1$, respectively. $\mathbf{Q}_{01}(i,j)$ is referred to as the quasi-probability generated by the spatiotemporal Born rule, and it was shown (Theorem 1 in Ref.~\cite{fullwood2025}) to be equivalent to the well-known Margenau-Hill distribution~\cite{Margenau}:
\begin{eqnarray}
    \mathbf{Q}_{01} (i,j) = \frac{1}{2}\text{Tr} [\mathcal{E}(\rho \Pi^0_i + \Pi^0_i \rho) \Pi^1_j].  \label{Margenau Hill}
\end{eqnarray} 

In the same reference, they separate the $\mathbf{Q}_{01} (i,j)$ into 2 terms (this separation was first proposed by~\cite{Johansen_MH}):
\begin{equation}
    \mathbf{Q}_{01}(i,j) = \mathbf{P}_{01}(i,j) + \mathbf{D}_{01}(i,j), \label{Q_P_D_relation}
\end{equation}
where 
\begin{eqnarray}
    \mathbf{P}_{01}(i,j) = \text{Tr} [\mathcal{E}(\Pi^0_i \rho \Pi^0_i) \Pi^1_j] ,\label{Ludders_vN}
\end{eqnarray}
comes from the LvN projection postulate and is the joint probability distribution to obtain the outcomes $Q^0_i$ and $Q^1_j$ after sequential projective measurements with the channel $\mathcal{E}$ between the measurements (see Fig.~\ref{two-time_steps_fig}). Moreover,  $\mathbf{D}_{01}(i,j)$ is a measure of the disturbance caused by the projective measurement $\Pi^0_i$ on the initial state. From Proposition 1 of Ref.~\cite{fullwood2025}, one can write $\mathbf{D}_{01}(i,j)$ in the following form 
\begin{align}
    \mathbf{D}_{01}(i,j) &= \frac{1}{2} \text{Tr} [ \mathcal{E} (\rho - \rho_{\overline{i}}) \Pi^1_j],\nonumber\\ 
    &= \frac{1}{2} \left( \langle \Pi^1_j \rangle_{\mathcal{E}(\rho)} - \langle \Pi^1_j \rangle_{\mathcal{E}(\rho_{\overline{i}})} \right), \label{D_definition}
\end{align}
where $\rho_{\overline{i}} := \Pi^0_i \rho \Pi^0_i + (\mathbb{I}_0-\Pi^0_i)\rho (\mathbb{I}_0-\Pi^0_i)$.  The average $\langle X \rangle_{\mathcal{E}(\rho)}$  of the operator $X$ is computed in the evolved state $\mathcal{E}(\rho)$.

\begin{figure}
    \centering
    \includegraphics[width=1.0\linewidth]{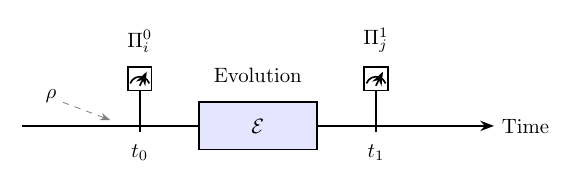}
    \caption{Two-time steps sequential measurements representing the LvN projection postulate (Eq.~\eqref{Ludders_vN}).}
    \label{two-time_steps_fig}
\end{figure}

\section{Main results} 
\label{sec_main_results}

In this section, we present our main results. All the theorems and lemmas presented in this section are original to this work, with the exception of Lemma \ref{lemma_LGI_negativity}, which was briefly claimed in Ref. \cite{fitzsimons2013quantumcorrelationsimplycausation}, but not proved. We will use these results in the next section to distinguish different kinds of temporal correlations, and we consider this distinction itself as original results. However, we claim our main results are Theorems \ref{theorem_NSIT_D} and \ref{non-signaling_D_theorem_3steps}. For the first main result, we show that the disturbance term is zero, $\mathbf{D}_{01}(i,j)=0$, if and only if the NSIT conditions -- and consequently the MR -- are satisfied for the choice of projective measurements considered. We emphasize that $\mathbf{Q}_{01}(i,j)$, $\mathbf{D}_{01}(i,j)$, and $\mathbf{P}_{01}(i,j)$ are dependent on the choice of projective measurements made (namely $\Pi_0^i$ and $\Pi_1^j$).

We begin with the simplest scenario in which measurements are performed at two different times. Next, we generalize to three-time-step, which will allow us to show a direct connection with the standard LGIs.

\subsection{Two-time steps} 
\label{two_time_steps_sec}

Consider the probability distribution $P_0(Q^0_i)$ to obtain the outcome $Q^0_i$ since the observable measured at time $t_0$ was $Q^0$. After some evolution, at time $t_1$, the observable $Q^1$ is measured, and the probability of obtaining the outcome $Q^1_j$ is  $P_1(Q^1_j)$, assuming that no measurement was performed at time $t_0$. Additionally, consider the joint probability distribution $P_{01}(Q^0_i,Q^1_j)$ (given by Eq.~\eqref{Ludders_vN} -- see Fig.~\eqref{two-time_steps_fig}) of sequentially obtaining the outcomes $Q^0_i$ and $Q^1_j$ of the observables $Q^0$ and $Q^1$ at times $t_0$ and $t_1$, respectively. The NSIT condition for this case is
\begin{equation}
    P_1(Q^1_j) = \sum_i P_{01}(Q^0_i,Q^1_j).  
\end{equation}
Furthermore, consider the system as an $n$-qubit system initialized in the state $\rho$ at time $t_0$. Assume that it evolves under a CPTP map $\mathcal{E}$ until time $t_1$. In this case, the NSIT condition above can be translated to
\begin{equation}
    \mathbf{P}_1(j) = \sum_i \mathbf{P}_{01}(i,j),  \label{NSIT_two_point} \tag{NSIT $\overline{0}$1}
\end{equation}
where the overline notation means that we are marginalizing the featured variable. For example, $\overline{0}1$ means that we are marginalizing the variables at time $t_0$. In addition, $\mathbf{P}_1(j) = \text{Tr} [\mathcal{E}(\rho)\Pi^1_j]$ is the probability of obtaining the outcome $Q^1_j$ at time $t_1$.
From now on, we use the bold letter $\mathbf{P}$ to refer to probabilities obtained by Born's rule or the LvN rule.

In the following Theorem, we prove that the NSIT condition is satisfied if and only if the disturbance $\mathbf{D}_{01}(i,j)$ term is zero.
\begin{theorem} \label{theorem_NSIT_D}
Let an $n$-qubit system in an initial state $\rho$ at time $t_0$ evolve under a CPTP map $\mathcal{E}$ until time $t_1$. Also, let $\mathbf{Q}_{01}(i,j)$ be the quasiprobability given by Eq.~\eqref{Margenau Hill}. Then, the two-time step no-signaling condition in time given by Eq.~(\ref{NSIT_two_point}) is satisfied if and only if $\mathbf{D}_{01}(i,j) = 0$ for all $i,j$. 
\end{theorem}
\begin{proof}
    From Eq.~\eqref{NSIT_two_point} we have
    \begin{align}
        \text{Tr}[\Pi^1_j \mathcal{E}(\rho)] & = \sum_i \text{Tr} [\mathcal{E}(\Pi^0_i \rho \Pi^0_i) \Pi^1_j] \nonumber \\
       & = \text{Tr} \left[\mathcal{E}\left(\sum_i \Pi^0_i \rho \Pi^0_i \right) \Pi^1_j \right]  \nonumber \\
       & =  \text{Tr} \left[\mathcal{E}\left(\Pi^0_i \rho \Pi^0_i + (\mathbb{I}_0-\Pi^0_i) \rho(\mathbb{I}_0-\Pi^0_i)   \right) \Pi^1_j \right] \nonumber \\
       & = \text{Tr} \left[\mathcal{E}\left( \rho_{\overline{i}}   \right) \Pi^1_j \right] \label{theorem_NSIT_D_step1}
    \end{align}
Now, note that, from Eq.~(\ref{D_definition}) we obtain
    \begin{eqnarray}
        & \mathbf{D}_{01}(i,j) = \frac{1}{2}\left( \text{Tr}[\Pi^1_j \mathcal{E}(\rho)] - \text{Tr} \left[ \mathcal{E}\left( \rho_{\overline{i}} \right) \Pi^1_j\right] \right),
    \end{eqnarray}
and therefore, $\mathbf{D}_{01}(i,j) = 0$ if and only if Eq.~\eqref{theorem_NSIT_D_step1} is satisfied.
\end{proof}
Therefore, a useful witness of how much NSIT is violated for a given choice of sequential measurements $\{\Pi^0_i \}_i$ and $\{\Pi^1_j \}_j$ is 
\begin{equation}
    \mathcal{N}_{01} := \sum_{i,j} |\mathbf{D}_{01}(i,j)|, \label{NSIT_witness}
\end{equation}
which is 0 if and only if NSIT is satisfied.

From Theorems~\ref{Theorem_Clemente} and~\ref{theorem_NSIT_D}, we can use the witness~\eqref{NSIT_witness} to monitor MR violation according to the choice of projective measurements (and consequently observables) made to generate $\mathbf{Q}_{01}(i,j)$ and consequently $\mathbf{D}_{01}(i,j)$. Observation can be performed experimentally on $\mathbf{Q}_{01}(i,j)$ and $\mathbf{D}_{01}(i,j)$ to confirm the violation, if predicted. Moreover, let $R_{01}$ be the PDM (which can be computed by Eq.~\eqref{R_2_steps}) related to an initial $n$-qubit state $\rho$ at time $t_0$, which evolves under a CPTP map $\mathcal{E}$ to time $t_1$. 

\subsection{Three-time steps} 
\label{3_time_steps_section}

Here we present analogous results to the two-time-step case, now considering the $R_{012}$ matrix for three-time-step. To do this, we must first propose a quasiprobability distribution for three points in time analogous to Eq.~\eqref{Margenau Hill}. For this case, we will consider the initial state $\rho$ describing a system of $n$ qubits before the first measurement, the first projective measurement, at time $t_0$, is given by the projector measurements $\{ \Pi^0_i \}_i$, the second projective measurement, at time $t_1$, is given by the projectors $\{ \Pi^1_j \}_j$, and the third projective measurement, at time $t_2$, is given by the projectors $\{ \Pi_k^2 \}_k$. The CPTP map evolving the system from $t_0$ to $t_1$ is $\mathcal{E}_1$, and from $t_1$ to $t_2$ is $\mathcal{E}_2$ (see Fig. \ref{three-time_steps_fig}). Our proposal of quasiprobability distribution is the following
\begin{eqnarray}
    \mathbf{Q}_{012}(i,j,k) = \frac{1}{4} \text{Tr} \left[ \mathcal{E}_2 (\{ \mathcal{E}_1 (\{ \rho, \Pi_i^0 \}) ,\Pi_j^1 \}) \Pi_k^2 \right], \label{Q_definition_3steps}
\end{eqnarray}
where it is important to stress that this distribution differs from the three-time-step Margenau-Hill quasiprobability distribution (see \cite{jia2026temporalkirkwooddiracquasiprobabilitydistribution} for a general definition), the justification for the quasiprobability status of this distribution will be given in Lemma \ref{3steps_quasiprobability_lemma}. The relation of this quasiprobability distribution to Born's rule applied to the PDM $R_{012}$ will be given by Theorem~\ref{theorem_R123_Q123}, and the interpretation of this quasiprobability will be given by Lemma~\ref{lemma_quasiprob_3steps_separation}.

\begin{theorem} \label{theorem_R123_Q123}
    For any initial $n$-qubits state $\rho$, projective measurements $\{\Pi_i^0\}_i,\{ \Pi_j^1 \}_j,\{ \Pi_k^2 \}_k$ made at times $t_0, t_1$ and $t_2$, respectively; and quantum channels $\mathcal{E}_1$ (acting between $t_0$ and $t_1$) and $\mathcal{\mathcal{E}}_2$ (acting between $t_1$ and $t_2$), described by the pseudo-density matrix $R_{012}$, we have the following equality
    \begin{eqnarray}
        \mathbf{Q}_{012}(i,j,k) = \text{Tr} \left[ R_{012} \Pi_i^0 \otimes \Pi_j^1 \otimes \Pi_k^2 \right], \label{theorem_R123_Q123_equation}
    \end{eqnarray}
    where $\mathbf{Q}_{012}(i,j,k)$ is given by Eq.~\eqref{Q_definition_3steps}. Moreover, $R_{012}$ is the unique operator that satisfies Eq.~\eqref{theorem_R123_Q123_equation}.
\end{theorem}
The proof of this theorem is given in Appendix~\ref{proof_of_theorem_R123_Q123_sec}. This implies that, for the three-time-step case, the generalized Born rule does not result in the Margenau-Hill distribution, unlike the two-time-step case. Furthermore, with the use of Theorem \ref{theorem_R123_Q123} and from the PDM properties, it is easy to see that the distribution given by Eq. \eqref{Q_definition_3steps} results in the corresponding quasiprobabilities of the two-time step case after marginalization (for example, $\sum_k\mathbf{Q}_{012}(i,j,k) = \mathbf{Q}_{01}(i,j)$), and is normalized to 1. Therefore, it satisfies the Kolmogorov consistency condition, and we have the following lemma.
\begin{lemma} \label{3steps_quasiprobability_lemma}
    The distribution given by Eqs. \eqref{Q_definition_3steps} and \eqref{theorem_R123_Q123_equation} satisfies the Kolmogorov consistency condition and is a quasiprobability distribution.
\end{lemma}
To obtain a clearer physical meaning for the quasiprobability $\mathbf{Q}_{012}(i,j,k)$, we proceed to obtain, for the case of three-time-step, a result similar to Eq.~\eqref{Q_P_D_relation}. To do this, notice that the LvN projector postulate for three points in time says that the joint probability of obtaining outcomes $(i,j,k)$ under the conditions described above is (see Fig.~\ref{three-time_steps_fig}) 
\begin{eqnarray}
    \mathbf{P}_{012}(i,j,k) = \text{Tr} [ \Pi_k^2 \mathcal{E}_2 ( \Pi_j^1 \mathcal{E}_1 (\Pi_i^0 \rho \Pi_i^0) \Pi_j^1 ) ]. \label{Ludders_nV_3points}
\end{eqnarray}

\begin{figure}
    \centering
    \includegraphics[width=1.0\linewidth]{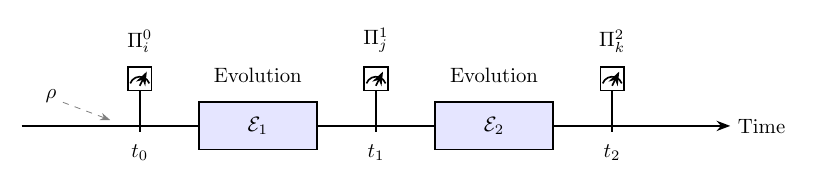}
    \caption{Three-time steps sequential measurements representing the LvN projection postulate (Eq.~\eqref{Ludders_nV_3points}).}
    \label{three-time_steps_fig}
\end{figure}

With this probability, we can state the following lemma.
\begin{lemma} \label{lemma_quasiprob_3steps_separation}
    The quasiprobability distribution for three-time-step, given by Eqs.~\eqref{Q_definition_3steps} and~\eqref{theorem_R123_Q123_equation} is related to LvN distribution (Eq.~\eqref{Ludders_nV_3points}) by the following equation
    \begin{eqnarray}
        \mathbf{Q}_{012}(i,j,k) = \mathbf{P}_{012}(i,j,k) +\mathbf{D}_{012}(i,j,k), \label{Q_P_D_relation_3points}
    \end{eqnarray}
    where 
    \begin{eqnarray}
        \mathbf{D}_{012}(i,j,k) : = & \mathbf{D}_{02,0\overline{1}2} (i,j,k) + \mathbf{D}_{2,\overline{1}2} (i,j,k) \nonumber \\
        & + \mathbf{D}_{\overline{01}2,\overline{0}2} (i,j,k) + \mathbf{D}_{12,\overline{0}12} (i,j,k),
    \end{eqnarray}
    and 
    \begin{align}
      &  \mathbf{D}_{02,0\overline{1}2} (i,j,k) := \frac{1}{2} \Tr [\mathcal{E}_2 \left(\mathcal{E}_1(\Pi_i^0\rho \Pi_i^0)-\mathcal{E}_1(\Pi_i^0\rho \Pi_i^0)_{\overline{j}} \right) \Pi_k^2], \label{D_13} \\
       & \mathbf{D}_{2,\overline{1}2} (i,j,k)  := \frac{1}{4} \Tr [\mathcal{E}_2 \left(\mathcal{E}_1(\rho)-\mathcal{E}_1(\rho)_{\overline{j}} \right)\Pi_k^2], \label{D_3_23} \\
        & \mathbf{D}_{\overline{01}2,\overline{0}2} (i,j,k) := \frac{1}{4} \Tr [\mathcal{E}_2 \left(\mathcal{E}_1(\rho_{\overline{i}})_{\overline{j}}-\mathcal{E}_1(\rho_{\overline{i}}) \right)\Pi_k^2], \label{D_3_13} \\
        & \mathbf{D}_{12,\overline{0}12} (i,j,k) := \frac{1}{2} \Tr [\mathcal{E}_2 \left( \Pi_j^1 \mathcal{E}_1(\rho)\Pi_j^1 - \Pi_j^1\mathcal{E}_1(\rho_{\overline{i}})\Pi_j^1 \right)\Pi_k^2], \label{D_23}
    \end{align}
    with $\mathcal{O}_{\overline{i}}:=\Pi_i^0 \mathcal{O} \Pi_i^0 + (\mathbb{I}_0-\Pi_i^0) \mathcal{O} (\mathbb{I}_0-\Pi_i^0)$ and $\mathcal{O}_{\overline{j}}:=\Pi_j^1 \mathcal{O} \Pi_j^1 + (\mathbb{I}_1-\Pi_j^1) \mathcal{O} (\mathbb{I}_1-\Pi_j^1)$, for any operator $\mathcal{O}$.
\end{lemma}
The proof of this lemma is presented in Appendix \ref{proof_lemma_quasiprob_3steps_separation_sec}.
With the Lemma above, we can, similarly to Eq.~\eqref{Q_P_D_relation}, understand the quasiprobability $\mathbf{Q}_{012}(i,j,k)$ as the LvN term plus a `disturbance' term $\mathbf{D}_{012}(i,j,k)$, for a choice of sequential projective measurements. It is instructive to explain the notation we chose for the subterms of the disturbance term presented above: here, we are extending the overline notation explained below Eq. \eqref{NSIT_two_point}, where each number indexing $\mathbf{D}$ means a measurement at the respective indexed time, the overline on some index means that this index is being summed over all outcomes, and the comma means that we are comparing the two probability distributions over such sequential measurement scenarios separated by the comma and finally multiplying this difference by some constant. For instance, the disturbing subterm $\mathbf{D}_{02,0\overline{1}2} (i,j,k)$ means the probability distribution $\mathbf{P}_{02}(i,k)$ of sequential measurements in $t_0$ and $t_2$, minus the sum over all outcomes of the second measurement of the probability distribution $\mathbf{P}_{012}(i,j,k)$ of sequential measurements in $t_0,~t_1$ and $t_2$, all multiplied by a factor of 1/2. 

Furthermore, in this scenario is also important to consider the quantity $\mathbf{D}_{02}(i,k)$, which is the two-time step disturbing term (Eq. \eqref{D_definition}), given by the evolution $\mathcal{E}_2 \circ \mathcal{E}_1 $ between the measurements $\{\Pi_i^0\}_i$ and $\{ \Pi_k^2 \}_k$ at times $t_0$ and $t_2$, respectivelly (that is, disconsidering a measurement made at time $t_1$). We shall relate the vanishing of each term of Eqs. \eqref{D_13}--\eqref{D_23} to the achievement of one NSIT condition for measurements in three-time-step in Lemma~\ref{non-signaling_D_lemma_3steps}, in analogy to Theorem~\ref{theorem_NSIT_D}. The complete explicit description of the NSIT and AoT conditions in three-time-step is shown in Appendix~\ref{proof_lemma_quasiprob_3steps_separation_sec}, right before the proof of the following theorem.

 \begin{theorem} \label{non-signaling_D_theorem_3steps}
    For the case of three-time-step, with the choice of projective measurements $ \{ \Pi_i^0 \}_i,~\{ \Pi_j^1 \}_j$ and $\{\Pi_k^2\}_k$, at instants $t_0,t_1$ and $t_2$, respectively, the NSIT condition (together with the AoT condition) is valid only if 
     \begin{align}
         \mathbf{D}_{012} (i,j,k)  = 0, \label{non-signaling_D_theorem_3steps_eq_1}
     \end{align}
    for any $(i,j,k)$. Additionally, if the projective measurements are rank-1 while $\mathbf{D}_{012} (i,j,k)  = 0$ and $\mathbf{D}_{02}(i,k)=0$, then NSIT is satisfied.
 \end{theorem}

Similarly to Eq.~\eqref{NSIT_witness}, we can define an NSIT violation witness for a choice of sequential projective measurements $\{ \Pi^0_i \}_i,~\{ \Pi^1_j \}_j$ and $\{ \Pi^2_k \}_k$, in the case of three-time-step: 
 \begin{equation}
     \mathcal{N}_{012}:= \sum_{i,j,k}|\mathbf{D}_{012} (i,j,k)|+\sum_{i,k}|\mathbf{D}_{02}(i,k)|. \label{NSIT_witness_3steps}
 \end{equation}
 
 According to Theorem~\ref{non-signaling_D_theorem_3steps}, the above witness is nonzero if and only if NSIT is violated for rank-1 projective measurements, and is nonzero only if NSIT is violated, in general. 

Again, the combination of Theorem~\ref{Theorem_Clemente} with Theorem~\ref{non-signaling_D_theorem_3steps}, possibly with the use of the witness $\mathcal{N}_{012}$, gives us tools to examine the violation of MR using the quasiprobability of the spatiotemporal Born rule $\mathbf{Q}_{012}(i,j,k)$, depending on the choice of projective measurements.

Moreover, this result enables us to increase the capacity for testing MR in three-time-step scenario, going beyond the standard Leggett-Garg inequalities (LGIs). We shall now define LGIs for a clear comparison presented in Sec. \ref{sec_types_of_temporal_nonclasscic} and for a relation between the PDM negativity and violation of LGIs, presented in Lemma \ref{lemma_LGI_negativity}. The standard LGI states the following: Let a quantum state be prepared and evolve from time $t_0$ to $t_2$. Between these instants of time, the same dichotomic observable $Q$, with outcomes $\pm 1$, is measured at two different instants of time among a combination of pairs of $t_0$, $t_1$ (where $t_0<t_1<t_2$), and $t_2$. Let statistics be obtained for the repetition of experiments with the same preparation and evolution, but by exchanging the measurements so that they are made in all possible combinations of pairs $t_0,~t_1$ and $t_2$. If MR is satisfied, then the correlations must respect~\cite{A_J_Leggett_2002,Leggett-Garg_original,Emary_2013,Vitagliano_2023}
\begin{equation}
    \langle Q^0 Q^1\rangle + \langle Q^1 Q^2\rangle - \langle Q^0 Q^2\rangle \leq 1, \label{Legget-Garg_ineq} \tag{LGI}
\end{equation}
where $Q^\alpha$ denotes the operator $Q$ at time $t_\alpha$, therefore, $\langle Q^\alpha Q^\beta \rangle$ is the correlation function of the same operator at different instants of time ($t_\alpha$ and $t_\beta$). From the symmetries that a dichotomic observable $Q$ must respect, there exists a class of Leggett-Garg inequalities (LGIs) other than Eq.~\eqref{Legget-Garg_ineq} that correlation functions must satisfy in the same experiment for MR to be valid~\cite{Emary_2013,Vitagliano_2023}; we call this set of inequalities the full class of LGIs. 

Given the PDM $R_{012}$ that describes a system evolving in three-time-step $t_0$, $t_1$, and $t_2$, the correlation function of a dichotomic observable $Q$ with outcomes $\pm 1$ at these instants will respect (from Eq.~\eqref{average_operators_R})
\begin{align}
   \langle Q^0 Q^1 \rangle & = \text{Tr} [R_{012} Q\otimes Q \otimes \mathbb{I}_2], \label{correlation_func_R_01} \\
   \langle Q^0 Q^2 \rangle & = \text{Tr} [R_{012} Q\otimes \mathbb{I}_1 \otimes Q], \label{correlation_func_R_02} \\
   \langle Q^1 Q^2 \rangle & = \text{Tr} [R_{012} \mathbb{I}_0 \otimes Q \otimes Q]. \label{correlation_func_R_12}
\end{align}
Equivalently, we could write $\langle Q^0 Q^1 \rangle =  \text{Tr} [R_{01} Q\otimes Q]$ (and analogously for the other correlation functions), where $R_{01} = \text{Tr}_{2}[R_{012}]$. From Eqs.~\eqref{correlation_func_R_01} --~\eqref{correlation_func_R_12}, we can use PDM to test the full class of LGIs. With this relation between PDMs and LGIs, we can prove the following lemma (the proof is presented in Appendix \ref{proof_lemma_LGI_negativity_sec}).
\begin{lemma} \label{lemma_LGI_negativity}
    If $R_{012}$ is a PDM describing a system evolving in three-time-step $t_0,$ $t_1$, and $t_2$, and there exists a dichotomic observable $Q$ that violates some inequality from the full class of LGIs, then $f(R_{012})>0$.
\end{lemma}

With the lemma above, we conclude that every positive semi-definite PDM doesn't violate LGIs for any choice of observables. One could guess that, in a similar way, positive semi-definite PDMs do not violate MR for any choice of measurements since, in this case, all quasiprobabilities generated by the spatiotemporal Born rule are actual probabilities. However, this guess is not true, as the examples of the next section will show and discuss. Furthermore, the negativity of $R_{012}$ does not necessarily imply the violation of some inequality of the full class of LGIs (again, see a simple example in the next section), and the negativity of $R_{012}$ does not imply the violation of MR (recalling that LGIs are a less conclusive test for the violation of MR than NSIT -- see Subsection~\ref{MR_LGIs_NSIT_sec}). Therefore, there is no direct relation between the violation of MR and the negativity of the PDM (see Fig. \ref{fig_sets_hierarchy}). Notwithstanding, with our main result (Theorem~\ref{non-signaling_D_theorem_3steps}), we can obtain direct tests for MR using the spatiotemporal Born rule that are stricter than the full class of LGIs, since they are necessary and sufficient for MR (see the end of Sec.~\ref{sec_types_of_temporal_nonclasscic} for a simple example).

\section{Different types of temporal non-classicality} 
\label{sec_types_of_temporal_nonclasscic}

Here, we contrast different notions of temporal non-classicality as we believe they represent different aspects of quantum correlations in the same way that spatial non-classical correlations are classified according to different facets -- quantum entanglement~\cite{Horodecki_2009_entanglement}, Bell nonlocality~\cite{Bell_nonlocality_RevModPhys.86.419}, quantum discord~\cite{Discord_Ollivier_PhysRevLett.88.017901}, and others~\cite{Modi_RevModPhys.84.1655,Bera_2018}. We shall analyze the already defined (see Section~\ref{sec_background}) violations of macroscopic realism (MR), the negativity of the PDM (which points to non-classical causality), temporal entanglement, and the violation of temporal Bell inequalities (temporal Bell nonlocality) for a set of examples, highlighting their main differences.

\subsection{Discussion for two-time-step}

We start our discussion of contrasting types of temporal non-classicality in the simplest two-time-step scenario of Section~\ref{two_time_steps_sec}. First, note that, as in the spatial entanglement case, temporal entanglement is necessary for the violation of the temporal CHSH inequality. This comes from the fact that the definition of temporally separable PDMs is identical to the definition of separable density matrices, together with the fact that the temporal CHSH inequality (Eq. \eqref{time_CHSH}) is completely analogous to the spatial case; therefore, the analogous conclusion that separable states don't violate this inequality holds.

In the following, we show a simple example of PDM, namely $R_{01}^{(1)}$, where $f(R_{01}^{(1)})>0$, presents temporal entanglement, maximally violates a temporal CHSH, but it does not violate NSIT for any choice of measurement and therefore satisfies macro realism. This PDM represents the case where the system is a qubit with a maximally mixed initial state $\rho = \mathbb{I}/2$, and the CPTP map for the evolution between $t_0$ and $t_1$ is the trivial evolution $\mathcal{C} = \mathcal{I}$. For this case, the Choi-Jamio{\l}kowski matrix of the chnnel is equal to the SWAP operator, and from Eq.~\eqref{R_2_steps}, we have
\begin{eqnarray}
    R_{01}^{(1)} = \frac{1}{2}\begin{pmatrix}
        1 & 0 & 0 & 0  \\
        0 & 0 & 1 & 0  \\
        0 & 1 & 0 & 0  \\
        0 & 0 & 0 & 1
        \end{pmatrix}. \label{two_time_PDM_example1}
\end{eqnarray}
This PDM is one of the maximally entangled pseudo-density operators, which is a basis of spatiotemporal states analogous to the Bell basis~\cite{Zhao_Geometryofquantumcorrelations_PhysRevA.98.052312}; these spatiotemporal states are also an indispensable resource for teleportation in time~\cite{Marletto_temporal_teleportation}. This matrix has eigenvalues $\{ -1/2,~1/2,~1/2,~1/2 \}$, therefore $f(R_{01}^{(1)})>0$ and also $R_{01}$ is temporally entangled. This PDM represents the situation analyzed in the temporal CHSH inequality~\eqref{time_CHSH}, and we can use Eq.~\eqref{average_operators_R} to obtain the proper correlation terms of Eq.~\eqref{time_CHSH_term_quantum} using the PDM according to $\langle A_i B_j \rangle = \text{Tr} [R_{01}^{(1)} A_i \otimes B_j]$. With the correct choice of operators to be measured (given in Eq.~\eqref{CHSH_max_violation}), the temporal CHSH inequality can be maximally violated. However, since the initial state is $\mathbb{I}/2$, we have $\mathbf{D}_{01}(i,j) = 0$ for any choice of projective measurements (this is a direct consequence of Eq.~\eqref{D_definition}); hence, from Theorem~\ref{theorem_NSIT_D}, there is no violation of NSIT for any measurement choice. 

The fact that $R_{01}^{(1)}$ can maximally violate a time Bell inequality while satisfying NSIT highlights a clear difference between the concepts of temporal Bell nonlocality and violation of MR. Bell inequalities and MR both suppose latent variables that predetermine the value of observables, independent of the disturbance caused by measurements. This is equivalent to having a joint probability distribution for the outcomes of all possible observables involved in the experiments, due to Fine's theorem~\cite{Fine_PhysRevLett.48.291}. The main difference is in the fact that, while in the standard MR scenario we have only one choice of observable for each instant of time~\cite{A_J_Leggett_2002,Leggett-Garg_original,Clemente_2015,Clemente2016,Bera_2018,Emary_2013,maroney2014quantumvsmacrorealism,Vitagliano_2023}, in the time Bell inequality scenario we have different choices of observables (possibly not compatible) at each instant of time~\cite{brukner2004quantumentanglementtime,Fritz_2010,Emary_2013}. This more comprehensive choice of observables at each instant imposes stronger constraints on the existence of the joint probability distribution for the outcomes of all possible observables. For example, for this case, the NSIT condition implies the existence of a joint probability distribution $P_{01}(a,b)$ (where $a(b)$ is the outcome of the observable $A(B)$ to be measured at instant $t_{0(1)}$), while the assumptions to obtain the time Bell inequality implies the existence of $P_{01}(a_1,a_2,b_1,b_2)$ (where $a_{1(2)}$ and $b_{1(2)}$ are the outcomes of $A_{1(2)}$ and $B_{1(2)}$). The existence of this last probability distribution is more demanding when we consider that the observable choices at the same time instants are not compatible; therefore, in this example, one kind of classicality is allowed, but not the other. Interestingly, if one wishes to show the violation of MR in this scenario (that is, maximally mixed state at the beggining and trivial evolution), one need to consider a third measurement, at time $t'$, such that $t_0 < t_1$ (see an example below); however, this three measurement scenario is no longer described by the two-time-step PDM $R_{01}^{(1)}$. 

Importantly, the PDM $R_{01}^{(1)}$ encompasses all correlations in the two-time measruement scenarios described above, since it does not depend on the choice of observables to be measured, and the main difference between spatial and temporal correlations is manifest in the negativity of $R_{01}^{(1)}$, which is a witness of causality. Furthermore, the distinction between signalling and causality is apparent in this example, since NSIT is not violated for any measurement choice while $f(R_{01}^{(1)})>0$.

To exemplify a situation in which the PDM has nonvanishing negativity, presents temporal entanglement but does not violate the temporal CHSH, we have the following case. Consider an initial state that is maximally mixed $\rho = \mathbb{I}/2$ and evolving under a depolarizing channel:
\begin{equation}
    \mathcal{E}_\eta (\cdot) = (1-\eta)(\cdot) + \text{Tr}(\cdot)\frac{\eta}{2} ~\mathbb{I}_2, \label{depolariaing_channel}
\end{equation}
where $\eta \in [0,1]$ is the depolarizing parameter.
Using Eq.~\eqref{R_2_steps}, we obtain the following PDM 
\begin{equation}
    R_{01_\eta} = \frac{1}{2} \begin{pmatrix}
        1- \frac{\eta}{2} & 0 & 0 & 0 \\
        0 &  \frac{\eta}{2} & 1-\eta & 0 \\
        0 & 1-\eta &  \frac{\eta}{2} & 0 \\
        0 & 0 & 0 & 1 - \frac{\eta}{2}
    \end{pmatrix}. \label{depolarizing_PDM}
\end{equation}
 To obtain the violations of Eq.~\eqref{time_CHSH}, we again write $A_{1(2)} = \vec{\sigma}\cdot \vec{a}_{1(2)}$ and $B_{1(2)} = \vec{\sigma}\cdot \vec{b}_{1(2)}$, with $|\vec{a}_{1(2)}| = |\vec{b}_{1(2)}| = 1$, and rewrite the correlations as $\langle A_i B_j \rangle = \vec{a}_i^\top T \vec{b}_j$, where $T_{ij} : = \text{Tr} [R^{\text{hc}} \cdot \sigma_i^0 \otimes \sigma_j ^t]$ (with $\sigma_{i}^{0}$ representing the Pauli operator $i$ acting at time $0$, $\sigma_{j}^{t}$ representing the Pauli operator $j$ acting at time $t$, and $i,j \in \{x,y,z\}$). This choice of representation gives the following optimal value for the Bell parameter $\mathcal{B}$:
\begin{equation}
    \mathcal{B}_{\text{max}} = 2\sqrt{\lambda_1 + \lambda_2}, \label{optimal_B_CHSH}
\end{equation}
where $\lambda_1$ and $\lambda_2$ are the two largest eigenvalues of the matrix $T^\top T$. This method is used in~\cite{stamatova2025complexheatcapacitywitness}, and the proof of the equation above is completely equivalent to the well-known Horodecki criterion~\cite{Horodecki_criterion}. 

With Eq.~\eqref{optimal_B_CHSH} in hand, we can discuss interesting features of this case. Again, Theorem~\ref{theorem_NSIT_D} states that NSIT is always satisfied, as a consequence of Eq.~\eqref{D_definition} and since the initial state is maximally mixed. From Eq.~\eqref{optimal_B_CHSH}, we see that the maximal violation of the temporal CHSH depends on the depolarizing parameter $\eta$: $\mathcal{B}_{\text{max}} = 2\sqrt{2}(1-\eta)$, and therefore the violation occurs only for $\eta < 1 - \sqrt{2}/2$. Moreover, we can compute the negativity of $R_{01_\eta}$ and obtain $$f(R_{01_\eta}) = \begin{cases}
    \frac{1}{2}(2- 3\eta) & \text{if } \eta \in [0,2/3), \\
    0 & \text{if } \eta \in [2/3,1].
\end{cases}$$
Hence, in the interval $\eta \in [0, 1-\sqrt{2}/2)$, this PDM presents the same temporal nonclassicality as the PDM of Eq.~\eqref{two_time_PDM_example1}; that is, it respects NSIT for any measurement, violates the temporal CHSH, and has nonvanishing negativity. In the interval $\eta \in [1-\sqrt{2}/2,2/3)$, it has nonvanishing negativity and, therefore, temporal entanglement, but it does not violate the temporal CHSH, while respecting NSIT. Finally, in the interval $\eta \in [2/3,1]$, there is no violation of NSIT and of temporal CHSH, the negativity vanishes, and there is no temporal entanglement (this can be verified with the use of the positive partial transpose (PPT) criterion \cite{PPT_HORODECKI19961} since, in this interval, the PDM is a density matrix); therefore, in this interval, the PDM presents no kind of temporal nonclassicality that we are considering. Importantly, note that, since the dephasing channel $\mathcal{E}_\eta$ does not modify a maximally mixed state, the density matrix of the system does not change after the action of $\mathcal{E}_\eta$; however, the PDM depends on the channel parameter $\eta$. This occurs because the PDM carries information not only on the density matrices at different time instants, but also on the channel responsible for the evolution. Consequently, the channel determines how evolution influences measurement effects and, therefore, determines the presence of temporal nonclassicality.

An important result we present in two-time-step in the temporal Bell scenario is the following theorem.

\begin{theorem} \label{temoral_CHSH_positivity_theorem}
    Given a positive semi-definite PDM in two-time-step $R_{01}$ describing the temporal correlations of a qubit evolving between times $t_0$ and $t_1$, the CHSH inequality (Eq. \eqref{time_CHSH}) is never violated.
\end{theorem}

The proof of this theorem is given in Appendix \ref{Appendix_details_temporal_nonclassic}. This theorem answers the question if the negativity of the PDM is necessary for temporal Bell inequalities violations for qubits. The previous examples agreed with the statement that positive semi-definite PDMs do not violate CHSH, and the next examples will make explicit the difference between MR and temporal Bell inequalities by showing that positive semi-definite PDMs can indeed violate MR.

Another important example considering the depolarizing channel is the following: consider the initial state $\rho = \text{diag}(0.6,0.4)$, with the depolarizing channel with $\eta=0.8$ (see Eq. \eqref{depolariaing_channel}) causing the evolution. For this case, the PDM will be
\begin{equation}
    R_{01} = \begin{pmatrix}
        0.36 & 0 & 0 & 0 \\
        0 & 0.24 & 0.10 & 0 \\
        0 & 0.10 & 0.16 & 0 \\
        0 & 0 & 0 & 0.24
    \end{pmatrix},
\end{equation}
which is a positive semi-definite operator. If we choose the first projective measurment to be $\Pi_1^0 = \ket{+}\bra{+}$ and $\Pi_{-1}^0 =\ket{-}\bra{-}$, while the second projective measurement to be $\Pi^1_1 =\ket{0}\bra{0}$ and $\Pi^1_{-1} =\ket{1}\bra{1}$, we can compute the LvN distribution (Eq. \eqref{Ludders_vN}) and the spatiotemporal quasiprobability (Eq. \eqref{Q_definition_R}) for this scenario to obtain the disturbing terms. This results in $|\mathbf{D}_{01}(1,1)| =|\mathbf{D}_{01}(-1,1)| = |\mathbf{D}_{01}(1,-1)| = |\mathbf{D}_{01}(-1,-1)|=0.01$ or $\mathcal{N}_{01}=0.04$, atesting the violation of MR. Since this is a positive semi-definite operator, one can also use the PPT criterion \cite{PPT_HORODECKI19961} to certify that this is a temporally separable PDM.

This last example shows that we can have positive semi-definite PDMs, and also temporally separable PDMs, that violate MR, depending on the choice of measurements. This means that, although here all quasiprobability distributions obtained by the spatiotemporal Born rule are non-negative and marginalize to the correct probability distribution for each time step, while respecting LGIs (see Lemma \ref{lemma_LGI_negativity}), they don't describe the correct probability distributions of the outcomes of sequential measurement scenarios. Importantly, this seems to contradict the hierarchy found in Ref. \cite{Ku_2018}, where positive-semidefinite PDMs never violate MR; the discussion of this contradiction is made in the last Section (see Sec. \ref{sec_conclusions}) of this paper. Similar violations of MR for positive semi-definite PDMs can be spotted for a range of parameters of initial diagonal density matrices $\rho = \text{diag}(p,1-p)$, for $p>0.5$, and with the depolarizing channel evolution (Eq. \eqref{depolariaing_channel}) with $\eta \gtrsim 0.65$. Moreover, in the three-time-step scenario, one can obtain positive semi-definite PDMs that violate NSIT and MR, with initial diagonal density matrices and two sequential depolarizing channels with similar parameters from the presented above.

From Theorem \ref{temoral_CHSH_positivity_theorem}, the PDM of the example above cannot generate a violation of the CHSH inequality. In Appendix~\ref{Appendix_details_temporal_nonclassic}, Subsection~\ref{subsection_violationMR}, we present a case where, for a range of parameters, the CHSH is not violated, yet MR is. In these cases, there are no latent variables independent of the disturbance caused by the measurements; however, the temporal CHSH is not violated. This supports the distinction between the concepts of MR and temporal CHSH. In Fig.~\ref{fig_sets_hierarchy}, we present the relation between the different notions of temporal nonclassicalities that we investigate in this manuscript, for two-time-step.

\begin{figure}
    \centering
    \includegraphics[width=1.0\linewidth]{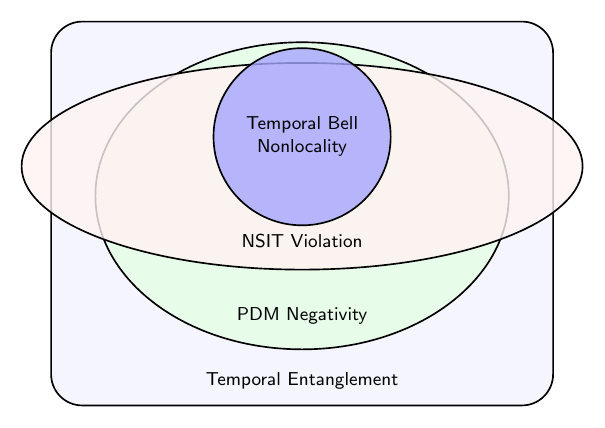}
    \caption{Relation between different notions of temporal correlations for one evolving qubit in two-time-step PDMs.}
    \label{fig_sets_hierarchy}
\end{figure}

For a final example, consider the case studied in~\cite{stamatova2025complexheatcapacitywitness}, where it is shown that the PDM elements can be directly identified with the imaginary part of the dynamical heat capacity. In particular, the authors consider an initial thermal state and probe the system through a weak periodic perturbation, which allows the definition of the complex heat capacity $C(\omega) = \frac{\partial \langle H \rangle_\omega}{\partial T}$, where the expectation value is taken in the steady state induced by a sinusoidal temperature modulation of frequency $\omega$. Crucially, they show that the dissipative part of the response function can be used to construct the off-diagonal terms of the two-time PDM associated with the thermal state and its driven dynamics, and that the appearance of an imaginary component in the heat capacity is equivalent to the presence of negativity of the PDM. 

This has immediate relevance for our framework. As stated earlier, any negativity of the PDM $R_{01}$ certifies temporal entanglement, and by Theorem \ref{temoral_CHSH_positivity_theorem}, such negativity is necessary for violations of the temporal CHSH depending on the measurement choice. Since Ref.~\cite{stamatova2025complexheatcapacitywitness} demonstrates that 
\begin{equation*}
    \text{Im}(C(\omega)) = 0 \Leftrightarrow R_{01} \succeq 0,
\end{equation*}
the imaginary heat capacity becomes an operationally accessible thermodynamic witness of PDM negativity, while it can certify the impossibility of CHSH violation if it vanishes. Furthermore,
\begin{equation*}
    \text{Im}(C(\omega)) \neq 0 \Leftrightarrow f(R_{01}) > 0,
\end{equation*}
implies that the system exhibits temporal entanglement.

This example strengthens the conceptual link between non-equilibrium quantum thermodynamics and the temporal correlation hierarchy developed in this work. Temporal entanglement is not merely a mathematical artifact but a physically measurable feature. At the same time, because Im$(C(\omega))$ can be non-zero even when NSIT holds for all projective measurements (as happens for the maximally mixed initial state discussed earlier), this example further illustrates that thermodynamic witnesses of temporal non-classicality are strictly stronger than NSIT violations, and therefore probe a distinct facet of spatiotemporal quantum correlations.

\subsection{Discussion for three-time-step}
\label{different_temporal_nonclass_three-time_sec}

We analyze the simple situation for three-time-step where the system is a qubit and the system is initially in a maximally mixed state $\rho = \mathbb{I}_2/2$ that evolves with two trivial evolutions $\mathcal{E}_1 = \mathcal{E}_2 = \mathcal{I}$. Using Eq.~\eqref{R_m_steps} we obtain
\begin{equation}
       R_{012}^{(1)} = \left(
\begin{array}{cccccccc}
 \frac{1}{2} & 0 & 0 & 0 & 0 & 0 & 0 & 0 \\
 0 & 0 & \frac{1}{4} & 0 & \frac{1}{4} & 0 & 0 & 0 \\
 0 & \frac{1}{4} & 0 & 0 & \frac{1}{4} & 0 & 0 & 0 \\
 0 & 0 & 0 & 0 & 0 & \frac{1}{4} & \frac{1}{4} & 0 \\
 0 & \frac{1}{4} & \frac{1}{4} & 0 & 0 & 0 & 0 & 0 \\
 0 & 0 & 0 & \frac{1}{4} & 0 & 0 & \frac{1}{4} & 0 \\
 0 & 0 & 0 & \frac{1}{4} & 0 & \frac{1}{4} & 0 & 0 \\
 0 & 0 & 0 & 0 & 0 & 0 & 0 & \frac{1}{2} \\
\end{array}
\right). \label{3_time_steps_identity_PDM}
\end{equation}
Interestingly, if we consinsider sequential projective measurements of the observables $\sigma_z$ at time $t_0$, of $\sigma_x$ at time $t_1$, and of $\sigma_z$ at time $t_2$, we have that the NSIT violation witness defined in Eq. \eqref{NSIT_witness_3steps} is $\mathcal{N}_{012} = 1$, which implies violation of MR, from Theorems \ref{Theorem_Clemente} and \ref{non-signaling_D_theorem_3steps}. This shows that in a scenario where the initial state (maximally mixed state) and evolution (trivial evolution) make it impossible to witness NSIT and MR violation in two-time-step (see the example of the PDM $R_{01}^{(0)}$ above), we can violate NSIT and MR in an analogous three-time step scenario where we add an incompatible measurement. This comes from the fact that NSIT in two-time-step is necessary but not sufficient to certify MR, while three-time-step NSIT is equivalent to MR \cite{Kofler_2013,Clemente_2015,maroney2014quantumvsmacrorealism}.

We can compare the scenario that the PDM above describes to a well-known scenario where a qubit violates the LGI. Let a qubit with generic initial state $\rho$ evolve under the unitary $U = e^{-i H t}$ with the Hamiltonian given by $H = \frac{1}{2} \Omega \sigma_x,$ where $\Omega >0$. Let the observable measured in pairs for the three time steps be $Q = \sigma_z$; then, it can be shown (see Refs.~\cite{Fritz_2010,Emary_2013,Vitagliano_2023}) that the correlation term of~\eqref{Legget-Garg_ineq} is
\begin{align}
    \mathcal{K}_3 : = & \langle Q^0 Q^1 \rangle + \langle Q^1 Q^2 \rangle - \langle Q^0 Q^2 \rangle \nonumber \\
    = & 2 \cos(\Omega t) - \cos (2 \Omega t). \label{LGI_bound_example}
\end{align}
Note that the term above does not depend on the initial state of the system; therefore, it is valid for $\rho = \mathbb{I}_2/2$. If we consider $t=0$ in the equation above, the evolution is trivial, and we are in the same scenario for the PDM $R_{012}$ from Eq.~\eqref{3_time_steps_identity_PDM}; for this case, we have $\mathcal{K}_3 = 1$ and no violation of the LGI. However, this PDM has negativity $f(R_{012})=2$, showing that the negativity of the PDM is not sufficient for the violation of LGI, although necessary (Corollary~\ref{lemma_LGI_negativity}). Additionally, we can compute the NSIT violation witness from Eq.~\eqref{NSIT_witness_3steps}, using the spatiotemporal Born rule (Eq.~\eqref{Q_definition_R}) and the LvN rule (Eq.~\eqref{Ludders_nV_3points}), for the initial state $\rho=\mathbb{I}_2/2$ and unitary evolution described above with the same projective measurement of the observable $\sigma_z$ in the three time steps, and obtain 
\begin{equation}
    \mathcal{N}_{012} = \sin^2(\Omega t).
\end{equation}
Notice that from Theorems \ref{Theorem_Clemente} and \ref{non-signaling_D_theorem_3steps}, this witness vanishes if and only if NSIT and MR are satisfied. For $t=0$, we are again in the situation described by the PDM $R_{012}$, but in this case, MR is satisfied due to the choice of three $\sigma_z$ sequential measurements. This confirms that the negativity of the PDM is not sufficient for the violation of the MR. Moreover, if we consider the entire cycle of unitary evolution, we have the plot shown in Fig.~\ref{LGI_NSIT_example_plot}. In this figure, we present the behaviour of the $\mathcal{K}_3$ correlations term of the LGI together with the NSIT violation witness $\mathcal{N}_{012}$ as a function of $\Omega t$ from $0$ to $2\pi$. This clearly shows that, using Theorem~\ref{non-signaling_D_theorem_3steps} and the NSIT violation witness, we visualize a broader region of MR violation (noticing that the green dashed line represents $\mathcal{N}_{012} = 0$, which is equivalent to NSIT) than the regions confirmed by the LGI violation (Eq.~\eqref{LGI_bound_example} and~\eqref{Legget-Garg_ineq}), which is the orange shaded area in the plot, separated by the black dashed line, which is the LGI saturation.

\begin{figure}
    \centering
    \includegraphics[width=1.0\linewidth]{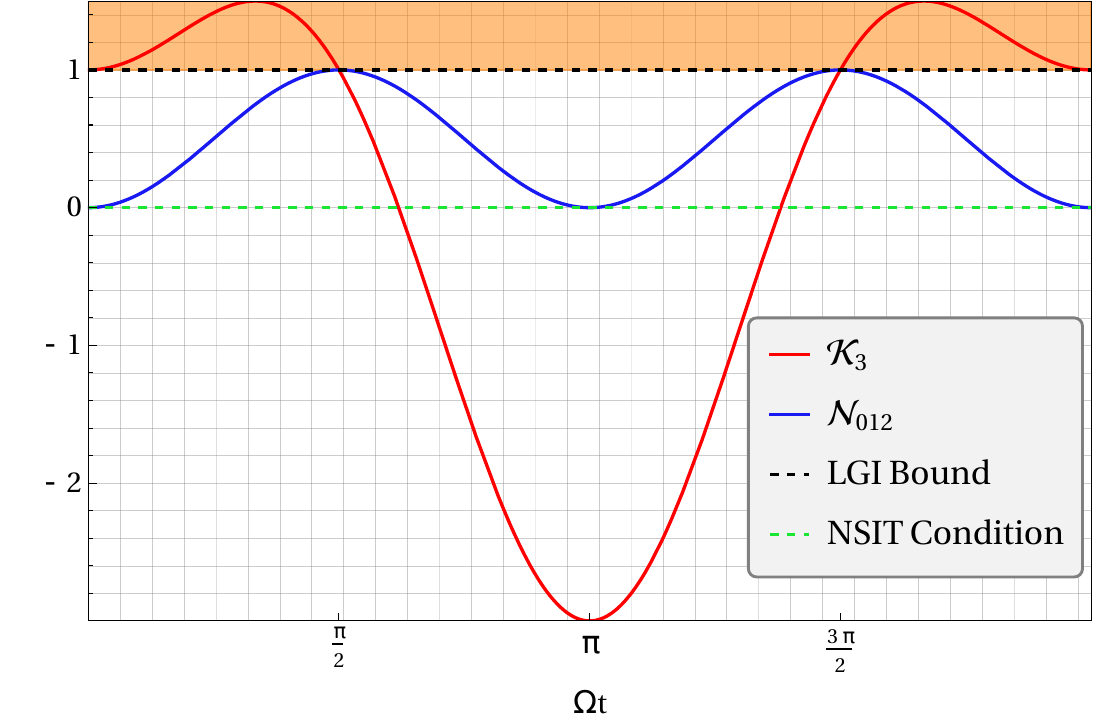}
    \caption{LGI correlations term $\mathcal{K}_3$, NSIT violation witness $\mathcal{N}_{012},$ LGI bound ($\mathcal{K}_3=1$, with the orange shaded area representing the region of LGI violation), and the NSIT condition ($\mathcal{N}_{012 }=0$), as a function of $\Omega t$.}
    \label{LGI_NSIT_example_plot}
\end{figure}

\section{Discussion and Conclusions} \label{sec_conclusions}

Our main results reveal a direct relation between the nonvanishing of the disturbance term of the temporal quasiprobability distribution and the violation of macroscopic realism, namely Theorems~\ref{theorem_NSIT_D} and~\ref{non-signaling_D_theorem_3steps}. In contrast to the two-time-step case, in the three-time-step case, the quasiprobability distribution does not coincide with the Margenau-Hill distribution, and the disturbance term is divided into the sum of smaller terms, where the vanishing of each of the independent terms corresponds to one of the NSIT conditions for three-time-step; the vanishing of each of them individually is equivalent to NSIT; furthermore, we used the internal relation between the disturbance subterms to show that the vanishing of the full disturbance term, together with the vanishing of the analogous distribution for two-time-step, for a scenario where the second measurement is not made, implies NSIT. Remarkably, this disturbance term is the deviation of the distribution proposed in Eq.~\eqref{Q_definition_3steps} from the LvN probability (see Lemma~\ref{lemma_quasiprob_3steps_separation}), which can be obtained by a generalized Born's rule in the corresponding PDM (Theorem~\ref{theorem_R123_Q123}), in complete analogy with the two-time step result from Ref.~\cite{fullwood2025}.

We enumerate the fact that we showed that the three-time-step quasiprobability from Eq. \eqref{Q_definition_3steps} comes naturally from the extension of the spatiotemporal Born rule (Theorem \ref{theorem_R123_Q123}), can be separated into a LvN probability distribution part plus a disturbing part (Lemma \ref{lemma_quasiprob_3steps_separation}), and has a clear connection with MR and NSIT (given by Theorem \ref{non-signaling_D_theorem_3steps}); therefore, we believe this quasiprobability carries with it important information about temporal correlations of sequential measurements, which deserves further resarch. This also justifies our interest in this quasiprobability distribution rather than the Margenau-Hill distribution for three-time-step, and raises the question of why fundamentally these quasiprobabilities don't match for the three-time-step case, as opposed to the two-time-step case. In general, more research could study the relation of this quasiprobability distribution to the Kirkwood-Dirac distribution~\cite{Arvidsson_Shukur_2024,Lostaglio2023kirkwooddirac}. For example, this can be explored using the temporal distribution framework developed in~\cite{jia2026temporalkirkwooddiracquasiprobabilitydistribution}, which relates the temporal Kirkwood-Dirac distribution to PDMs and other proposals of spatiotemporal states.

Due to the successful generalization of the spatiotemporal Born rule on testing NSIT and MR for the three-time-step case, we believe that the same pattern happens to distributions representing more steps in time, in which we can properly define disturbance terms related to NSIT violations. Furthermore, the extension of the spatiotemporal Born rule and its relation to MR when the measurements considered are not projective could be a possible extension for future research. Above all, these results certify the classical status of a spatiotemporal Born rule that is identical to the sequential measurement probability, since in this case, the MR is satisfied. Therefore, we believe that this can lead to the study of this distribution behavior in physical situations where the transition from quantum to classical is explored, possibly using the NSIT violation witness proposals to certify nonclassicality. For example, further research can study the application of Thorems~\ref{theorem_NSIT_D} and~\ref{non-signaling_D_lemma_3steps} (and the MR violation witnesss) in systems restrained by coarse-grained measurements, considered in~\cite{Kofler_CG_PhysRevLett.99.180403,Conditions_violating_MR_2008_PhysRevLett.101.090403,Geong_2014_PhysRevLett.112.010402,harmonic_oscilator_CG_PhysRevLett.120.210402,bibak2025classicallimitquantummechanics,Frowis_2018,bose2015uncoveringnonclassicalityschrodingercoherent}. 

We related a proposal of temporal entanglement based on the separability of the PDM, with the clear consequence that every negative PDM is temporally entangled. The relation between temporal entanglement and temporal nonlocality is exactly as spatial nonlocality relates to spatial entanglement, with the addendum that the PDM can be negative. To complete the picture examined of temporal correlations, we show the necessity of negativity of the PDM to violate the temporal CHSH (Theorem \ref{temoral_CHSH_positivity_theorem}), and we believe our discussion in a set of simple examples helps mainly to distinguish the concepts of violations of MR and temporal Bell nonlocalities, which has been a long-standing debate~\cite{Emary_2013,Vitagliano_2023,maroney2014quantumvsmacrorealism,Fritz_2010}, as their certification could easily be contrasted with the use of Theorem~\ref{theorem_NSIT_D}. 

In this set of examples, we found examples of MR violations even when the PDM is positive semi-definite; this means that, although a probability distribution is always possible to construct for every measurement, the disturbing effects will deviate it from the correct LvN distribution for sequential measurements. This is in contradiction with the results of \cite{Ku_2018} and comes from the fact that this reference uses the assumptions of MR and of temporal Bell locality as synonymous and shows that time-separable PDMs do not violate temporal Bell inequalities. Furthermore, they also assume that this proof implies no violation of temporal Bell inequalities for positive semi-definite PDMs. However, we have a distinction between temporally entangled PDMs and negative PDMs: every negative PDM is temporally entangled, but not vice versa. However, our Theorem \ref{temoral_CHSH_positivity_theorem} is a proof that even temporally entangled positive semi-definite PDMs do not violate temporal CHSH.

For future work, we can extend this comparison of temporal correlations by considering~\cite {Ku_2018}, which relates temporal steering to PDMs. We could also explore extensions to sequential measurements that are not projective and dichotomic, which could, in turn, reveal different relations among the nonclassical temporal correlations and would entail the need for other QSOT applications, like the ones proposed in \cite{PRL_Lie}. Additionally, interesting new possibilities for temporal quasiprobability definitions could be explored for the analogues of QSOT proposals in operational probabilistic theories \cite{goswami2026multitimestatesquantumoperational} and tests of macroscopic realism in generalized proababilistic theories with more robust definitions of macroscopic realism \cite{Schmid_2024}. Furthermore, the relationship between PDMs' negativity, MR, temporal Bell nonlocality, and undefined causal order can be explored in further research, in particular, investigating the relation between PDMs and quantum supermaps \cite{supermaps_Chiribella_2008,Wilson_2026}, as well as searching for the relation between these types of temporal nonclassicalities and the axiomatic proposal of a temporal correlation monotone based on the Jordan product of quantum channels \cite{Guo_PRA}; the relationship between PDMs' negativity and MR promises applications in quantum thermodynamics and nonequilibrium physics via the relationship between PDMs and the imaginary part of heat capacity~\cite{stamatova2025complexheatcapacitywitness}. 

Finally, we show that the method for certifying MR violation with the use of Theorem~\ref{non-signaling_D_theorem_3steps} is more efficient than the violation of the Leggett-Garg inequality for a well-known case, indicating that this method can be utilized to certify MR violation in new scenarios. This shows a relatively clear and simple operational way to verify MR in a necessary and sufficient manner with the use of the spatiotemporal Born rule. This could be useful, for instance, to certify nonclassicality in quantum clock systems~\cite{Burdroni_PhysRevResearch.3.033051,Woods_PRXQuantum.3.010319}, where the presence of nonclassicality was related to the violation of LGIs.

\begin{acknowledgments}
The authors acknowledge helpful discussions with Leandro Morais, Santiago Zamora, Bárbara Amaral, Alisson Tezzin,  and George Moreno. The authors especially acknowledge the anonymous Physical Review A referees, Seok Hyung Lie, and James Fullwood, since their important corrections, criticism, and suggestions improved the quality of this paper substantially. In particular, James Fullwood presented us with the example of a positive semi-definite PDM that violates MR. N.E.C. and R.C. acknowledge the support of the Financiadora de Estudos e Projetos (Grant No. 1699/24 IIF-FINEP), the Fundação de Amparo e Promoção da Ciência, Tecnologia e Inovação do RN (FAPERN, Grant No. 06/2026), the Simons Foundation (Grant No. 1023171, R.C.), and the Brazilian National Council for Scientific and Technological Development (CNPq, Grants No.403181/2024-0 and 301687/2025-0). LCC acknowledges support from CNPq through grant 308065/2022-0, the National Institute of Science and Technology for Applied Quantum Computing through CNPq grant 408884/2024-0, FAPEG through grant 202510267001843, and FAPESP through grant 2025/23726-4.
\end{acknowledgments}

\bibliography{biblio}

@article{Guo_PRA,
  title = {Temporal correlating power of quantum channels},
  author = {Guo, Yi and Luo, Shunlong},
  journal = {Phys. Rev. A},
  volume = {111},
  issue = {6},
  pages = {062415},
  numpages = {8},
  year = {2025},
  month = {Jun},
  publisher = {American Physical Society},
  doi = {10.1103/ms8v-qnj8},
  url = {https://link.aps.org/doi/10.1103/ms8v-qnj8}
}

@article{PRL_Lie,
  title = {Multipartite Quantum States over Time from Two Fundamental Assumptions},
  author = {Lie, Seok Hyung and Fullwood, James},
  journal = {Phys. Rev. Lett.},
  volume = {135},
  issue = {23},
  pages = {230204},
  numpages = {8},
  year = {2025},
  month = {Dec},
  publisher = {American Physical Society},
  doi = {10.1103/lbf3-snp8},
  url = {https://link.aps.org/doi/10.1103/lbf3-snp8}
}

@article{Wilson_2026,
   title={Quantum Supermaps are Characterized by Locality},
   volume={10},
   ISSN={2521-327X},
   url={http://dx.doi.org/10.22331/q-2026-03-09-2013},
   DOI={10.22331/q-2026-03-09-2013},
   journal={Quantum},
   publisher={Verein zur Forderung des Open Access Publizierens in den Quantenwissenschaften},
   author={Wilson, Matt and Chiribella, Giulio and Kissinger, Aleks},
   year={2026},
   month=Mar, pages={2013} }

@article{Schmid_2024,
   title={A review and reformulation of macroscopic realism: resolving its deficiencies using the framework of generalized probabilistic theories},
   volume={8},
   ISSN={2521-327X},
   url={http://dx.doi.org/10.22331/q-2024-01-03-1217},
   DOI={10.22331/q-2024-01-03-1217},
   journal={Quantum},
   publisher={Verein zur Forderung des Open Access Publizierens in den Quantenwissenschaften},
   author={Schmid, David},
   year={2024},
   month=Jan, pages={1217} }

@misc{goswami2026multitimestatesquantumoperational,
      title={Multi-Time States in Quantum and Operational Probabilistic Theories: An Axiomatic Framework}, 
      author={Kaumudibikash Goswami and Archana Barvar and Manabendra Nath Bera},
      year={2026},
      eprint={2610.22489},
      archivePrefix={arXiv},
      primaryClass={quant-ph},
      url={https://arxiv.org/abs/2610.22489}, 
}

@ARTICLE{Horodecki_2003,
       author = {{Horodecki}, Michael and {Shor}, Peter W. and {Ruskai}, Mary Beth},
        title = "{Entanglement Breaking Channels}",
      journal = {Reviews in Mathematical Physics},
         year = 2003,
        month = jan,
       volume = {15},
       number = {6},
        pages = {629-641},
          doi = {10.1142/S0129055X03001709},
archivePrefix = {arXiv},
       eprint = {quant-ph/0302031},
 primaryClass = {quant-ph},
       adsurl = {https://ui.adsabs.harvard.edu/abs/2003RvMaP..15..629H}
}

@article{Ruskai_2003,
   title={Qubit Entanglement Breaking Channels},
   volume={15},
   ISSN={1793-6659},
   url={http://dx.doi.org/10.1142/S0129055X03001710},
   DOI={10.1142/s0129055x03001710},
   number={06},
   journal={Reviews in Mathematical Physics},
   publisher={World Scientific Pub Co Pte Lt},
   author={Ruskai, Mary Beth},
   year={2003},
   month=Aug, pages={643–662} }

@book{nielsen2010quantum,
  title={Quantum Computation and Quantum Information: 10th Anniversary Edition},
  author={Nielsen, M.A. and Chuang, I.L.},
  isbn={9781139495486},
  url={https://books.google.com.br/books?id=-s4DEy7o-a0C},
  year={2010},
  publisher={Cambridge University Press}
}

@article{fullwood2025,
  title = {Quasiprobabilistic nature of the spatiotemporal {Born} rule},
  author = {Fullwood, James and Ma, Zhihao and Wu, Zhen},
  journal = {Physical Review A},
  volume = {113},
  number = {3},
  pages = {032204},
  numpages = {10},
  year = {2026},
  month = {Mar},
  publisher = {American Physical Society},
  doi = {10.1103/cw25-pqx1},
  url = {https://link.aps.org/doi/10.1103/cw25-pqx1}
}

@article{temporalsteering_PhysRevA.89.032112,
  title = {Temporal steering inequality},
  author = {Chen, Yueh-Nan and Li, Che-Ming and Lambert, Neill and Chen, Shin-Liang and Ota, Yukihiro and Chen, Guang-Yin and Nori, Franco},
  journal = {Physical Review A},
  volume = {89},
  number = {3},
  pages = {032112},
  numpages = {7},
  year = {2014},
  month = {Mar},
  publisher = {American Physical Society},
  doi = {10.1103/PhysRevA.89.032112},
  url = {https://link.aps.org/doi/10.1103/PhysRevA.89.032112}
}

@article{Ku_2018,
   title={Hierarchy in temporal quantum correlations},
   volume={98},
   issn={2469-9934},
   url={http://dx.doi.org/10.1103/PhysRevA.98.022104},
   doi={10.1103/physreva.98.022104},
   number={2},
   journal={Physical Review A},
   publisher={American Physical Society},
   author={Ku, Huan-Yu and Chen, Shin-Liang and Lambert, Neill and Chen, Yueh-Nan and Nori, Franco},
   year={2018},
   month=aug }

@misc{liu2025spatialincompatibilitywitnessesquantum,
      title={Spatial incompatibility witnesses for quantum temporal correlations}, 
      author={Liu, Xiangjing and Verma, Harshit and Xiao, Yunlong and Dahlsten, Oscar and Gu, Mile},
      year={2025},
      eprint={2511.01179},
      archivePrefix={arXiv},
      primaryClass={quant-ph},
      url={https://arxiv.org/abs/2511.01179}, 
}

@article{Liu_arrow_of_time,
  title = {Inferring the arrow of time in quantum spatiotemporal correlations},
  author = {Liu, Xiangjing and Chen, Qian and Dahlsten, Oscar},
  journal = {Physical Review A},
  volume = {109},
  number = {3},
  pages = {032219},
  numpages = {8},
  year = {2024},
  month = {Mar},
  publisher = {American Physical Society},
  doi = {10.1103/PhysRevA.109.032219},
  url = {https://link.aps.org/doi/10.1103/PhysRevA.109.032219}
}

@article{ringbauer2017probing,
  title={Probing the non-classicality of temporal correlations},
  author={Ringbauer, Martin and Chaves, Rafael},
  journal={Quantum},
  volume={1},
  pages={35},
  year={2017},
  publisher={Verein zur F{\"o}rderung des Open Access Publizierens in den Quantenwissenschaften}
}

@article{hutter2023quantifying,
  title={Quantifying quantum causal influences},
  author={Hutter, Lucas and Chaves, Rafael and Nery, Ranieri Vieira and Moreno, George and Brod, Daniel Jost},
  journal={Physical Review A},
  volume={108},
  number={2},
  pages={022222},
  year={2023},
  publisher={American Physical Society}
}

@article{brask2017bell,
  title={{Bell} scenarios with communication},
  author={Brask, Jonatan Bohr and Chaves, Rafael},
  journal={Journal of Physics A: Mathematical and Theoretical},
  volume={50},
  number={9},
  pages={094001},
  year={2017},
  publisher={IOP Publishing}
}

@article{chaves2015information,
  title={Information-theoretic implications of quantum causal structures},
  author={Chaves, Rafael and Majenz, Christian and Gross, David},
  journal={Nature Communications},
  volume={6},
  number={1},
  pages={5766},
  year={2015},
  publisher={Nature Publishing Group UK London}
}

@article{Uola_2020,
   title={Quantum steering},
   volume={92},
   issn={1539-0756},
   url={http://dx.doi.org/10.1103/RevModPhys.92.015001},
   doi={10.1103/revmodphys.92.015001},
   number={1},
   journal={Reviews of Modern Physics},
   publisher={American Physical Society},
   author={Uola, Roope and Costa, Ana C. S. and Nguyen, H. Chau and Gühne, Otfried},
   year={2020},
   month=mar }

@article{quantumstatesovertimeisunique_PhysRevResearch.6.033144,
  title = {Quantum state over time is unique},
  author = {Lie, Seok Hyung and Ng, Nelly H. Y.},
  journal = {Physical Review Research},
  volume = {6},
  number = {3},
  pages = {033144},
  numpages = {6},
  year = {2024},
  month = {Aug},
  publisher = {American Physical Society},
  doi = {10.1103/PhysRevResearch.6.033144},
  url = {https://link.aps.org/doi/10.1103/PhysRevResearch.6.033144}
}

@article{Fullwood_2022,
   title={On quantum states over time},
   volume={478},
   issn={1471-2946},
   url={http://dx.doi.org/10.1098/rspa.2022.0104},
   doi={10.1098/rspa.2022.0104},
   number={2264},
   journal={Proceedings of the Royal Society A: Mathematical, Physical and Engineering Sciences},
   publisher={The Royal Society},
   author={Fullwood, James and Parzygnat, Arthur J.},
   year={2022},
   month=aug }

@article{Leifer_2013,
   title={Towards a formulation of quantum theory as a causally neutral theory of {Bayesian} inference},
   volume={88},
   issn={1094-1622},
   url={http://dx.doi.org/10.1103/PhysRevA.88.052130},
   doi={10.1103/physreva.88.052130},
   number={5},
   journal={Physical Review A},
   publisher={American Physical Society},
   author={Leifer, M. S. and Spekkens, Robert W.},
   year={2013},
   month=nov }

@article{Horsman_2017,
   title={Can a quantum state over time resemble a quantum state at a single time?},
   volume={473},
   issn={1471-2946},
   url={http://dx.doi.org/10.1098/rspa.2017.0395},
   doi={10.1098/rspa.2017.0395},
   number={2205},
   journal={Proceedings of the Royal Society A: Mathematical, Physical and Engineering Sciences},
   publisher={The Royal Society},
   author={Horsman, Dominic and Heunen, Chris and Pusey, Matthew F. and Barrett, Jonathan and Spekkens, Robert W.},
   year={2017},
   month=sep, pages={20170395} }

@article{Pisarczyk_2019,
   title={Causal limit on quantum communication},
   volume={123},
   issn={1079-7114},
   url={http://dx.doi.org/10.1103/PhysRevLett.123.150502},
   doi={10.1103/physrevlett.123.150502},
   number={15},
   journal={Physical Review Letters},
   publisher={American Physical Society},
   author={Pisarczyk, Robert and Zhao, Zhikuan and Ouyang, Yingkai and Vedral, Vlatko and Fitzsimons, Joseph F.},
   year={2019},
   month=oct }

@article{Ried_2015,
   title={A quantum advantage for inferring causal structure},
   volume={11},
   issn={1745-2481},
   url={http://dx.doi.org/10.1038/nphys3266},
   doi={10.1038/nphys3266},
   number={5},
   journal={Nature Physics},
   publisher={Springer Science and Business Media LLC},
   author={Ried, Katja and Agnew, Megan and Vermeyden, Lydia and Janzing, Dominik and Spekkens, Robert W. and Resch, Kevin J.},
   year={2015},
   month=mar, pages={414--420} }

@article{supermaps_Chiribella_2008,
   title={Transforming quantum operations: quantum supermaps},
   volume={83},
   issn={1286-4854},
   url={http://dx.doi.org/10.1209/0295-5075/83/30004},
   doi={10.1209/0295-5075/83/30004},
   number={3},
   journal={Europhysics Letters},
   publisher={IOP Publishing},
   author={Chiribella, G. and D'Ariano, G. M. and Perinotti, P.},
   year={2008},
   month=jul, pages={30004} }

@article{Chiribella_2009,
   title={Theoretical framework for quantum networks},
   volume={80},
   issn={1094-1622},
   url={http://dx.doi.org/10.1103/PhysRevA.80.022339},
   doi={10.1103/physreva.80.022339},
   number={2},
   journal={Physical Review A},
   publisher={American Physical Society},
   author={Chiribella, Giulio and D'Ariano, Giacomo Mauro and Perinotti, Paolo},
   year={2009},
   month=aug }

@article{Chiribella_2013,
   title={Quantum computations without definite causal structure},
   volume={88},
   issn={1094-1622},
   url={http://dx.doi.org/10.1103/PhysRevA.88.022318},
   doi={10.1103/physreva.88.022318},
   number={2},
   journal={Physical Review A},
   publisher={American Physical Society},
   author={Chiribella, Giulio and D'Ariano, Giacomo Mauro and Perinotti, Paolo and Valiron, Benoit},
   year={2013},
   month=aug }

@article{Dowker_1995,
   title={Properties of consistent histories},
   volume={75},
   issn={1079-7114},
   url={http://dx.doi.org/10.1103/PhysRevLett.75.3038},
   doi={10.1103/physrevlett.75.3038},
   number={17},
   journal={Physical Review Letters},
   publisher={American Physical Society},
   author={Dowker, Fay and Kent, Adrian},
   year={1995},
   month=oct, pages={3038--3041} }

@article{Cotler_2018,
   title={Superdensity operators for spacetime quantum mechanics},
   volume={2018},
   issn={1029-8479},
   url={http://dx.doi.org/10.1007/JHEP09(2018)093},
   doi={10.1007/jhep09(2018)093},
   number={9},
   journal={Journal of High Energy Physics},
   publisher={Springer Science and Business Media LLC},
   author={Cotler, Jordan and Jian, Chao-Ming and Qi, Xiao-Liang and Wilczek, Frank},
   year={2018},
   month=sep }

@article{Aharonov_2009,
   title={Multiple-time states and multiple-time measurements in quantum mechanics},
   volume={79},
   issn={1094-1622},
   url={http://dx.doi.org/10.1103/PhysRevA.79.052110},
   doi={10.1103/physreva.79.052110},
   number={5},
   journal={Physical Review A},
   publisher={American Physical Society},
   author={Aharonov, Yakir and Popescu, Sandu and Tollaksen, Jeff and Vaidman, Lev},
   year={2009},
   month=may }

@article{Ahanorov_1_PhysRev.134.B1410,
  title = {Time symmetry in the quantum process of measurement},
  author = {Aharonov, Yakir and Bergmann, Peter G. and Lebowitz, Joel L.},
  journal = {Physical Review},
  volume = {134},
  number = {6B},
  pages = {B1410--B1416},
  numpages = {0},
  year = {1964},
  month = {Jun},
  publisher = {American Physical Society},
  doi = {10.1103/PhysRev.134.B1410},
  url = {https://link.aps.org/doi/10.1103/PhysRev.134.B1410}
}

@article{Milz_2019,
   title={Completely positive divisibility does not mean {Markovianity}},
   volume={123},
   issn={1079-7114},
   url={http://dx.doi.org/10.1103/PhysRevLett.123.040401},
   doi={10.1103/physrevlett.123.040401},
   number={4},
   journal={Physical Review Letters},
   publisher={American Physical Society},
   author={Milz, Simon and Kim, M. S. and Pollock, Felix A. and Modi, Kavan},
   year={2019},
   month=jul }

@article{White_2020,
   title={Demonstration of non-{Markovian} process characterisation and control on a quantum processor},
   volume={11},
   issn={2041-1723},
   url={http://dx.doi.org/10.1038/s41467-020-20113-3},
   doi={10.1038/s41467-020-20113-3},
   number={1},
   journal={Nature Communications},
   publisher={Springer Science and Business Media LLC},
   author={White, G. A. L. and Hill, C. D. and Pollock, F. A. and Hollenberg, L. C. L. and Modi, K.},
   year={2020},
   month=dec }

@article{Jorgensen_2019,
   title={Exploiting the causal tensor network structure of quantum processes to efficiently simulate non-{Markovian} path integrals},
   volume={123},
   issn={1079-7114},
   url={http://dx.doi.org/10.1103/PhysRevLett.123.240602},
   doi={10.1103/physrevlett.123.240602},
   number={24},
   journal={Physical Review Letters},
   publisher={American Physical Society},
   author={Jørgensen, Mathias R. and Pollock, Felix A.},
   year={2019},
   month=dec }

@article{Pollock_2018_2,
   title={Operational {Markov} condition for quantum processes},
   volume={120},
   issn={1079-7114},
   url={http://dx.doi.org/10.1103/PhysRevLett.120.040405},
   doi={10.1103/physrevlett.120.040405},
   number={4},
   journal={Physical Review Letters},
   publisher={American Physical Society},
   author={Pollock, Felix A. and Rodríguez-Rosario, César and Frauenheim, Thomas and Paternostro, Mauro and Modi, Kavan},
   year={2018},
   month=jan }

@article{Pollock_2018,
   title={Non-{Markovian} quantum processes: complete framework and efficient characterization},
   volume={97},
   issn={2469-9934},
   url={http://dx.doi.org/10.1103/PhysRevA.97.012127},
   doi={10.1103/physreva.97.012127},
   number={1},
   journal={Physical Review A},
   publisher={American Physical Society},
   author={Pollock, Felix A. and Rodríguez-Rosario, César and Frauenheim, Thomas and Paternostro, Mauro and Modi, Kavan},
   year={2018},
   month=jan }

@article{Oreshkov_2016,
   title={Causal and causally separable processes},
   volume={18},
   issn={1367-2630},
   url={http://dx.doi.org/10.1088/1367-2630/18/9/093020},
   doi={10.1088/1367-2630/18/9/093020},
   number={9},
   journal={New Journal of Physics},
   publisher={IOP Publishing},
   author={Oreshkov, Ognyan and Giarmatzi, Christina},
   year={2016},
   month=sep, pages={093020} }

@article{Araujo_2015,
   title={Witnessing causal nonseparability},
   volume={17},
   issn={1367-2630},
   url={http://dx.doi.org/10.1088/1367-2630/17/10/102001},
   doi={10.1088/1367-2630/17/10/102001},
   number={10},
   journal={New Journal of Physics},
   publisher={IOP Publishing},
   author={Araújo, Mateus and Branciard, Cyril and Costa, Fabio and Feix, Adrien and Giarmatzi, Christina and Brukner, Caslav},
   year={2015},
   month=oct, pages={102001} }

@article{Procopio_2015,
   title={Experimental superposition of orders of quantum gates},
   volume={6},
   issn={2041-1723},
   url={http://dx.doi.org/10.1038/ncomms8913},
   doi={10.1038/ncomms8913},
   number={1},
   journal={Nature Communications},
   publisher={Springer Science and Business Media LLC},
   author={Procopio, Lorenzo M. and Moqanaki, Amir and Araújo, Mateus and Costa, Fabio and Alonso Calafell, Irati and Dowd, Emma G. and Hamel, Deny R. and Rozema, Lee A. and Brukner, Caslav and Walther, Philip},
   year={2015},
   month=aug }

@article{Castro_Ruiz_2018,
   title={Dynamics of quantum causal structures},
   volume={8},
   issn={2160-3308},
   url={http://dx.doi.org/10.1103/PhysRevX.8.011047},
   doi={10.1103/physrevx.8.011047},
   number={1},
   journal={Physical Review X},
   publisher={American Physical Society},
   author={Castro-Ruiz, Esteban and Giacomini, Flaminia and Brukner, Caslav},
   year={2018},
   month=mar }

@article{Araujo_2014,
   title={Computational advantage from quantum-controlled ordering of gates},
   volume={113},
   issn={1079-7114},
   url={http://dx.doi.org/10.1103/PhysRevLett.113.250402},
   doi={10.1103/physrevlett.113.250402},
   number={25},
   journal={Physical Review Letters},
   publisher={American Physical Society},
   author={Araújo, Mateus and Costa, Fabio and Brukner, Caslav},
   year={2014},
   month=dec }

@article{Oreshkov_2012,
   title={Quantum correlations with no causal order},
   volume={3},
   issn={2041-1723},
   url={http://dx.doi.org/10.1038/ncomms2076},
   doi={10.1038/ncomms2076},
   number={1},
   journal={Nature Communications},
   publisher={Springer Science and Business Media LLC},
   author={Oreshkov, Ognyan and Costa, Fabio and Brukner, Caslav},
   year={2012},
   month=oct }

@article{Paz_PhysRevLett.71.3235,
  title = {Proposed test for temporal {Bell} inequalities},
  author = {Paz, Juan Pablo and Mahler, G\"unter},
  journal = {Physical Review Letters},
  volume = {71},
  number = {20},
  pages = {3235--3239},
  numpages = {0},
  year = {1993},
  month = {Nov},
  publisher = {American Physical Society},
  doi = {10.1103/PhysRevLett.71.3235},
  url = {https://link.aps.org/doi/10.1103/PhysRevLett.71.3235}
}

@article{Causality_assimetry_PhysRevX.8.031013,
  title = {Causal asymmetry in a quantum world},
  author = {Thompson, Jayne and Garner, Andrew J. P. and Mahoney, John R. and Crutchfield, James P. and Vedral, Vlatko and Gu, Mile},
  journal = {Physical Review X},
  volume = {8},
  number = {3},
  pages = {031013},
  numpages = {15},
  year = {2018},
  month = {Jul},
  publisher = {American Physical Society},
  doi = {10.1103/PhysRevX.8.031013},
  url = {https://link.aps.org/doi/10.1103/PhysRevX.8.031013}
}

@article{Elliott_2020,
   title={Extreme dimensionality reduction with quantum modeling},
   volume={125},
   issn={1079-7114},
   url={http://dx.doi.org/10.1103/PhysRevLett.125.260501},
   doi={10.1103/physrevlett.125.260501},
   number={26},
   journal={Physical Review Letters},
   publisher={American Physical Society},
   author={Elliott, Thomas J. and Yang, Chengran and Binder, Felix C. and Garner, Andrew J. P. and Thompson, Jayne and Gu, Mile},
   year={2020},
   month=dec }

@article{Elliott_2018,
   title={Superior memory efficiency of quantum devices for the simulation of continuous-time stochastic processes},
   volume={4},
   issn={2056-6387},
   url={http://dx.doi.org/10.1038/s41534-018-0064-4},
   doi={10.1038/s41534-018-0064-4},
   number={1},
   journal={npj Quantum Information},
   publisher={Springer Science and Business Media LLC},
   author={Elliott, Thomas J. and Gu, Mile},
   year={2018},
   month=mar }

@article{Garner_2017,
   title={Provably unbounded memory advantage in stochastic simulation using quantum mechanics},
   volume={19},
   issn={1367-2630},
   url={http://dx.doi.org/10.1088/1367-2630/aa82df},
   doi={10.1088/1367-2630/aa82df},
   number={10},
   journal={New Journal of Physics},
   publisher={IOP Publishing},
   author={Garner, Andrew J P and Liu, Qing and Thompson, Jayne and Vedral, Vlatko and Gu, Mile},
   year={2017},
   month=oct, pages={103009} }

@misc{zukowski2010temporalleggettgargbellinequalitiessequential,
      title={Temporal {Leggett-Garg-Bell} inequalities for sequential multi-time actions in quantum information processing, and a re-definition of macroscopic realism}, 
      author={Żukowski, Marek},
      year={2010},
      eprint={1009.1749},
      archivePrefix={arXiv},
      primaryClass={quant-ph},
      url={https://arxiv.org/abs/1009.1749}, 
}

@article{Pironio_2003,
   title={Violations of {Bell} inequalities as lower bounds on the communication cost of nonlocal correlations},
   volume={68},
   issn={1094-1622},
   url={http://dx.doi.org/10.1103/PhysRevA.68.062102},
   doi={10.1103/physreva.68.062102},
   number={6},
   journal={Physical Review A},
   publisher={American Physical Society},
   author={Pironio, Stefano},
   year={2003},
   month=dec }

@article{Brierley_2015,
   title={Nonclassicality of temporal correlations},
   volume={115},
   issn={1079-7114},
   url={http://dx.doi.org/10.1103/PhysRevLett.115.120404},
   doi={10.1103/physrevlett.115.120404},
   number={12},
   journal={Physical Review Letters},
   publisher={American Physical Society},
   author={Brierley, Stephen and Kosowski, Adrian and Markiewicz, Marcin and Paterek, Tomasz and Przysiezna, Anna},
   year={2015},
   month=sep }

@article{Ambainis_10.1145/581771.581773,
author = {Ambainis, Andris and Nayak, Ashwin and Ta-Shma, Amnon and Vazirani, Umesh},
title = {Dense quantum coding and quantum finite automata},
year = {2002},
issue_date = {July 2002},
publisher = {Association for Computing Machinery},
address = {New York, NY, USA},
volume = {49},
number = {4},
issn = {0004-5411},
url = {https://doi.org/10.1145/581771.581773},
doi = {10.1145/581771.581773},
journal = {Journal of the ACM},
month = jul,
pages = {496--511},
numpages = {16}
}

@article{wiesner_10.1145/1008908.1008920,
author = {Wiesner, Stephen},
title = {Conjugate coding},
year = {1983},
issue_date = {Winter-Spring 1983},
publisher = {Association for Computing Machinery},
address = {New York, NY, USA},
volume = {15},
number = {1},
issn = {0163-5700},
url = {https://doi.org/10.1145/1008908.1008920},
doi = {10.1145/1008908.1008920},
journal = {ACM SIGACT News},
month = jan,
pages = {78--88},
numpages = {11}
}

@article{physicalimplicationsoffundamentalperiod_PhysRevLett.124.241301,
  title = {Physical implications of a fundamental period of time},
  author = {Wendel, Garrett and Mart\'{\i}nez, Luis and Bojowald, Martin},
  journal = {Physical Review Letters},
  volume = {124},
  number = {24},
  pages = {241301},
  numpages = {5},
  year = {2020},
  month = {Jun},
  publisher = {American Physical Society},
  doi = {10.1103/PhysRevLett.124.241301},
  url = {https://link.aps.org/doi/10.1103/PhysRevLett.124.241301}
}

@article{Autonomus_temporal_prob_PhysRevX.11.011046,
  title = {Autonomous temporal probability concentration: clockworks and the second law of thermodynamics},
  author = {Schwarzhans, Emanuel and Lock, Maximilian P. E. and Erker, Paul and Friis, Nicolai and Huber, Marcus},
  journal = {Physical Review X},
  volume = {11},
  number = {1},
  pages = {011046},
  numpages = {22},
  year = {2021},
  month = {Mar},
  publisher = {American Physical Society},
  doi = {10.1103/PhysRevX.11.011046},
  url = {https://link.aps.org/doi/10.1103/PhysRevX.11.011046}
}

@misc{rankovic2015quantumclockssynchronisation,
      title={Quantum clocks and their synchronisation - the {Alternate Ticks Game}}, 
      author={Ranković, Sandra and Liang, Yeong-Cherng and Renner, Renato},
      year={2015},
      eprint={1506.01373},
      archivePrefix={arXiv},
      primaryClass={quant-ph},
      url={https://arxiv.org/abs/1506.01373}, 
}

@article{bounding_temporal_corr_PhysRevLett.111.020403,
  title = {Bounding temporal quantum correlations},
  author = {Budroni, Costantino and Moroder, Tobias and Kleinmann, Matthias and G\"uhne, Otfried},
  journal = {Physical Review Letters},
  volume = {111},
  number = {2},
  pages = {020403},
  numpages = {5},
  year = {2013},
  month = {Jul},
  publisher = {American Physical Society},
  doi = {10.1103/PhysRevLett.111.020403},
  url = {https://link.aps.org/doi/10.1103/PhysRevLett.111.020403}
}

@misc{bose2015uncoveringnonclassicalityschrodingercoherent,
      title={Uncovering a nonclassicality of the {Schr\"odinger} coherent state up to the macro-domain}, 
      author={Bose, S. and Home, D. and Mal, S.},
      year={2015},
      eprint={1509.00196},
      archivePrefix={arXiv},
      primaryClass={quant-ph},
      url={https://arxiv.org/abs/1509.00196}, 
}

@article{Calarco_1995,
   title={Optimal measurements of magnetic flux in superconducting circuits and macroscopic quantum mechanics},
   volume={198},
   issn={0375-9601},
   url={http://dx.doi.org/10.1016/0375-9601(94)01020-U},
   doi={10.1016/0375-9601(94)01020-u},
   number={4},
   journal={Physics Letters A},
   publisher={Elsevier BV},
   author={Calarco, Tommaso and Onofrio, Roberto},
   year={1995},
   month=mar, pages={279--285} }

@article{MAL20162265,
title = {Optimal violation of the {Leggett–Garg} inequality for arbitrary spin and emergence of classicality through unsharp measurements},
journal = {Physics Letters A},
volume = {380},
number = {29},
pages = {2265--2270},
year = {2016},
issn = {0375-9601},
doi = {10.1016/j.physleta.2016.05.020},
url = {https://www.sciencedirect.com/science/article/pii/S0375960116302122},
author = {Mal, Shiladitya and Majumdar, A. S.}
}

@article{Hardy_2007,
   title={Towards quantum gravity: a framework for probabilistic theories with non-fixed causal structure},
   volume={40},
   issn={1751-8121},
   url={http://dx.doi.org/10.1088/1751-8113/40/12/S12},
   doi={10.1088/1751-8113/40/12/s12},
   number={12},
   journal={Journal of Physics A: Mathematical and Theoretical},
   publisher={IOP Publishing},
   author={Hardy, Lucien},
   year={2007},
   month=mar, pages={3081--3099} }

@article{Fullwood_Geometry_PhysRevA.111.052438,
  title = {Geometry from quantum temporal correlations},
  author = {Fullwood, James and Vedral, Vlatko},
  journal = {Physical Review A},
  volume = {111},
  number = {5},
  pages = {052438},
  numpages = {5},
  year = {2025},
  month = {May},
  publisher = {American Physical Society},
  doi = {10.1103/PhysRevA.111.052438},
  url = {https://link.aps.org/doi/10.1103/PhysRevA.111.052438}
}

@misc{barrett2020quantumcausalmodels,
      title={Quantum causal models}, 
      author={Barrett, Jonathan and Lorenz, Robin and Oreshkov, Ognyan},
      year={2020},
      eprint={1906.10726},
      archivePrefix={arXiv},
      primaryClass={quant-ph},
      url={https://arxiv.org/abs/1906.10726}, 
}

@article{Allen_PhysRevX.7.031021,
  title = {Quantum common causes and quantum causal models},
  author = {Allen, John-Mark A. and Barrett, Jonathan and Horsman, Dominic C. and Lee, Ciar\'an M. and Spekkens, Robert W.},
  journal = {Physical Review X},
  volume = {7},
  number = {3},
  pages = {031021},
  numpages = {22},
  year = {2017},
  month = {Jul},
  publisher = {American Physical Society},
  doi = {10.1103/PhysRevX.7.031021},
  url = {https://link.aps.org/doi/10.1103/PhysRevX.7.031021}
}

@article{Woods_PRXQuantum.3.010319,
  title = {Quantum clocks are more accurate than classical ones},
  author = {Woods, Mischa P. and Silva, Ralph and P\"utz, Gilles and Stupar, Sandra and Renner, Renato},
  journal = {PRX Quantum},
  volume = {3},
  number = {1},
  pages = {010319},
  numpages = {66},
  year = {2022},
  month = {Feb},
  publisher = {American Physical Society},
  doi = {10.1103/PRXQuantum.3.010319},
  url = {https://link.aps.org/doi/10.1103/PRXQuantum.3.010319}
}

@article{Burdroni_PhysRevResearch.3.033051,
  title = {Ticking-clock performance enhanced by nonclassical temporal correlations},
  author = {Budroni, Costantino and Vitagliano, Giuseppe and Woods, Mischa P.},
  journal = {Physical Review Research},
  volume = {3},
  number = {3},
  pages = {033051},
  numpages = {18},
  year = {2021},
  month = {Jul},
  publisher = {American Physical Society},
  doi = {10.1103/PhysRevResearch.3.033051},
  url = {https://link.aps.org/doi/10.1103/PhysRevResearch.3.033051}
}

@article{Erker_PhysRevX.7.031022,
  title = {Autonomous quantum clocks: does thermodynamics limit our ability to measure time?},
  author = {Erker, Paul and Mitchison, Mark T. and Silva, Ralph and Woods, Mischa P. and Brunner, Nicolas and Huber, Marcus},
  journal = {Physical Review X},
  volume = {7},
  number = {3},
  pages = {031022},
  numpages = {12},
  year = {2017},
  month = {Aug},
  publisher = {American Physical Society},
  doi = {10.1103/PhysRevX.7.031022},
  url = {https://link.aps.org/doi/10.1103/PhysRevX.7.031022}
}

@article{Vieira_2024,
   title={Witnessing environment dimension through temporal correlations},
   volume={8},
   issn={2521-327X},
   url={http://dx.doi.org/10.22331/q-2024-01-10-1224},
   doi={10.22331/q-2024-01-10-1224},
   journal={Quantum},
   publisher={Verein zur Forderung des Open Access Publizierens in den Quantenwissenschaften},
   author={Vieira, Lucas B. and Milz, Simon and Vitagliano, Giuseppe and Budroni, Costantino},
   year={2024},
   month=jan, pages={1224} }

@article{Spee_2020,
doi = {10.1088/1367-2630/ab6d42},
url = {https://doi.org/10.1088/1367-2630/ab6d42},
year = {2020},
month = {feb},
publisher = {IOP Publishing},
volume = {22},
number = {2},
pages = {023028},
author = {Spee, Cornelia and Siebeneich, Hendrik and Gloger, Timm Florian and Kaufmann, Peter and Johanning, Michael and Kleinmann, Matthias and Wunderlich, Christof and Gühne, Otfried},
title = {Genuine temporal correlations can certify the quantum dimension},
journal = {New Journal of Physics}
}

@article{Wolf_PhysRevLett.102.190504,
  title = {Assessing quantum dimensionality from observable dynamics},
  author = {Wolf, Michael M. and Perez-Garcia, David},
  journal = {Physical Review Letters},
  volume = {102},
  number = {19},
  pages = {190504},
  numpages = {4},
  year = {2009},
  month = {May},
  publisher = {American Physical Society},
  doi = {10.1103/PhysRevLett.102.190504},
  url = {https://link.aps.org/doi/10.1103/PhysRevLett.102.190504}
}

@article{semi-device_independent_temporal_corr_PhysRevLett.132.220201,
  title = {Semi-device-independently characterizing quantum temporal correlations},
  author = {Chen, Shin-Liang and Eisert, Jens},
  journal = {Physical Review Letters},
  volume = {132},
  number = {22},
  pages = {220201},
  numpages = {7},
  year = {2024},
  month = {May},
  publisher = {American Physical Society},
  doi = {10.1103/PhysRevLett.132.220201},
  url = {https://link.aps.org/doi/10.1103/PhysRevLett.132.220201}
}

@misc{jia2026temporalkirkwooddiracquasiprobabilitydistribution,
      title={Temporal {Kirkwood-Dirac} quasiprobability distribution and unification of temporal state formalisms through temporal {Bloch} tomography}, 
      author={Jia, Zhian and Modi, Kavan and Kaszlikowski, Dagomir},
      year={2026},
      eprint={2601.05294},
      archivePrefix={arXiv},
      primaryClass={quant-ph},
      url={https://arxiv.org/abs/2601.05294}, 
}

@article{Arvidsson_Shukur_2024,
   title={Properties and applications of the {Kirkwood–Dirac} distribution},
   volume={26},
   issn={1367-2630},
   url={http://dx.doi.org/10.1088/1367-2630/ada05d},
   doi={10.1088/1367-2630/ada05d},
   number={12},
   journal={New Journal of Physics},
   publisher={IOP Publishing},
   author={Arvidsson-Shukur, David R M and Braasch Jr, William F and De Bièvre, Stephan and Dressel, Justin and Jordan, Andrew N and Langrenez, Christopher and Lostaglio, Matteo and Lundeen, Jeff S and Halpern, Nicole Yunger},
   year={2024},
   month=dec, pages={121201} }

@article{Lostaglio2023kirkwooddirac,
  doi = {10.22331/q-2023-10-09-1128},
  url = {https://doi.org/10.22331/q-2023-10-09-1128},
  title = {{Kirkwood-{D}irac} quasiprobability approach to the statistics of incompatible observables},
  author = {Lostaglio, Matteo and Belenchia, Alessio and Levy, Amikam and Hern{\'{a}}ndez-G{\'{o}}mez, Santiago and Fabbri, Nicole and Gherardini, Stefano},
  journal = {Quantum},
  issn = {2521-327X},
  publisher = {{Verein zur F{\"{o}}rderung des Open Access Publizierens in den Quantenwissenschaften}},
  volume = {7},
  pages = {1128},
  month = oct,
  year = {2023}
}

@article{Frowis_2018,
   title={Macroscopic quantum states: measures, fragility, and implementations},
   volume={90},
   issn={1539-0756},
   url={http://dx.doi.org/10.1103/RevModPhys.90.025004},
   doi={10.1103/revmodphys.90.025004},
   number={2},
   journal={Reviews of Modern Physics},
   publisher={American Physical Society},
   author={Fröwis, Florian and Sekatski, Pavel and Dür, Wolfgang and Gisin, Nicolas and Sangouard, Nicolas},
   year={2018},
   month=may }

@article{Geong_2014_PhysRevLett.112.010402,
  title = {Coarsening measurement references and the quantum-to-classical transition},
  author = {Jeong, Hyunseok and Lim, Youngrong and Kim, M. S.},
  journal = {Physical Review Letters},
  volume = {112},
  number = {1},
  pages = {010402},
  numpages = {5},
  year = {2014},
  month = {Jan},
  publisher = {American Physical Society},
  doi = {10.1103/PhysRevLett.112.010402},
  url = {https://link.aps.org/doi/10.1103/PhysRevLett.112.010402}
}

@article{harmonic_oscilator_CG_PhysRevLett.120.210402,
  title = {Nonclassicality of the harmonic-oscillator coherent state persisting up to the macroscopic domain},
  author = {Bose, S. and Home, D. and Mal, S.},
  journal = {Physical Review Letters},
  volume = {120},
  number = {21},
  pages = {210402},
  numpages = {6},
  year = {2018},
  month = {May},
  publisher = {American Physical Society},
  doi = {10.1103/PhysRevLett.120.210402},
  url = {https://link.aps.org/doi/10.1103/PhysRevLett.120.210402}
}

@article{Conditions_violating_MR_2008_PhysRevLett.101.090403,
  title = {Conditions for quantum violation of macroscopic realism},
  author = {Kofler, Johannes and Brukner, Caslav},
  journal = {Physical Review Letters},
  volume = {101},
  number = {9},
  pages = {090403},
  numpages = {4},
  year = {2008},
  month = {Aug},
  publisher = {American Physical Society},
  doi = {10.1103/PhysRevLett.101.090403},
  url = {https://link.aps.org/doi/10.1103/PhysRevLett.101.090403}
}

@misc{bibak2025classicallimitquantummechanics,
      title={The classical limit of quantum mechanics through coarse-grained measurements}, 
      author={Bibak, Fatemeh and Cepollaro, Carlo and Medina Sanchez, Nicolas and Dakic, Borivoje and Brukner, Caslav},
      year={2025},
      eprint={2503.15642},
      archivePrefix={arXiv},
      primaryClass={quant-ph},
      url={https://arxiv.org/abs/2503.15642}, 
}

@article{Kofler_CG_PhysRevLett.99.180403,
  title = {Classical world arising out of quantum physics under the restriction of coarse-grained measurements},
  author = {Kofler, Johannes and Brukner, Caslav},
  journal = {Physical Review Letters},
  volume = {99},
  number = {18},
  pages = {180403},
  numpages = {4},
  year = {2007},
  month = {Nov},
  publisher = {American Physical Society},
  doi = {10.1103/PhysRevLett.99.180403},
  url = {https://link.aps.org/doi/10.1103/PhysRevLett.99.180403}
}

@article{Horodecki_criterion,
title = {Violating {Bell} inequality by mixed spin-1/2 states: necessary and sufficient condition},
journal = {Physics Letters A},
volume = {200},
number = {5},
pages = {340--344},
year = {1995},
issn = {0375-9601},
doi = {10.1016/0375-9601(95)00214-N},
url = {https://www.sciencedirect.com/science/article/pii/037596019500214N},
author = {Horodecki, R. and Horodecki, P. and Horodecki, M.}
}

@article{Fine_PhysRevLett.48.291,
  title = {Hidden variables, joint probability, and the {Bell} inequalities},
  author = {Fine, Arthur},
  journal = {Physical Review Letters},
  volume = {48},
  number = {5},
  pages = {291--295},
  numpages = {0},
  year = {1982},
  month = {Feb},
  publisher = {American Physical Society},
  doi = {10.1103/PhysRevLett.48.291},
  url = {https://link.aps.org/doi/10.1103/PhysRevLett.48.291}
}

@article{2021Sonner,
  title = {Temporal entanglement, quasiparticles, and the role of interactions},
  author = {Giudice, Giacomo and Giudici, Giuliano and Sonner, Michael and Thoenniss, Julian and Lerose, Alessio and Abanin, Dmitry A. and Piroli, Lorenzo},
  journal = {Physical Review Letters},
  volume = {128},
  number = {22},
  pages = {220401},
  numpages = {6},
  year = {2022},
  month = {Jun},
  publisher = {American Physical Society},
  doi = {10.1103/PhysRevLett.128.220401},
  url = {https://link.aps.org/doi/10.1103/PhysRevLett.128.220401}
}

@misc{luchnikov2024scalabletomographymanybodyquantum,
      title={Scalable tomography of many-body quantum environments with low temporal entanglement}, 
      author={Luchnikov, Ilia A. and Sonner, Michael and Abanin, Dmitry A.},
      year={2024},
      eprint={2406.18458},
      archivePrefix={arXiv},
      primaryClass={quant-ph},
      url={https://arxiv.org/abs/2406.18458}, 
}

@article{Thoenniss_2025,
   title={Efficient pseudomode representation and complexity of quantum impurity models},
   volume={112},
   issn={2469-9969},
   url={http://dx.doi.org/10.1103/h8g7-bmng},
   doi={10.1103/h8g7-bmng},
   number={15},
   journal={Physical Review B},
   publisher={American Physical Society},
   author={Thoenniss, Julian and Vilkoviskiy, Ilya and Abanin, Dmitry A.},
   year={2025},
   month=oct }

@article{Lerose_2021,
   title={Scaling of temporal entanglement in proximity to integrability},
   volume={104},
   issn={2469-9969},
   url={http://dx.doi.org/10.1103/PhysRevB.104.035137},
   doi={10.1103/physrevb.104.035137},
   number={3},
   journal={Physical Review B},
   publisher={American Physical Society},
   author={Lerose, Alessio and Sonner, Michael and Abanin, Dmitry A.},
   year={2021},
   month=jul }

@misc{milekhin2025observablecomputableentanglementtime,
      title={Observable and computable entanglement in time}, 
      author={Milekhin, Alexey and Adamska, Zofia and Preskill, John},
      year={2025},
      eprint={2502.12240},
      archivePrefix={arXiv},
      primaryClass={quant-ph},
      url={https://arxiv.org/abs/2502.12240}, 
}

@article{PPT_HORODECKI19961,
title = {Separability of mixed states: necessary and sufficient conditions},
journal = {Physics Letters A},
volume = {223},
number = {1},
pages = {1--8},
year = {1996},
issn = {0375-9601},
doi = {10.1016/S0375-9601(96)00706-2},
url = {https://www.sciencedirect.com/science/article/pii/S0375960196007062},
author = {Horodecki, Michał and Horodecki, Paweł and Horodecki, Ryszard}
}

@article{stamatova2025complexheatcapacitywitness,
  title = {Complex heat capacity as a witness of spatio-temporal entanglement},
  author = {Stamatova, Mia and Vedral, Vlatko},
  journal = {Phys. Rev. A},
  volume = {114},
  issue = {2},
  pages = {022448},
  numpages = {7},
  year = {2026},
  month = {Aug},
  publisher = {American Physical Society},
  doi = {10.1103/vpbk-gcbl},
  url = {https://link.aps.org/doi/10.1103/vpbk-gcbl}
}

@article{Zhao_Geometryofquantumcorrelations_PhysRevA.98.052312,
  title = {Geometry of quantum correlations in space-time},
  author = {Zhao, Zhikuan and Pisarczyk, Robert and Thompson, Jayne and Gu, Mile and Vedral, Vlatko and Fitzsimons, Joseph F.},
  journal = {Physical Review A},
  volume = {98},
  number = {5},
  pages = {052312},
  numpages = {5},
  year = {2018},
  month = {Nov},
  publisher = {American Physical Society},
  doi = {10.1103/PhysRevA.98.052312},
  url = {https://link.aps.org/doi/10.1103/PhysRevA.98.052312}
}

@article{Marletto_temporal_teleportation,
author = {Marletto, Chiara and Vedral, Vlatko and Virzì, Salvatore and Avella, Alessio and Piacentini, Fabrizio and Gramegna, Marco and Degiovanni, Ivo Pietro and Genovese, Marco},
title = {Temporal teleportation with pseudo-density operators: how dynamics emerges from temporal entanglement},
journal = {Science Advances},
volume = {7},
number = {38},
pages = {eabe4742},
year = {2021},
doi = {10.1126/sciadv.abe4742},
url = {https://www.science.org/doi/abs/10.1126/sciadv.abe4742},
eprint = {https://www.science.org/doi/pdf/10.1126/sciadv.abe4742}}

@ARTICLE{CHSH_1969PhRvL..23..880C,
       author = {{Clauser}, John F. and {Horne}, Michael A. and {Shimony}, Abner and {Holt}, Richard A.},
        title = "Proposed experiment to test local hidden-variable theories",
      journal = {Physical Review Letters},
         year = 1969,
        month = oct,
       volume = {23},
       number = {15},
        pages = {880--884},
          doi = {10.1103/PhysRevLett.23.880},
       adsurl = {https://ui.adsabs.harvard.edu/abs/1969PhRvL..23..880C}
}

@article{Fritz_2010,
doi = {10.1088/1367-2630/12/8/083055},
url = {https://doi.org/10.1088/1367-2630/12/8/083055},
year = {2010},
month = {aug},
publisher = {},
volume = {12},
number = {8},
pages = {083055},
author = {Fritz, Tobias},
title = {Quantum correlations in the temporal {Clauser–Horne–Shimony–Holt} ({CHSH}) scenario},
journal = {New Journal of Physics}
}

@misc{brukner2004quantumentanglementtime,
      title={Quantum entanglement in time}, 
      author={Brukner, Caslav and Taylor, Samuel and Cheung, Sancho and Vedral, Vlatko},
      year={2004},
      eprint={quant-ph/0402127},
      archivePrefix={arXiv},
      primaryClass={quant-ph},
      url={https://arxiv.org/abs/quant-ph/0402127}, 
}

@article{Bera_2018,
doi = {10.1088/1361-6633/aa872f},
url = {https://doi.org/10.1088/1361-6633/aa872f},
year = {2017},
month = {dec},
publisher = {IOP Publishing},
volume = {81},
number = {2},
pages = {024001},
author = {Bera, Anindita and Das, Tamoghna and Sadhukhan, Debasis and Singha Roy, Sudipto and Sen(De), Aditi and Sen, Ujjwal},
title = {Quantum discord and its allies: a review of recent progress},
journal = {Reports on Progress in Physics}
}

@article{Modi_RevModPhys.84.1655,
  title = {The classical-quantum boundary for correlations: discord and related measures},
  author = {Modi, Kavan and Brodutch, Aharon and Cable, Hugo and Paterek, Tomasz and Vedral, Vlatko},
  journal = {Reviews of Modern Physics},
  volume = {84},
  number = {4},
  pages = {1655--1707},
  numpages = {0},
  year = {2012},
  month = {Nov},
  publisher = {American Physical Society},
  doi = {10.1103/RevModPhys.84.1655},
  url = {https://link.aps.org/doi/10.1103/RevModPhys.84.1655}
}

@article{Discord_Ollivier_PhysRevLett.88.017901,
  title = {Quantum discord: a measure of the quantumness of correlations},
  author = {Ollivier, Harold and Zurek, Wojciech H.},
  journal = {Physical Review Letters},
  volume = {88},
  number = {1},
  pages = {017901},
  numpages = {4},
  year = {2001},
  month = {Dec},
  publisher = {American Physical Society},
  doi = {10.1103/PhysRevLett.88.017901},
  url = {https://link.aps.org/doi/10.1103/PhysRevLett.88.017901}
}

@article{Horodecki_2009_entanglement,
   title={Quantum entanglement},
   volume={81},
   issn={1539-0756},
   url={http://dx.doi.org/10.1103/RevModPhys.81.865},
   doi={10.1103/revmodphys.81.865},
   number={2},
   journal={Reviews of Modern Physics},
   publisher={American Physical Society},
   author={Horodecki, Ryszard and Horodecki, Paweł and Horodecki, Michał and Horodecki, Karol},
   year={2009},
   month=jun, pages={865--942} }

@article{Bell_nonlocality_RevModPhys.86.419,
  title = {{Bell} nonlocality},
  author = {Brunner, Nicolas and Cavalcanti, Daniel and Pironio, Stefano and Scarani, Valerio and Wehner, Stephanie},
  journal = {Reviews of Modern Physics},
  volume = {86},
  number = {2},
  pages = {419--478},
  numpages = {60},
  year = {2014},
  month = {Apr},
  publisher = {American Physical Society},
  doi = {10.1103/RevModPhys.86.419},
  url = {https://link.aps.org/doi/10.1103/RevModPhys.86.419}
}

@article{Causal_classification_of_spatiotemporal_cor_Song,
  title = {Causal classification of spatiotemporal quantum correlations},
  author = {Song, Minjeong and Narasimhachar, Varun and Regula, Bartosz and Elliott, Thomas J. and Gu, Mile},
  journal = {Physical Review Letters},
  volume = {133},
  number = {11},
  pages = {110202},
  numpages = {6},
  year = {2024},
  month = {Sep},
  publisher = {American Physical Society},
  doi = {10.1103/PhysRevLett.133.110202},
  url = {https://link.aps.org/doi/10.1103/PhysRevLett.133.110202}
}

@article{Wu_2025_mutualinformationintime,
doi = {10.1088/1367-2630/adde7d},
url = {https://dx.doi.org/10.1088/1367-2630/adde7d},
year = {2025},
month = {jun},
publisher = {IOP Publishing},
volume = {27},
number = {6},
pages = {064504},
author = {Wu, Zhen and Parzygnat, Arthur J and Vedral, Vlatko and Fullwood, James},
title = {Quantum mutual information in time},
journal = {New Journal of Physics}
}

@ARTICLE{Jamiolkovski1972R,
       author = {{Jamio{\l}kowski}, A.},
        title = "Linear transformations which preserve trace and positive semidefiniteness of operators",
      journal = {Reports on Mathematical Physics},
         year = 1972,
        month = dec,
       volume = {3},
       number = {4},
        pages = {275--278},
          doi = {10.1016/0034-4877(72)90011-0},
       adsurl = {https://ui.adsabs.harvard.edu/abs/1972RpMP....3..275J}
}

@article{CHOI1975285,
title = {Completely positive linear maps on complex matrices},
journal = {Linear Algebra and its Applications},
volume = {10},
number = {3},
pages = {285--290},
year = {1975},
issn = {0024-3795},
doi = {10.1016/0024-3795(75)90075-0},
url = {https://www.sciencedirect.com/science/article/pii/0024379575900750},
author = {Choi, Man-Duen}
}

@article{CJ_Jiang,
  title = {Channel-state duality},
  author = {Jiang, Min and Luo, Shunlong and Fu, Shuangshuang},
  journal = {Physical Review A},
  volume = {87},
  number = {2},
  pages = {022310},
  numpages = {8},
  year = {2013},
  month = {Feb},
  publisher = {American Physical Society},
  doi = {10.1103/PhysRevA.87.022310},
  url = {https://link.aps.org/doi/10.1103/PhysRevA.87.022310}
}

@article{Johansen_MH,
  title = {Quantum theory of successive projective measurements},
  author = {Johansen, Lars M.},
  journal = {Physical Review A},
  volume = {76},
  number = {1},
  pages = {012119},
  numpages = {6},
  year = {2007},
  month = {Jul},
  publisher = {American Physical Society},
  doi = {10.1103/PhysRevA.76.012119},
  url = {https://link.aps.org/doi/10.1103/PhysRevA.76.012119}
}

@article{Clemente2016,
  title = {No {Fine} theorem for macrorealism: limitations of the {Leggett-Garg} inequality},
  author = {Clemente, Lucas and Kofler, Johannes},
  journal = {Physical Review Letters},
  volume = {116},
  number = {15},
  pages = {150401},
  numpages = {6},
  year = {2016},
  month = {Apr},
  publisher = {American Physical Society},
  doi = {10.1103/PhysRevLett.116.150401},
  url = {https://link.aps.org/doi/10.1103/PhysRevLett.116.150401}
}

@article{A_J_Leggett_2002,
doi = {10.1088/0953-8984/14/15/201},
url = {https://dx.doi.org/10.1088/0953-8984/14/15/201},
year = {2002},
month = {apr},
publisher = {},
volume = {14},
number = {15},
pages = {R415},
author = {Leggett, A. J.},
title = {Testing the limits of quantum mechanics: motivation, state of play, prospects},
journal = {Journal of Physics: Condensed Matter}
}

@misc{maroney2014quantumvsmacrorealism,
      title={Quantum versus macrorealism: what does the {Leggett-Garg} inequality actually test?}, 
      author={Maroney, Owen J. E. and Timpson, Christopher G.},
      year={2014},
      eprint={1412.6139},
      archivePrefix={arXiv},
      primaryClass={quant-ph},
      url={https://arxiv.org/abs/1412.6139}, 
}

@article{Leggett-Garg_original,
  title = {Quantum mechanics versus macroscopic realism: is the flux there when nobody looks?},
  author = {Leggett, A. J. and Garg, Anupam},
  journal = {Physical Review Letters},
  volume = {54},
  number = {9},
  pages = {857--860},
  numpages = {0},
  year = {1985},
  month = {Mar},
  publisher = {American Physical Society},
  doi = {10.1103/PhysRevLett.54.857},
  url = {https://link.aps.org/doi/10.1103/PhysRevLett.54.857}
}

@article{Kofler_2013,
   title={Condition for macroscopic realism beyond the {Leggett-Garg} inequalities},
   volume={87},
   issn={1094-1622},
   url={http://dx.doi.org/10.1103/PhysRevA.87.052115},
   doi={10.1103/physreva.87.052115},
   number={5},
   journal={Physical Review A},
   publisher={American Physical Society},
   author={Kofler, Johannes and Brukner, Caslav},
   year={2013},
   month=may }

@article{Clemente_2015,
   title={Necessary and sufficient conditions for macroscopic realism from quantum mechanics},
   volume={91},
   issn={1094-1622},
   url={http://dx.doi.org/10.1103/PhysRevA.91.062103},
   doi={10.1103/physreva.91.062103},
   number={6},
   journal={Physical Review A},
   publisher={American Physical Society},
   author={Clemente, Lucas and Kofler, Johannes},
   year={2015},
   month=jun }

@article{Margenau,
    author = {Margenau, Henry and Hill, Robert Nyden},
    title = {Correlation between measurements in quantum theory},
    journal = {Progress of Theoretical Physics},
    volume = {26},
    number = {5},
    pages = {722--738},
    year = {1961},
    month = {11},
    issn = {0033-068X},
    doi = {10.1143/PTP.26.722},
    url = {https://doi.org/10.1143/PTP.26.722},
    eprint = {https://academic.oup.com/ptp/article-pdf/26/5/722/5454875/26-5-722.pdf},
}

@article{Vidal_2000,
   title={Entanglement monotones},
   volume={47},
   issn={1362-3044},
   url={http://dx.doi.org/10.1080/09500340008244048},
   doi={10.1080/09500340008244048},
   number={2–3},
   journal={Journal of Modern Optics},
   publisher={Informa UK Limited},
   author={Vidal, Guifré},
   year={2000},
   month=feb, pages={355--376} }

@article{Wei_2003,
   title={Maximal entanglement versus entropy for mixed quantum states},
   volume={67},
   issn={1094-1622},
   url={http://dx.doi.org/10.1103/PhysRevA.67.022110},
   doi={10.1103/physreva.67.022110},
   number={2},
   journal={Physical Review A},
   publisher={American Physical Society},
   author={Wei, Tzu-Chieh and Nemoto, Kae and Goldbart, Paul M. and Kwiat, Paul G. and Munro, William J. and Verstraete, Frank},
   year={2003},
   month=feb }

@article{Emary_2013,
   title={{Leggett–Garg} inequalities},
   volume={77},
   issn={1361-6633},
   url={http://dx.doi.org/10.1088/0034-4885/77/1/016001},
   doi={10.1088/0034-4885/77/1/016001},
   number={1},
   journal={Reports on Progress in Physics},
   publisher={IOP Publishing},
   author={Emary, Clive and Lambert, Neill and Nori, Franco},
   year={2013},
   month=dec, pages={016001} }

@article{Vitagliano_2023,
   title={{Leggett-Garg} macrorealism and temporal correlations},
   volume={107},
   issn={2469-9934},
   url={http://dx.doi.org/10.1103/PhysRevA.107.040101},
   doi={10.1103/physreva.107.040101},
   number={4},
   journal={Physical Review A},
   publisher={American Physical Society},
   author={Vitagliano, Giuseppe and Budroni, Costantino},
   year={2023},
   month=apr }

@misc{fitzsimons2013quantumcorrelationsimplycausation,
      title={Quantum correlations which imply causation}, 
      author={Fitzsimons, Joseph and Jones, Jonathan and Vedral, Vlatko},
      year={2013},
      eprint={1302.2731},
      archivePrefix={arXiv},
      primaryClass={quant-ph},
      url={https://arxiv.org/abs/1302.2731}, 
}

@article{Liu_2025,
   title={Quantum causal inference with extremely light touch},
   volume={11},
   issn={2056-6387},
   url={http://dx.doi.org/10.1038/s41534-024-00956-0},
   doi={10.1038/s41534-024-00956-0},
   number={1},
   journal={npj Quantum Information},
   publisher={Springer Science and Business Media LLC},
   author={Liu, Xiangjing and Qiu, Yixian and Dahlsten, Oscar and Vedral, Vlatko},
   year={2025},
   month=mar }

@article{Jia_2023,
doi = {10.1088/1367-2630/ad1416},
url = {https://dx.doi.org/10.1088/1367-2630/ad1416},
year = {2023},
month = {dec},
publisher = {IOP Publishing},
volume = {25},
number = {12},
pages = {123038},
author = {Jia, Zhian and Song, Minjeong and Kaszlikowski, Dagomir},
title = {Quantum space-time marginal problem: global causal structure from local causal information},
journal = {New Journal of Physics}
}

@article{Zhang_2020,
doi = {10.1088/1367-2630/ab6b9f},
url = {https://dx.doi.org/10.1088/1367-2630/ab6b9f},
year = {2020},
month = {feb},
publisher = {IOP Publishing},
volume = {22},
number = {2},
pages = {023029},
author = {Zhang, Tian and Dahlsten, Oscar and Vedral, Vlatko},
title = {Different instances of time as different quantum modes: quantum states across space-time for continuous variables},
journal = {New Journal of Physics}
}

\appendix
\onecolumngrid

\section{The Choi-Jamio{\l}kowski matrix for a CPTP map $\mathcal{E}$} \label{CJ_Appendix}
The Choi matrix~\cite{CHOI1975285,Jamiolkovski1972R} to represent a CPTP map $\mathcal{E}:\mathcal{D}(\mathcal{H}^1) \to \mathcal{D}(\mathcal{H}^2)$ (where $\mathcal{H}^1$ and $\mathcal{H}^2$ have same dimension) is commonly defined as~\cite{CJ_Jiang}
\begin{equation}
    \mathcal{C}_\mathcal{E} = \mathcal{I}_R \otimes \mathcal{E}_2 \left( \ket{\Phi}\bra{\Phi} \right) = \sum_{ij} \ket{i}\bra{j}_R \otimes \mathcal{E}(\ket{i}\bra{j})_2,
\end{equation}
where the index $R$ means that the operator is acting in an auxiliary Hilbert space $\mathcal{H}_R$ with the same dimension as $\mathcal{H}^2$, the index 2 means that the operator is acting in the Hilbert space $\mathcal{H}^2$, $\mathcal{I}$ is the identity channel, and $\ket{\Phi} = \sum_{ii} \ket{ii}$ is the unnormalized maximally entangled state where $\{\ket{i}\}_i$ is a basis on its respective Hilbert space. With the above definition, we have the channel-state duality~\cite{CHOI1975285,Jamiolkovski1972R,CJ_Jiang}, that is, we can use this matrix to compute the action of the channel $\mathcal{E}$ in any quantum state $\sigma \in \mathcal{D}(\mathcal{H}^1)$
\begin{equation}
    \mathcal{E}(\sigma) = \text{Tr}_R \left[\sigma^\top_R \otimes \mathbb{I}_2  \mathcal{C}_\mathcal{E} \right].
\end{equation}

In this paper, we use a slightly different channel-state duality choice, which is the Choi-Jamio{\l}kowski matrix, to follow the same choice and structure as in Refs.~\cite{Liu_2025,Zhang_2020,fullwood2025}. We define the Choi-Jamio{\l}kowski matrix of a channel $\mathcal{E}$ as 
\begin{equation}
     M_\mathcal{E}= \sum_{ij} (\ket{i}\bra{j})^\top_R \otimes \mathcal{E}(\ket{i}\bra{j})_2, \label{CJ_definition}
\end{equation}
which results in the following channel-state duality
\begin{equation}
    \mathcal{E}(\sigma) = \text{Tr}_R \left[\sigma_R \otimes \mathbb{I}_2  M_\mathcal{E} \right]. \label{CJ_channel-state duality}
\end{equation}
An important detail is that the Choi-Jamio{\l}kowski matrix of a channel is the partial transpose of the Choi matrix of a channel: $M_\mathcal{E}= (T \otimes \mathbb{I}) \mathcal{C}_\mathcal{E}$, where $T$ is the transpose operator; therefore, it is not necessarily positive semi-definite.

\section{Proofsof Subsection~\ref{3_time_steps_section}}

\subsection{Proof of Theorem~\ref{theorem_R123_Q123}} \label{proof_of_theorem_R123_Q123_sec}
\begin{proof}
    To prove this theorem, we proceed similarly to the proof of Theorem 1 in Ref.~\cite{fullwood2025}.  To show the equality, we use Eq.~\eqref{R_m_steps} to obtain
    \begin{align}
       & \text{Tr} [R_{012} \Pi_i^0 \otimes \Pi_j^1 \otimes \Pi_k^2 ] \nonumber \\
      & = \frac{1}{2}\text{Tr}_{123}[(R_{01}\otimes \mathbb{I}_2 M_{12}+M_{12}R_{01}\otimes \mathbb{I}_2) \Pi_i^0 \otimes \Pi_j^1 \otimes \Pi_k^2 ] \nonumber \\
      & = \frac{1}{2}\text{Tr}_{13}[\text{Tr}_2[(R_{01} M_{12} \otimes \mathbb{I}_2 +M_{12}R_{01} \otimes \mathbb{I}_2)\Pi_j^1] \Pi_i^0 \otimes \Pi_k^2 ] \nonumber \\
      & = \frac{1}{2}\text{Tr}_{13}[\text{Tr}_2[ M_{12}(\Pi_j^1R_{01}+R_{01}\Pi_j^1)\otimes \mathbb{I}_2] \Pi_i^0 \otimes \Pi_k^2 ] \nonumber \\
      & = \frac{1}{2}\text{Tr}_{13}[\mathcal{E}_2 (\{ R_{01},\Pi_j^1 \}) \Pi_i^0 \otimes \Pi_k^2 ] \nonumber \\
      & = \frac{1}{4}\text{Tr}_{13}[\mathcal{E}_2 (\{ \rho \otimes \mathbb{I}_1  M_{01} + M_{01} \rho \otimes \mathbb{I}_1 ,\Pi_j^1 \}) \Pi_i^0 \otimes \Pi_k^2 ] \nonumber \\
      & = \frac{1}{4}\text{Tr}_3[\mathcal{E}_2 (\{ \text{Tr}_1[ (\rho M_{01} \otimes \mathbb{I}_1 + M_{01} \rho \otimes \mathbb{I}_1 )  \Pi_i^0 ],\Pi_j^1\}) \Pi_k^2 ] \nonumber \\
      & = \frac{1}{4}\text{Tr}_3[\mathcal{E}_2 (\{ \text{Tr}_1[ M_{01} (\Pi_i^0 \rho + \rho  \Pi_i^0 ) \otimes \mathbb{I}_1],\Pi_j^1  \}) \Pi_k^2 ] \nonumber \\
      & = \frac{1}{4}\text{Tr}_3[\mathcal{E}_2 (\{ \mathcal{E}_1 ( \{ \rho, \Pi_i^0 \} ),\Pi_j^1  \}) \Pi_k^2 ] \nonumber \\
      & = \mathbf{Q}_{012}(i,j,k) \nonumber
    \end{align}
    where $M_{01}$ is the CJ matrix of the channel $\mathcal{E}_1$, and $M_{12}$ is the CJ matrix of the channel $\mathcal{E}_2$; in the second and seventh equalities we used the cyclic property of the trace, in the fifth equality we used Eq.~\eqref{R_2_steps}, and in the forth and eighth equalities we used the CJ channel-state duality (Eq.~\eqref{CJ_channel-state duality}) of the respective CPTP maps. 

    To show that $R_{012}$ is the unique operator satisfying Eq.~\eqref{Q_definition_3steps}, suppose that exists some operator $\mathcal{O}$ satisfying $\text{Tr} [\mathcal{O} \Pi_i^0 \otimes \Pi_j^1 \otimes \Pi_k^2] = \text{Tr} [R_{012} \Pi_i^0 \otimes \Pi_j^1 \otimes \Pi_k^2] $, for all possible projectors $\{\Pi_i^0, \Pi_j^1, \Pi_k^2 \}$ . Then we have  $\text{Tr} [(\mathcal{O} - R_{012}) \Pi_i^0 \otimes \Pi_j^1 \otimes \Pi_k^2] = 0 $, which imples that $\mathcal{O} = R_{012}$ due to the non-degenerecy of the Hilbert-Schmidt inner product.
\end{proof}

\subsection{Proof of Lemma~\ref{lemma_quasiprob_3steps_separation}} \label{proof_lemma_quasiprob_3steps_separation_sec}

\begin{proof}
    For proving Eq.~\eqref{Q_P_D_relation_3points}, it is useful to notice that, for any operator $\mathcal{O}$, we have 
    \begin{align}
         \mathcal{O}_{\overline{i}} = \mathcal{O} - \{ \mathcal{O}, \Pi_i^0 \} + 2\Pi_i^0 \mathcal{O} \Pi_i^0, ~~ \text{and} ~~ \mathcal{O}_{\overline{j}} = \mathcal{O} - \{ \mathcal{O}, \Pi_j^1 \} + 2\Pi_j^1 \mathcal{O} \Pi_j^1. \nonumber
    \end{align}
    Therefore 
    \begin{align}
         \{ \mathcal{O}, \Pi_i^0 \} = - \mathcal{O}_{\overline{i}} + \mathcal{O} + 2\Pi_i^0 \mathcal{O} \Pi_i^0, ~~ \text{and} ~~
        \{ \mathcal{O}, \Pi_j^1 \} = - \mathcal{O}_{\overline{j}} + \mathcal{O}+ 2\Pi_j^1 \mathcal{O} \Pi_j^1. \label{lemma2_anticom}
    \end{align}
    Using this and Eqs.~\eqref{Q_definition_3steps} and~\eqref{Ludders_nV_3points}, we compute 
    \begin{align}
        & \mathbf{Q}_{012}(i,j,k) - \mathbf{P}_{012}(i,j,k) = \frac{1}{4} \text{Tr} \left[ \mathcal{E}_2 \left( \{ \mathcal{E}_1 \left( \{ \rho, \Pi_i^0 \} \right), \Pi_j^1 \} - 4 \Pi_j^1 \mathcal{E}_1 (\Pi_i^0 \rho \Pi_i^0) \Pi_j^1 \right) \Pi_k^2 \right] \nonumber \\
        & = \frac{1}{4} \text{Tr} \Big[ \mathcal{E}_2 \Big( \{ \mathcal{E}_1 \left(  -\rho_{\overline{i}} + \rho + 2\Pi_i^0 \rho \Pi_i^0 \right), \Pi_j^1 \} - 4 \Pi_j^1 \mathcal{E}_1 (\Pi_i^0 \rho \Pi_i^0) \Pi_j^1 \Big) \Pi_k^2 \Big] \nonumber \\
        & = \frac{1}{4} \text{Tr} \Big[ \mathcal{E}_2 \Big( -\{ \mathcal{E}_1 \left(  \rho_{\overline{i}} \right), \Pi_j^1 \} + \{ \mathcal{E}_1 \left(\rho \right) , \Pi_j^1 \} + 2 \{ \mathcal{E}_1 \left(\Pi_i^0 \rho \Pi_i^0\right) , \Pi_j^1 \} - 4 \Pi_j^1 \mathcal{E}_1 (\Pi_i^0 \rho \Pi_i^0) \Pi_j^1 \Big) \Pi_k^2 \Big] \nonumber \\
        & = \frac{1}{4} \text{Tr} \Big[ \mathcal{E}_2 \Big(\mathcal{E}_1 (\rho_{\overline{i}})_{\overline{j}} - \mathcal{E}_1 (\rho_{\overline{i}}) - 2 \Pi_j^1 \mathcal{E}_1 (\rho_{\overline{i}}) \Pi_j^1 - \mathcal{E}_1 (\rho)_{\overline{j}}  + \mathcal{E}_1 (\rho) +  2 \Pi_j^1 \mathcal{E}_1 (\rho) \Pi_j^1 - 2 \mathcal{E}_1(\Pi_i^0 \rho \Pi_i^0)_{\overline{j}} \nonumber \\
        & + 2 \mathcal{E}_1(\Pi_i^0 \rho \Pi_i^0) + 4\Pi_j^1 \mathcal{E}_1(\Pi_i^0 \rho \Pi_i^0) \Pi_j^1 - 4 \Pi_j^1 \mathcal{E}_1 (\Pi_i^0 \rho \Pi_i^0) \Pi_j^1 \Big) \Pi_k^2 \Big] \nonumber \\
        & = \frac{1}{4} \text{Tr} \Big[\mathcal{E}_2 (\mathcal{E}_1 (\rho_{\overline{i}})_{\overline{j}} - \mathcal{E}_1 (\rho_{\overline{i}})) \Pi_k^2\Big] + \frac{1}{2} \text{Tr} \Big[\mathcal{E}_2(\Pi_j^1 \mathcal{E}_1 (\rho) \Pi_j^1 - \Pi_j^1 \mathcal{E}_1 (\rho_{\overline{i}}) \Pi_j^1) \Pi_k^2\Big]  +   \frac{1}{4} \text{Tr}\Big[\mathcal{E}_2(\mathcal{E}_1 (\rho)- \mathcal{E}_1 (\rho)_{\overline{j}}) \Pi_k^2 \Big]   \nonumber \\
        & + \frac{1}{2} \text{Tr} \Big[\mathcal{E}_2(\mathcal{E}_1(\Pi_i^0 \rho \Pi_i^0) - \mathcal{E}_1(\Pi_i^0 \rho \Pi_i^0)_{\overline{j}}) \Pi_k^2 \Big] \nonumber \\
        & = \mathbf{D}_{\overline{0}\overline{1}2,\overline{1}2}(i,j,k) + \mathbf{D}_{12,\overline{0}12}(i,j,k)  + \mathbf{D}_{2,\overline{1}2}(i,j,k)  + \mathbf{D}_{02,0\overline{1}2}(i,j,k), \nonumber
    \end{align}
where in the second and fourth equalities, we used Eqs.~\eqref{lemma2_anticom}, and in the fifth equality, we just rearranged the terms.
\end{proof}

\subsection{AoT and NSIT for three-time-step and proof of Theorem~\ref{non-signaling_D_theorem_3steps}} \label{proof_no-signaling_D_lemma_3steps_sec}

For completeness and further use, we present the explicit form of the Arrow of Time (AoT) condition~\cite{Clemente_2015,Kofler_2013,A_J_Leggett_2002,Clemente2016}, which is always ensured by the postulates of quantum mechanics and CPTP evolutions, and the explicit form of the No-Signaling in Time (NSIT)~\cite{Kofler_2013,Clemente_2015,Clemente2016}, for the case of three-time-step presented in Subsection~\ref{3_time_steps_section}. Suppose that there are three observable quantities $Q^0, Q^1$, and $Q^2$ to be measured at times $t_0,~t_1$, and $t_2$, respectively. If the experiment is chosen to be made with only one measurement at an instant of time $t_\alpha$, then $P_\alpha (Q^\alpha_i)$ means the probability of obtaining an outcome $Q^\alpha_i$ for the observable $Q^\alpha$. If the experiment is chosen to be made with one measurement at an instant of time $t_\alpha$ and another measurement at an instant $t_\beta$, then $P_{\alpha \beta}(Q^\alpha_i,Q^\beta_j)$ means the probability of obtaining the sequential outcomes $Q^\alpha_i$ and $Q^\beta_j$ for the observables $Q^\alpha$ and $Q^\beta$, respectively. Likewise, if the experiment is chosen to make measurements at the three instants of time, then $P_{012}(Q^0_i,Q^1_j,Q^2_k)$ means the probability of obtaining the outcomes $Q^0_i,~Q^1_j$, and $Q^2_k$ at instants $t_0,~t_1$, and $t_2$, respectively. 

In this scenario, the AoT condition is given by the following equations
\begin{align}
    P_{01}(Q^0_i,Q^1_j) & = \sum_k P_{012}(Q^0_i,Q^1_j,Q^2_k),  \nonumber \\
    P_1(Q^1_j) & = \sum_k P_{12}(Q^1_j,Q^2_k),  \nonumber \\
    P_0(Q^0_i) & = \sum_j P_{01}(Q^0_i,Q^1_j),  \nonumber \\
    P_0(Q^0_i) & = \sum_k P_{02}(Q^0_i,Q^2_k).  \nonumber
\end{align}
If we relate projective the measurements $\{ \Pi^0_i \}_i,~\{ \Pi^1_j \}_j$, and $\{ \Pi^2_k \}_k$, to the eigenvectors $\{ \ket{Q^0_i} \}_i$, $\{ \ket{Q^1_j} \}_j$, and $\{ \ket{Q^2_k} \}_k$ of the observables $Q^0$, $Q^1$, and $Q^2$, respectively; the AoT conditions above will translate into
\begin{align}
    \mathbf{P}_{01}(i,j) & = \sum_k \mathbf{P}_{012}(i,j,k),  \label{AOT_1} \tag{AoT. 01$\overline{2}$} \\
    \mathbf{P}_1(j) & = \sum_k \mathbf{P}_{12}(j,k),  \label{AOT_2} \tag{AoT. 1$\overline{2}$} \\
     \mathbf{P}_0(i) & = \sum_j  \mathbf{P}_{01}(i,j),  \label{AOT_3} \tag{AoT. 0$\overline{1}$} \\
     \mathbf{P}_0(i) & = \sum_k  \mathbf{P}_{02}(i,k).  \label{AOT_4} \tag{AoT. 0$\overline{2}$}
\end{align}
Here we use the notation, already adopted in the main text, where $ \mathbf{P}_\alpha(i)$ means the probability, given by the Born rule, of obtaining the outcome $i$ for the respective projective measurement in the system at time $t_\alpha$; $\mathbf{P}_{\alpha\beta}(i,j)$ means the joint probability, given by the LvN rule (for example, Eq~\eqref{Ludders_vN}), to obtain the outcomes $i$ and $j$ for the respective projective measurements at times $t_\alpha$ and $t_\beta$; and $\mathbf{P}_{012}(i,j,k)$ means the joint probability, given by the LvN rule at three-time-step (Eq.~\eqref{Ludders_nV_3points}), to obtain the outcomes $i$, $j$ and $k$ for the respctive projective measurements at times $t_0$, $t_1$ and $t_2$.

We now describe the NSIT condition for this scenario. They obtain the following form~\cite{Clemente_2015}
\begin{align}
    \mathbf{P}_{12}(j,k) & = \sum_i \mathbf{P}_{012}(i,j,k),  \label{NSIT3_1} \tag{NSIT. $\overline{0}$12} \\
    \mathbf{P}_{02}(i,k) & = \sum_j \mathbf{P}_{012}(i,j,k),  \label{NSIT3_2} \tag{NSIT. 0$\overline{1}$2} \\
    \mathbf{P}_2(k) & = \sum_j \mathbf{P}_{12}(j,k),  \label{NSIT3_3} \tag{NSIT. $\overline{1}$2} 
\end{align}
Two other NSIT conditions are a consequence of the previous one, plus AoT. First,
\begin{equation}
    \mathbf{P}_2(k)  = \sum_i \mathbf{P}_{02}(i,k),  \label{NSIT3_4} \tag{NSIT. $\overline{0}$2} 
\end{equation}
which is a consequence of conditions~\eqref{NSIT3_1},~\eqref{NSIT3_2} and~\eqref{NSIT3_3}. Also,
\begin{equation}
    \mathbf{P}_1(j) = \sum_i \mathbf{P}_{01}(i,j), \tag{NSIT. $\overline{0}$1}
\end{equation}
which is a direct consequence of conditions $\eqref{NSIT3_1}$,~\eqref{AOT_1}, and~\eqref{AOT_2}.

    With explicit formulas for the AoT and NSIT conditions for three-time-step, we are now able to prove the following lemma.

 \begin{lemma} \label{non-signaling_D_lemma_3steps}
    For the case of three-time-step, with the choice of projective measurements $ \{ \Pi_i^0 \}_i,~\{ \Pi_j^1 \}_j$ and $\{\Pi_k^2\}_k$, at instants $t_0,t_1$ and $t_2$, respectively, the NSIT condition (together with the AoT condition) is valid if and only if 
     \begin{align}
         \mathbf{D}_{02,0\overline{1}2} (i,j,k) & = \mathbf{D}_{2,\overline{1}2} (i,j,k)  = \mathbf{D}_{12,\overline{0}12} (i,j,k) = 0, \label{non-signaling_D_lemma_3steps_eq}
     \end{align}
    for any $(i,j,k)$.
 \end{lemma}
 \begin{proof}[]
     We must prove that the validity of the three NSIT conditions (Eqs.~\eqref{NSIT3_1},~\eqref{NSIT3_2}, and~\eqref{NSIT3_3}) is equivalent to Eq.~\eqref{non-signaling_D_theorem_3steps_eq_1}. For this, notice that the first NSIT condition (Eq.~\eqref{NSIT3_1}) is equivalent to
     \begin{align}
       \text{Tr} [\mathcal{E}_2(\Pi_j^1\mathcal{E}_1(\rho)\Pi_j^1)\Pi_k^2] & = \sum_i \text{Tr} [\mathcal{E}_2(\Pi_j^1\mathcal{E}_1(\Pi_i^0 \rho \Pi_i^0)\Pi_j^1)\Pi_k^2] \nonumber \\
        \Leftrightarrow  \text{Tr} [\mathcal{E}_2(\Pi_j^1\mathcal{E}_1(\rho)\Pi_j^1)\Pi_k^2] & =  \text{Tr} [\mathcal{E}_2(\Pi_j^1\mathcal{E}_1( \sum_i \Pi_i^0 \rho \Pi_i^0)\Pi_j^1)\Pi_k^2] \nonumber \\
        \Leftrightarrow  \text{Tr} [\mathcal{E}_2(\Pi_j^1\mathcal{E}_1(\rho)\Pi_j^1)\Pi_k^2] & = \text{Tr} [\mathcal{E}_2(\Pi_j^1\mathcal{E}_1(\rho_{\overline{i}})\Pi_j^1)\Pi_k^2], \label{nsit_D_proof_step1}
     \end{align}
     where in the last step we used that $\sum_i \Pi_i^0 \rho \Pi_i^0 =\Pi_i^0 \rho \Pi_i^0 +(\mathbb{I}_0-\Pi_i^0)\rho(\mathbb{I}_0-\Pi_i^0) = \rho_{\overline{i}} $; this passage also gives an interpretation to the `over line' notation, meaning that the measurement respective to the index with the over line was averaged. Eq.~\eqref{nsit_D_proof_step1} is equivalent to $\mathbf{D}_{12,\overline{0}12} (i,j,k) = 0$ (see Eq.~\eqref{D_23}), therefore the condition of Eq.~\eqref{NSIT3_1} is equivalent to $\mathbf{D}_{12,\overline{0}12} (i,j,k) = 0$.

     In a completely analogous way, we can show that Eq.~\eqref{NSIT3_2} is equivalent to $\mathbf{D}_{02,0\overline{1}2} (i,j,k) = 0$ (see Eq.~\eqref{D_3_23}), and that Eq.~\eqref{NSIT3_3} is equivalent to $\mathbf{D}_{2,\overline{1}2} (i,j,k) = 0$ (see Eq.~\eqref{D_3_13}).
 \end{proof}
 
 Notice that Eq.~\eqref{non-signaling_D_theorem_3steps_eq_1} is a condition that does not involve the term $\mathbf{D}_{\overline{01}2,\overline{0}2} (i,j,k)$. This is because condition $\mathbf{D}_{02,0\overline{1}2} (i,j,k) = 0$ already implies $\mathbf{D}_{\overline{01}2,\overline{0}2} (i,j,k) =0$, since $\mathbf{D}_{\overline{01}2,\overline{0}2} = -\frac{1}{2} \sum_i \mathbf{D}_{02,0\overline{1}2} (i,j,k)$. Importantly, if Eq.~\eqref{non-signaling_D_theorem_3steps_eq_1} is satisfied, then we must have $\mathbf{D}_{012}(i,j,k) = 0$. With this lemma, we can prove Theorem~\ref{non-signaling_D_theorem_3steps}.

 \begin{proof}[\textbf{Proof of Theorem 5}]
    Notice that, from Lemma~\ref{non-signaling_D_lemma_3steps}, if NIST is satisfied, than clearly $\mathbf{D}_{012}(i,j,k)=0, \forall ~i,j,k$, which proves the first part of the theorem. Now, we must prove that if the projective measurements are of rank-1, $\mathbf{D}_{012}(i,j,k)=0, \forall ~i,j,k$, and $\mathbf{D}_{02}(i,k)=0, \forall ~i,k$, then all subterms in Eq.~\eqref{non-signaling_D_lemma_3steps_eq} are 0, implying the validity of NSIT. 
    
    In fact, suppose 
    \begin{equation}
        \mathbf{D}_{012}(i,j,k)=0, \forall ~i,j,k. \label{proof_thm_3steps_1}
    \end{equation}
    From Eq.~\eqref{Q_P_D_relation_3points}, we obtain
    \begin{align}
       & \sum_k \mathbf{Q}_{012}(i,j,k) = \sum_k\mathbf{P}_{012}(i,j,k)  \nonumber \\
       &  \Rightarrow \mathbf{Q}_{01}(i,j) = \mathbf{P}_{01}(i,j) \nonumber  \\
       &  \Rightarrow \mathbf{P}_{01}(i,j) + \mathbf{D}_{01}(i,j) = \mathbf{P}_{01}(i,j) \Rightarrow \mathbf{D}_{01}(i,j)=0.  \nonumber
    \end{align}
    Where, in the second equality above, we used the marginalization of the quasiprobability distribution and the AoT condition of Eq.~\eqref{AOT_1}, and in the third equality, we used Eq.~\eqref{Q_P_D_relation}. Now, $\mathbf{D}_{01}(i,j)=0$ implies (see Eq.~\eqref{D_definition})
   \begin{align}
       & \text{Tr}[\Pi_j^1 \left( \mathcal{E}_1(\rho) - \mathcal{E}_1(\rho_{\bar{i}}) \right)] = 0  \nonumber \\
        & \Rightarrow \mathcal{E}_1(\rho)_{j,j} = \mathcal{E}_1(\rho_{\bar{i}})_{j,j}, \label{proof_thm_3steps_2}
    \end{align}
where $X_{j,j}$ means the diagonal terms $(j,j)$ of the operator $X$ in the basis corresponding to the projector $\Pi_j^1$; the implication above holds since the trace of the first equality is valid for any $j$ for the rank-1 projectors $\Pi_j^1$. Using the equation above in Eq.~\eqref{D_23}, we obtain
\begin{align}
    \mathbf{D}_{12,\overline{0}12} (i,j,k) & = \frac{1}{2} \text{Tr} [\mathcal{E}_2 \left( \Pi_j^1 \mathcal{E}_1(\rho_{i})\Pi_j^1 - \Pi_j^1\mathcal{E}_1(\rho_{\overline{i}})\Pi_j^1 \right)\Pi_k^2] \nonumber \\
    & = \frac{1}{2} \text{Tr} [\mathcal{E}_2 \left( \Pi_j^1 ( \mathcal{E}_1(\rho_{i})_{j,j} - \mathcal{E}_1(\rho_{\overline{i}})_{j,j})\right)\Pi_k^2] =0, \label{proof_thm_3steps_3}
\end{align}
from which we obtain the vanishing of the subterm $\mathbf{D}_{12,\overline{0}12} (i,j,k)$ in Eq. \eqref{non-signaling_D_lemma_3steps_eq}. Using again Eq.~\eqref{proof_thm_3steps_1}, and summing over the first outcome, we obtain
\begin{align}
    & \sum_i \mathbf{Q}_{012}(i,j,k) = \sum_i\mathbf{P}_{012}(i,j,k)  \nonumber \\
       &  \Rightarrow \mathbf{Q}_{12}(j,k) = \mathbf{P}_{12}(j,k)  \nonumber \\
       &  \Rightarrow \mathbf{P}_{12}(j,k) + \mathbf{D}_{12}(j,k) = \mathbf{P}_{12}(j,k) \Rightarrow \mathbf{D}_{12}(j,k)=0, \nonumber
\end{align}
where, in the second equality, we used the marginalization of the quasiprobability, and we used Eq.~\eqref{proof_thm_3steps_3} which implies $\sum_i\mathbf{P}_{012}(i,j,k) = \mathbf{P}_{12}(j,k)$. From $\mathbf{D}_{12}(j,k)=0$ and Eq.~\eqref{D_definition} (recalling that, with the absence of measurement in $t_0$, the density matrix in $t_1$ is $\mathcal{E}_1(\rho)$), we obtain $\text{Tr}[\Pi_k^2\mathcal{E}_2(\mathcal{E}_1(\rho))] = \text{Tr}[\Pi_k^2\mathcal{E}_2(\mathcal{E}_1(\rho)_{\bar{j}})]$; using this in Eq.~\eqref{D_3_23}, we obtain
\begin{equation}
    \mathbf{D}_{2,\overline{1}2} (i,j,k)  = \frac{1}{4} \text{Tr} [\mathcal{E}_2 \left(\mathcal{E}_1(\rho)_{j}-\mathcal{E}_1(\rho)_{\overline{j}} \right)\Pi_k^2] = \frac{1}{4} \text{Tr} [\mathcal{E}_2 \left(\mathcal{E}_1(\rho)_{\overline{j}}-\mathcal{E}_1(\rho)_{\overline{j}} \right)\Pi_k^2] = 0, \label{proof_thm_3steps_4}
\end{equation}
from which we obtain the vanishing of the subterm $\mathbf{D}_{2,\overline{1}2} (i,j,k)$ in Eq. \eqref{non-signaling_D_lemma_3steps_eq}. 

Finally, we must prove the vanishing of the remaining subterm $\mathbf{D}_{02,0\overline{1}2} (i,j,k)$ of Eq. \eqref{non-signaling_D_lemma_3steps_eq}. To do this, notice that, from the vanishing of $\mathbf{D}_{123} (i,j,k)$, we have $\mathbf{P}_{012}(i,j,k) = \mathbf{Q}_{012}(i,j,k)$; therefore $\sum_j\mathbf{P}_{012}(i,j,k) = \sum_j \mathbf{Q}_{012}(i,j,k)$. Furthermore, notice that $\sum_{j} \textbf{Q}(i,j,k)=\text{Tr}[\Pi_i^0 \otimes \sum_j\Pi_j^1 \otimes \Pi_k^2 R_{012}] = \text{Tr}_{012}[\Pi_i^0 \otimes \mathbb{I} \otimes \Pi_k^2 R_{012}] = \text{Tr}_{02}[\Pi_i^0 \otimes \Pi_k^2 R_{02}] = \mathbf{Q}_{02}(i,k)$, where in the last equality we used Eq. \eqref{partial_trace_PDM}. Hence, we have

\begin{align}
  & \sum_j\mathbf{P}_{012}(i,j,k) = \mathbf{Q}_{02}(i,k) = \mathbf{P}_{02}(i,k) + \mathbf{D}_{02}(i,k) \nonumber \\
 \Rightarrow & \sum_j\mathbf{P}_{012}(i,j,k) - \mathbf{P}_{02}(i,k) = \mathbf{D}_{02}(i,k) , \nonumber \\
 \Rightarrow & \sum_j\text{Tr}\left[ \mathcal{E}_2(\Pi_j^1 \mathcal{E}_1( \Pi_i^0 \rho \Pi_i^0) \Pi_j^1) \Pi_k^2 \right] - \text{Tr}\left[ \mathcal{E}_2( \mathcal{E}_1( \Pi_i^0 \rho \Pi_i^0) ) \Pi_k^2 \right] = \mathbf{D}_{02}(i,k) , \nonumber \\
 \Rightarrow & \text{Tr}\left[ \mathcal{E}_2( \mathcal{E}_1( \Pi_i^0 \rho \Pi_i^0)_{\overline{j}} ) \Pi_k^2 \right] - \text{Tr}\left[ \mathcal{E}_2( \mathcal{E}_1( \Pi_i^0 \rho \Pi_i^0) ) \Pi_k^2 \right] = \mathbf{D}_{02}(i,k) , \nonumber \\
 \Rightarrow & -2 \mathbf{D}_{02,0\overline{1}2}(i,j,k) = \mathbf{D}_{02}(i,k) = 0 \Rightarrow \mathbf{D}_{02,0\overline{1}2}(i,j,k) = 0, \nonumber
\end{align}

where, in the last line, we used the definition of $\mathbf{D}_{02,0\overline{1}2}(i,j,k)$ (Eq. \eqref{D_13}) and the assumption that $\mathbf{D}_{02}(i,k)=0$. 
 \end{proof}

 \subsection{Proof of Lemma \ref{lemma_LGI_negativity}} \label{proof_lemma_LGI_negativity_sec}
 \begin{proof}
    To prove this lemma, we assume that the PDM $R_{012}$ for this three-time-step scenario is positive semi-definite, and we show that in this case the statistics of the dichotomic observable $Q$ must satisfy LGIs. Assuming that the PDM $R_{012}$ is positive semi-definite, it will satisfy all the properties of a density matrix. Consequently, the quasiprobability obtained by the temporal Born rule is $\mathbf{Q}_{012}(i,j,k) = \text{Tr}[R_{012} ~\Pi^0_i \otimes \Pi^1_j \otimes \Pi^2_k] \geq 0$. Furthermore, note that we can write the correlation function of Eq. \eqref{correlation_func_R_01} as 
    \begin{align}
        \langle Q^0 Q^1 \rangle & = \text{Tr} [R_{012} Q\otimes Q \otimes \mathbb{I}_2] \nonumber \\
         & =\text{Tr} \left[ R_{012} \left( \sum_i q_i^0 \Pi_i^0 \right) \otimes \left( \sum_j q_j^1 \Pi_j^1 \right) \otimes \left(\sum_k \Pi_k^2 \right) \right] \nonumber \\
         & = \sum_{i,j,k} q_i^0 q_j^1 \mathbf{Q}_{012}(i,j,k), \nonumber
    \end{align}
    where $q_i^0(q_j^1)$ are the eigenvalues of the observable $Q$ at time $t_0(t_1)$ and $\Pi_i^0, \Pi_j^1$, and $\Pi_k^2$ are the projectors for the eigenvectors of $Q$ at times $t_0,t_1$ and $t_2$, respectively. Similarly, we obtain $\langle Q^0 Q^2 \rangle = \sum_{i,j,k} q_i^0 q_k^2 \mathbf{Q}_{012}(i,j,k)$ and $\langle Q^1 Q^2 \rangle = \sum_{i,j,k} q_j^1 q_k^2 \mathbf{Q}_{012}(i,j,k)$. Using these correlations in the l.h.s. of the Leggett-Garg inequality (Eq. \eqref{Legget-Garg_ineq}), we obtain
    \begin{align}
        & \langle Q^0 Q^1 \rangle + \langle Q^1 Q^2 \rangle - \langle Q^0 Q^2 \rangle \nonumber \\
        = & \sum_{i,j,k} ( q_i^0 q_j^1 + q_j^1 q_k^2 - q_i^0 q_k^2 ) \mathbf{Q}_{012}(i,j,k) \nonumber \\
        \leq & \sum_{i,j,k}  \mathbf{Q}_{012}(i,j,k) = 1, \nonumber
    \end{align}
    where the inequality above is obtained from the fact that, if there are three fixed numbers $x,y,z$ with values $\pm 1$, then $xy+yz-xz \leq 1$. The proof that the other LGI inequalities are also satisfied when $f(R_{012})=0$ is completely analogous to this one.
\end{proof}

\section{Additional details of Section~\ref{sec_types_of_temporal_nonclasscic}}
\label{Appendix_details_temporal_nonclassic}

\subsection{Proof of Theorem \ref{temoral_CHSH_positivity_theorem}}
To prove this theorem, first we must define the entanglement-breaking channel \cite{Horodecki_2003,Ruskai_2003}. The entanglement-breaking channel can be defined as a channel $\Phi(\cdot)$ such that, for an input density matrix $\rho$:
\begin{equation}
    \Phi(\rho) = \sum_k \sigma_k \Tr[E_k \rho], \label{entanglement_breaking_def}
\end{equation}
where $\{ \sigma \}_k$ is a set of density matrices and $\{ E_k \}_k$ is a POVM with outcomes $k$. This type of channel has very important equivalences, properties, and applications \cite{Vieira_2024}; here, we use the fact that a channel is entanglement-breaking if and only if it has a separable Choi matrix \cite{Horodecki_2003}.
\begin{lemma}
    For a qubit initially in a full-rank state $\rho$, if the PDM $R_{01} = \frac{1}{2} \{ \rho \otimes \mathbb{I}, M_\mathcal{E} \}$ is positive semi-definite, then the channel $\mathcal{E}$, described by the Choi-Jamio{\l}kowski matrix $M_\mathcal{E}$, is an entanglement-breaking channel. 
\end{lemma}
\begin{proof}
    If $\rho$ is a full-rank density matrix, we must have $\rho > 0$. Furthermore, let $R_{01} = \frac{1}{2} \{ \rho \otimes \mathbb{I}, M_\mathcal{E} \}$ be positive semi-definite. Then, by making the spectral decomposition of the Choi-Jamio{\l}kowski matrix $ M_\mathcal{E} = \sum_\lambda \lambda \ket{\lambda} \bra{\lambda}$, we have
    \begin{align}
        \bra{\lambda} R_{01} \ket{\lambda} = \lambda \bra{\lambda} \rho \otimes \mathbb{I} \ket{\lambda}.
    \end{align}
    From the equation above, we conclude that if there is a negative eigenvalue of $M_\mathcal{E}$, then $\bra{\lambda} R_{01} \ket{\lambda}<0$, since $\rho >0$. This contradicts the assumption that $ R_{01}$ is positive semi-definite; therefore, we must have $M_\mathcal{E} \geq 0 $. 

    Now, $M_\mathcal{E}=(T\otimes \mathbb{I}) \mathcal{C}_\mathcal{E}$ is the partial transpose of the Choi matrix $\mathcal{C}_\mathcal{E}$ of the channel $\mathcal{E}$ (see Appendix \ref{CJ_Appendix}). Therefore, from the PPT criterion for a two-qubit density matrix \cite{PPT_HORODECKI19961}, the state represented by the Choi matrix $\mathcal{C}_\mathcal{E}$ is separable. Finally, Theorem 4 of Ref. \cite{Horodecki_2003} shows that a channel is entanglement-breaking if and only if its Choi matrix is separable; hence, the channel $\mathcal{E}$ is entanglement-breaking.
\end{proof}

Using the lemma above, we have the following proof of Theorem \ref{temoral_CHSH_positivity_theorem}:
\begin{proof}
    For a qubit system, if the initial density matrix $\rho$ describing the system is full-rank and the PDM describing the evolution from time $t_0$ to $t_1$ is $R_{01} = \frac{1}{2} \{ \rho \otimes \mathbb{I}, M_\mathcal{E} \}$ is positive semi-definite, then, from the lemma above, the channel $\mathcal{E}$ associated with the Choi-Jamio{\l}kowski matrix $M_\mathcal{E}$ is entanglement-breaking. 

    If the channel $\mathcal{E}$ is entanglement-breaking, then it has the form (see Ref. \cite{Ruskai_2003})
    \begin{equation}
        \mathcal{E}(X) = \sum_{\lambda = \pm 1} \Tr[P_\lambda X] \tau_\lambda,
    \end{equation}
    where $P_\lambda = \frac{1}{2}(\mathbb{I}+\lambda \vec{n} \cdot \vec{\sigma})$, $\tau_\lambda= \frac{1}{2} (\mathbb{I}+ \vec{s_\lambda}\cdot \vec{\sigma})$, for $\lambda = \pm 1$, and $\vec{n}$ and $\vec{s}_\lambda$ are unit vectors. By defining $\vec{u} = \frac{\vec{s}_+ + \vec{s}_-}{2}$ and $\vec{v} = \frac{\vec{s}_+ - \vec{s}_-}{2}$, we have $\vec{s}_\lambda =  \vec{u} + \lambda \vec{v}$. Furthermore, for an input state $\rho_x = \frac{1}{2}(\mathbb{I}+\vec{x}\cdot \vec{\sigma})$, where $|\vec{x}|\leq 1$, we have
    \begin{equation}
        \mathcal{E}(\rho_x) = \frac{1}{2} \left( \mathbb{I} + ((\vec{n}\cdot \vec{x})\vec{v}+\vec{u})\cdot \vec{\sigma} \right). \label{theorem_CHSH_positivity_step1}
    \end{equation}
    Using this equation, we can compute the temporal correlation $\langle A B \rangle$ between an observable $A$ in time $t_0$ and $B$ in time $t_1$. Supposing that $A = \vec{a}\cdot \vec{\sigma}$ and $B = \vec{b}\cdot \vec{\sigma}$, for unit vectors $\vec{a}$ and $\vec{b}$, we can compute
    \begin{equation}
        \langle A B \rangle = \sum_{\alpha = \pm 1} \sum_{\beta = \pm 1} \alpha \beta ~ p(\beta|\alpha) p(\alpha), \label{theorem_CHSH_positivity_step2}
    \end{equation}
    where $p(\alpha) $ is the probability of obtaining an outcome $\alpha$, given a projective measurement of the operator $A$ at time $t_0$, and $p(\beta | \alpha)$ is the probability of obtaining the outcome $\beta$ for a projective measurement of the operator $B$ at time $t_1$ after the first measurement and the evolution of the system. 

    For an initial state $\rho = \frac{1}{2} (\mathbb{I} + \vec{r}\cdot \vec{\sigma})$, where $|\vec{r}| \leq 1$, we have $p(\alpha) = \frac{1}{2}(1+ \alpha \vec{r} \cdot \vec{a})$ and 
    \begin{equation}
        p(\beta | \alpha) = \Tr[ \frac{1}{2}(\mathbb{I}+\beta \vec{b}\cdot \vec{\sigma}) \mathcal{E}(\rho_\alpha) ],
    \end{equation}
     where $p_\alpha = \frac{1}{2}(\mathbb{I}+\alpha \vec{a}\cdot\vec{\sigma})$ is the state collapsed after the first measurement. Using the equation above and Eq. \eqref{theorem_CHSH_positivity_step1} in \eqref{theorem_CHSH_positivity_step2}, we obtain
     \begin{align}
         \langle A B \rangle & = \sum_{\alpha=\pm 1} \sum_{\beta=\pm 1} \alpha \beta ~ \frac{1}{2}\left( 1 + \alpha \beta (\vec{n}\cdot \vec{a})(\vec{b}\cdot \vec{v}) + \beta (\vec{u}\cdot \vec{b}) \right) \frac{1}{2}(1 + \alpha \vec{r} \cdot \vec{a}) \nonumber \\
         & = (\vec{n} \cdot \vec{a}) (\vec{b}\cdot \vec{v}) + (\vec{u} \cdot \vec{b}) (\vec{r}\cdot \vec{a}).
     \end{align}
    Writing the temporal correlations matrix $T$ as $\langle AB \rangle = \vec{a}^\top T \vec{b}$, the equation above implies
    \begin{equation}
        T = \vec{n} \vec{v}^\top + \vec{r} \vec{u}^\top.
    \end{equation}
    Computing the trace norm squared of this matrix, we obtain
    \begin{align}
        \| T \|^2 & = |\vec{n}|^2 |\vec{v}|^2 + 2 (\vec{u} \cdot \vec{v})(\vec{r}\cdot \vec{n}) + |\vec{u}|^2 |\vec{r}|^2 \nonumber \\
        & \leq |\vec{v}|^2 + |\vec{u}|^2 + 2 | \vec{u} \cdot \vec{v}| = \max \{ s_+, s_- \} \leq 1.
    \end{align}
    From the Horodecki criterion \cite{Horodecki_criterion} (see Eq. \eqref{CHSH_max_violation}), the maximum Bell parameter (Eq. \eqref{time_CHSH}) is $\mathcal{B}_\text{max} = 2\sqrt{\lambda_1 + \lambda_2}$, where $\lambda_1$ and $\lambda_2$ are the largest eigenvalues of $T^\top T$. Since $\sqrt{\lambda_1 + \lambda_2} \leq \| T \|$, then $\mathcal{B}_\text{max} = 2 \sqrt{\lambda_1+\lambda_2} \leq 2\|T\| \leq 2.$

    Finally, for the case where the initial state $\rho = \ket{\psi} \bra{\psi}$ is rank-1, the positive semi-definite PDM $R_{01}$ is equivalent to a density matrix with a marginal that is pure. This necessarily implies that the PDM is in the product form 
    \begin{equation}
        R_{01} = \ket{\psi} \bra{\psi} \otimes \rho_1,
    \end{equation}
    where $\rho_1$ is some density matrix acting in the Hilbert space corresponding to $t_1$. Since this PDM is separable, it cannot violate CHSH.
    
\end{proof}

\subsection{Example of violation of MR without violation of temporal CHSH} \label{subsection_violationMR}
In a two-time-step scenario, a simple example of a violation of MR without the violation of temporal CHSH can be obtained when the initial state of the system is the pure state $\rho = \ket
{0} \bra{0}$ and it evolves under the depolarizing channel described in Eq.~\eqref{depolariaing_channel}. The PDM for this case is obtained using Eq.~\eqref{R_2_steps}, resulting in
\begin{equation}
    R_{01} = \begin{pmatrix}
        1 - \frac{\eta}{2} & 0 & 0 & 0 \\
        0 & \frac{\eta}{2} & \frac{1-\eta}{2} & 0 \\
        0 & \frac{1-\eta}{2} & 0 & 0 \\
        0 & 0 & 0 & 0
    \end{pmatrix}. \label{PDM_two_steps_last_example}
\end{equation}
Using Eq.~\eqref{negativity_def}, we obtain 
\begin{align}
    f(R_{01}) = \frac{1}{4} \left( \sqrt{4+6 \eta^2 -2\eta(4+\sqrt{4+\eta(-8+5\eta)})} + \sqrt{4+6 \eta^2 -2\eta(4-\sqrt{4+\eta(-8+5\eta)})} -2 \eta \right). \label{f_last_2steps_example}
\end{align}
which implies that $f(R_{01})$ does not vanish for $\eta \in [0,1)$ (see Fig.~\ref{plot_CHSH_NSIT_example}). Moreover, the value of $\mathcal{B}_{\text{max}}$ for this case is identical to the case of Eq.~\eqref{depolarizing_PDM}. That is, using Eq.~\eqref{optimal_B_CHSH} for the PDM $R_{01}$ above, we obtain $\mathcal{B}_{\text{max}} = 2\sqrt{2} (1-\eta)$ (more generally, one can show that for an arbitrary initial qubit state evolving under the depolarizing channel of Eq.~\eqref{depolariaing_channel}, the value $\mathcal{B}_{\text{max}}$ gives this same function of $\eta$; that is, it is independent of the initial state, analogously to what happens in the case of unitary evolution). Finally, using Eqs.~\eqref{D_definition} and ~\eqref{NSIT_witness} for the respective initial state and channel, and choosing the measurements at time $t_0$ and $t_1$ to be the projective measurements of the observables $A_0 = \vec{\sigma}\cdot \vec{r}_0$ and $A_1 = \vec{\sigma}\cdot \vec{r}_1$,
where $\vec{r}_0 = (\cos \phi \sin\theta_0,\sin \phi \sin\theta_0,\cos\theta_0)$ and $\vec{r}_1 = (\cos \phi \sin\theta_1,\sin \phi \sin\theta_1,\cos\theta_1)$ we obtain

\begin{equation}
    \mathcal{N}_{01} = (1-\eta) |\sin \theta_0 \sin (\theta_1 - \theta_0)|.
\end{equation}
This shows that MR is satisfied only if $\eta =1$ (full depolarizing), or $A_0 = \sigma_z$ (that is, the initial state is an eigenvector of the first observable), or $\theta_1 = \theta_0 +n\pi, ~ n\in \mathbb{Z}$ (the operators $A_0$ and $A_1$ are the same or parallel in the Bloch sphere). Therefore, for most measurement choices and $\eta \in [1-\sqrt{2}/2,1)$, we have a region where MR is violated while the temporal CHSH is satisfied, additionally, the PDM is always negative, and therefore temporally entangled, except for the full depolarizing case (see Fig.~\ref{plot_CHSH_NSIT_example} for a comparison for $\theta_0 = \pi/2$ and $\theta_1 = \pi$). Note that we are assuming that the operators $A_0$ and $A_1$ have the same direction in the $x,y$-plane; if this is not assumed, the conditions for $\mathcal{N}_{01} = 0$ become even more restrictive.

As a final remark, notice that all temporal nonclassicality witnesses vanish for $\eta = 1$. This occurs because the full depolarizing channel replaces any initial state with the maximally mixed state, resulting in no correlations (classical or quantum) between the initial and final systems. Consequently, the PDM must be separable, as we can clearly see if we consider $\eta = 1$ in Eq.~\eqref{PDM_two_steps_last_example}:
\begin{equation}
    R_{01}= \frac{1}{2} \begin{pmatrix}
        1  & 0 & 0 & 0 \\
        0 & 1 & 0 & 0 \\
        0 & 0 & 0 & 0 \\
        0 & 0 & 0 & 0
    \end{pmatrix} = ( \ket{0}\bra{0} ) \otimes \left( \frac{\mathbb{I}}{2} \right),
\end{equation}
that is, it is the tensor product of the initial state $\ket{0}\bra{0}$ in the Hilbert space of time $t_0$ with the final state, which is maximally mixed, in the Hilbert space of time $t_1$.

\begin{figure}
    \centering
    \includegraphics[width=0.6\linewidth]{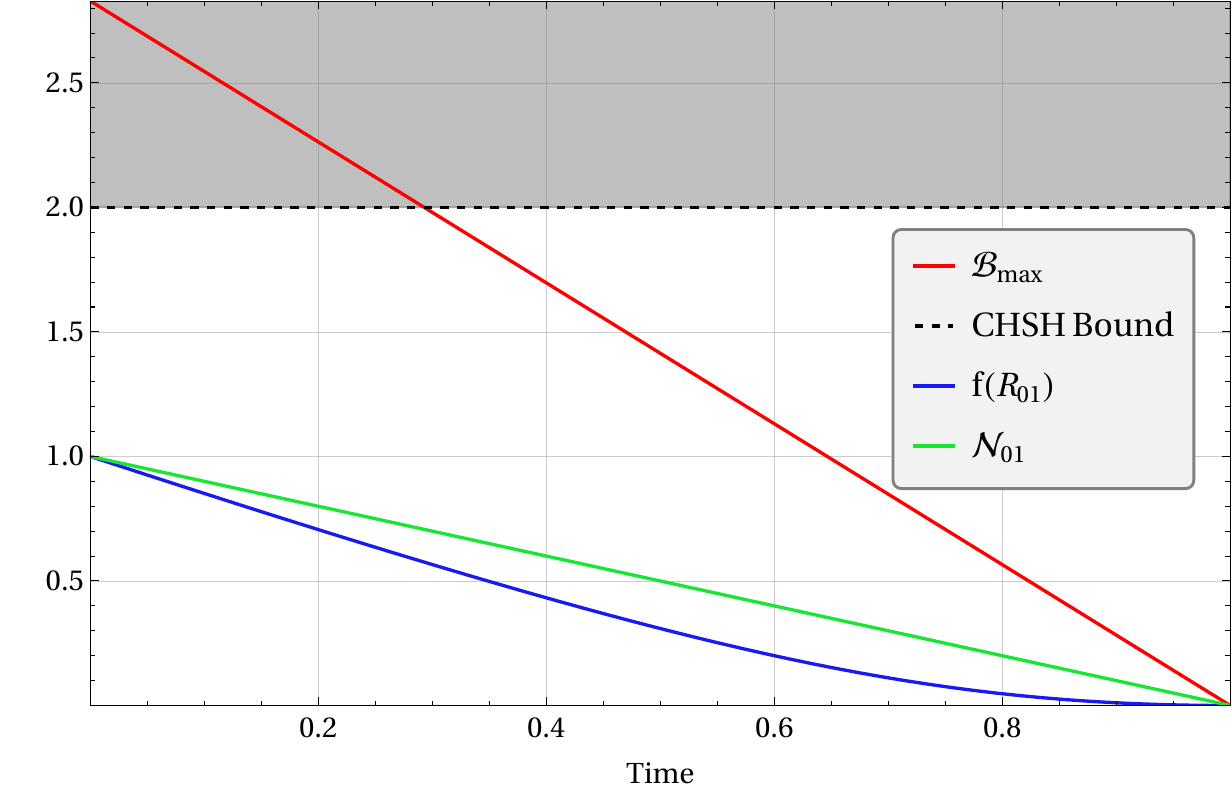}
    \caption{$\mathcal{B}_{\text{max}}$, temporal CHSH bound (the gray shaded area represents the region where temporal CHSH is violated), negativity of the PDM $R_{01}$, and the absolute violation of NSIT $\mathcal{N}_{01}$, for the case of initial pure state $\ket{0} \bra{0}$ and depolarizing channel, with measurements chosen with angles $\theta_0 = \pi/2$ and $\theta_1 = \pi$.}
    \label{plot_CHSH_NSIT_example}
\end{figure}

\end{document}